\documentclass[sigplan,10pt]{acmart}
\acmConference[PL'18]{ACM SIGPLAN Conference on Programming Languages}{January 01--03, 2018}{New York, NY, USA}
\acmYear{2018}
\acmISBN{} 
\acmDOI{} 
\startPage{1}

\setcopyright{none}

\usepackage{booktabs}   
\usepackage{subcaption} 

\usepackage{stmaryrd} 
\usepackage{proof}
\usepackage{amssymb}
\usepackage{amsmath}
\usepackage{mathtools}
\usepackage{amsthm}
\usepackage{mathpartir}
\usepackage{color}
\usepackage{xstring}
\usepackage{ntabbing}
\usepackage{listings}
\usepackage{varwidth}
\usepackage{enumitem}
\usepackage{multicol}
\usepackage{lipsum}

\newlength{\rWidth}

\newcommand{\m}[1]{\mathsf{#1}}
\newcommand{\mb}[1]{\mathbf{#1}}

\newenvironment{sill}{\begin{tabbing}}{\end{tabbing}}

\newcommand{\W}{\Omega}
\newcommand{\Sg}{\Sigma}
\newcommand{\xvdash}[1]{%
  \vdash^{\mkern-8mu\scriptstyle\rule[-.9ex]{0pt}{0pt}#1}%
}

\newcommand{\potconf}[1]{\overset{#1}{\vDash}}

\newcommand{\D}{\Delta}
\renewcommand{\G}{\Gamma}

\newcommand{\proves}{\vDash}

\newcommand{\proc}[2]{\m{proc}(#1, #2)}
\newcommand{\msg}[2]{\m{msg}(#1, #2)}

\newcommand{\step}{\mapsto}

\newcommand{\fresh}[1]{(#1 \text{ fresh})}

\newcommand{\ecase}[3]{\m{case} \; #1 \; (#2 \Rightarrow #3)}

\newcommand{\erecvch}[2]{#2 \leftarrow \m{recv} \; #1}
\newcommand{\erecvn}[2]{\{#2\} \leftarrow \m{recv} \; #1}

\newcommand{\esendch}[2]{\m{send} \; #1 \; #2}
\newcommand{\esendn}[2]{\m{send} \; #1 \; \{#2\}}

\newcommand{\ewait}[1]{\m{wait} \; #1}

\newcommand{\eclose}[1]{\m{close} \; #1}

\newcommand{\fwd}[2]{#1 \leftarrow #2}

\newcommand{\esendl}[2]{#1.#2}

\newcommand{\ecut}[4]{#1 \leftarrow #2 \leftarrow #3 \semi #4}

\newcommand{\ework}[1]{\m{work} \, \{#1\}}
\newcommand{\eget}[2]{\m{get} \, #1 \, \{#2\}}
\newcommand{\epay}[2]{\m{pay} \, #1 \, \{#2\}}
\newcommand{\procdef}[3]{#3 \leftarrow #1 \leftarrow #2}
\newcommand{\procdefna}[2]{#2 \leftarrow #1}

\newcommand{\eassume}[2]{\m{assume} \; #1 \; \{#2\}}
\newcommand{\eassert}[2]{\m{assert} \; #1 \; \{#2\}}

\newcommand{\eimposs}{\m{impossible}}

\newcommand{\lolli}{\multimap}
\newcommand{\tensor}{\otimes}
\newcommand{\with}{\mathbin{\binampersand}}

\newcommand{\one}{\mathbf{1}}

\newcommand{\semi}{\; ; \;}
\newcommand{\ichoiceop}{\oplus}
\newcommand{\echoiceop}{\with}
\newcommand{\ichoice}[1]{\ichoiceop \{ #1 \}}
\newcommand{\echoice}[1]{\echoiceop \{ #1 \}}

\newcommand{\mi}[1]{\mbox{\it #1}}

\newcommand{\tassertop}{{?}}  
\newcommand{\tassumeop}{{!}}  
\newcommand{\tassert}[1]{\tassertop\{#1\}. \,} 
\newcommand{\tassume}[1]{\tassumeop\{#1\}. \,} 
\newcommand{\true}{\top}

\newcommand{\entailpot}[1]{\xvdash{#1}}
\newcommand{\ivdash}[1]{\; _i{\entailpot{#1}}\;}

\newcommand{\paypot}{\triangleright}
\newcommand{\getpot}{\triangleleft}
\newcommand{\tgetpot}[2]{\getpot^{#2} #1}
\newcommand{\tpaypot}[2]{\paypot^{#2} #1}

\newcommand{\tforall}[1]{\forall #1. \, }
\newcommand{\texists}[1]{\exists #1. \, }

\newcommand{\Next}{\raisebox{0.3ex}{$\scriptstyle\bigcirc$}}
\renewcommand{\next}[1]{\Next #1}
\newcommand{\tdelay}[2]{
    \IfEqCase{#2}{%
        {1}{\next{#1}}%
    }[{\Next^{#2} (#1)}]%
}%

\newcommand{\indv}[1]{\overline{[#1]}}

\newcommand{\config}{\mathcal{S}}

\newcommand{\queue}[1]{\m{queue}_{#1}}
\newcommand{\qu}[1]{\m{qu}_{#1}}

\newcommand{\B}[1]{\textcolor{blue}{#1}}

\newcommand{\imposs}{\m{impossible}}

\newcommand{\cons}{\mathcal{C}}
\newcommand{\vars}{\mathcal{V}}

\newcommand{\unfldsg}[2]{\m{unfold}_{#1}(#2)}
\newcommand{\unfold}[1]{\unfldsg{\Sg}{#1}}
\newcommand{\rel}{\mathcal{R}}
\newcommand{\concs}[2]{\mathtt{conc}(#1) = #2}

\newtheorem{theorem}{Theorem}
\newtheorem{definition}{Definition}
\newtheorem{lemma}{Lemma}

\definecolor{verylightgray}{rgb}{.97,.97,.97}

\lstdefinelanguage{Solidity}{
	keywords=[1]{anonymous, assembly, assert, balance, break, call, callcode, case, catch, class, constant, continue, constructor, contract, debugger, default, delegatecall, delete, do, else, emit, event, experimental, export, external, false, finally, for, function, gas, if, implements, import, in, indexed, instanceof, interface, internal, is, length, library, log0, log1, log2, log3, log4, memory, modifier, new, payable, pragma, private, protected, public, pure, push, require, return, returns, revert, selfdestruct, send, solidity, storage, struct, suicide, super, switch, then, this, throw, transfer, true, try, typeof, using, value, view, while, with, addmod, ecrecover, keccak256, mulmod, ripemd160, sha256, sha3}, 
	keywordstyle=[1]\color{blue}\bfseries,
	keywords=[2]{address, bool, byte, bytes, bytes1, bytes2, bytes3, bytes4, bytes5, bytes6, bytes7, bytes8, bytes9, bytes10, bytes11, bytes12, bytes13, bytes14, bytes15, bytes16, bytes17, bytes18, bytes19, bytes20, bytes21, bytes22, bytes23, bytes24, bytes25, bytes26, bytes27, bytes28, bytes29, bytes30, bytes31, bytes32, enum, int, int8, int16, int24, int32, int40, int48, int56, int64, int72, int80, int88, int96, int104, int112, int120, int128, int136, int144, int152, int160, int168, int176, int184, int192, int200, int208, int216, int224, int232, int240, int248, int256, mapping, string, uint, uint8, uint16, uint24, uint32, uint40, uint48, uint56, uint64, uint72, uint80, uint88, uint96, uint104, uint112, uint120, uint128, uint136, uint144, uint152, uint160, uint168, uint176, uint184, uint192, uint200, uint208, uint216, uint224, uint232, uint240, uint248, uint256, var, void, ether, finney, szabo, wei, days, hours, minutes, seconds, weeks, years},	
	keywordstyle=[2]\color{teal}\bfseries,
	keywords=[3]{block, blockhash, coinbase, difficulty, gaslimit, number, timestamp, msg, data, gas, sender, sig, value, now, tx, gasprice, origin},	
	keywordstyle=[3]\color{violet}\bfseries,
	identifierstyle=\color{black},
	sensitive=false,
	comment=[l]{//},
	morecomment=[s]{/*}{*/},
	commentstyle=\color{gray}\ttfamily,
	stringstyle=\color{red}\ttfamily,
	morestring=[b]',
	morestring=[b]"
}

\newcommand{\focus}[1]{[#1]}
\newcommand{\phiL}[1]{[#1]^{?}_L}
\newcommand{\phiR}[1]{[#1]^{!}_R}
\newcommand{\tpL}[1]{[#1]^{\tau}_L}
\newcommand{\tpR}[1]{[#1]^{\tau}_R}

\newcommand{\ins}{\iota}
\newcommand{\tcm}{\mathcal{M}}
\newcommand{\inc}[1]{\m{inc}(#1)}
\newcommand{\dec}[1]{\m{dec}(#1)}
\newcommand{\goto}{\m{goto}}
\newcommand{\zeroc}[1]{\m{zero}(#1) ?}
\newcommand{\halt}{\m{halt}}

\begin{document}

\title[Session Types with Arithmetic Refinements]{Session Types with Arithmetic Refinements
and Their Application to Work Analysis}        


\author{Ankush Das}
\affiliation{
  \institution{Carnegie Mellon University}            
  \country{USA}                    
}
\email{ankushd@cs.cmu.edu}          

\author{Frank Pfenning}
\affiliation{
  \institution{Carnegie Mellon University}            
  \country{USA}                    
}
\email{fp@cs.cmu.edu}          

\lstset{basicstyle=\ttfamily\footnotesize, columns=fullflexible}

\begin{abstract}
  Session types statically prescribe bidirectional communication
  protocols for message-passing processes and are in a Curry-Howard
  correspondence with linear logic propositions.  However, simple
  session types cannot specify properties beyond the type of exchanged
  messages. In this paper we extend the type system by using index
  refinements from linear arithmetic capturing intrinsic
  attributes of data structures and algorithms so that we can express
  and verify amortized cost of programs using ergometric types.
  We show that, despite the
  decidability of Presburger arithmetic, type equality and therefore
  also type checking are now undecidable, which stands in contrast to
  analogous dependent refinement type systems from functional
  languages.  We also present a practical
  incomplete algorithm for type equality and an algorithm for type
  checking which is complete relative to an oracle for type equality.
  Process expressions in this explicit language are rather verbose, so
  we also introduce an implicit form and a sound and complete
  algorithm for reconstructing explicit programs, borrowing ideas from
  the proof-theoretic technique of focusing.
  We conclude by illustrating our systems and algorithms with a variety of examples
  that have been verified in our implementation.
\end{abstract}


\maketitle

\section{Introduction}\label{sec:intro}


\emph{Session types}~\cite{Honda93CONCUR,Honda98esop,Honda08POPL,Vasconcelos12ST}
provide a structured way of
prescribing communication protocols of message-passing systems.
This paper focuses on \emph{binary session types}
governing the interactions along channels with two endpoints.  Binary
session types without general recursion exhibit a Curry-Howard
isomorphism with linear logic
\cite{Caires2010CONCUR,Wadler12icfp,Caires16mscs} and are therefore of
particular foundational significance.  Moreover, type safety derives
from properties of cut reduction and guarantees \emph{freedom from deadlocks}
(global progress) and \emph{session fidelity} (type preservation)
ensuring that at runtime the sender and receiver exchange messages
conforming to the channel's type.

However, basic session types have limited expressivity. As a
simple example, consider the session type offered by a queue
data structure storing elements of type $A$.
\begin{sill}
$\queue{A} = \echoice{$\=$\mb{ins} : A \lolli \queue{A},$\\
\>$\mb{del} : \ichoice{$\=$\mb{none} : \one,$\\
\>\>$\mb{some} : A \tensor \queue{A}}}$
\end{sill}

This type describes a queue interface supporting insertion and
deletion. The \emph{external choice} operator $\echoiceop$ dictates
that the process providing this data structure accepts either one of
two messages: the labels $\mb{ins}$ or $\mb{del}$.  In the case of the
label $\mb{ins}$, it then receives an element of type $A$ denoted by
the $\lolli$ operator, and then the type recurses back to
$\queue{A}$. On receiving a $\mb{del}$ request, the process can
respond with one of two labels ($\mb{none}$ or $\mb{some}$), indicated
by the \emph{internal choice} operator $\ichoiceop$.  It responds with
$\mb{none}$ and then \emph{terminates} (indicated by $\one$) if the
queue is empty, or with $\mb{some}$ followed by the element of type
$A$ (expressed with the $\tensor$ operator) and recurses if the queue
is nonempty.  However, the simple session type does not express the
conditions under which the $\mb{none}$ and $\mb{some}$ branches must
be chosen, which requires tracking the length of the queue.

We propose extending session types with simple arithmetic
refinements to express, for instance, the size of a queue.
The more precise type

\begin{sill}
$\queue{A}[n] = \echoice{$\=$\mb{ins} : A \lolli \queue{A}[n+1],$\\
\>\hspace{-3em}$\mb{del} : \ichoice{$\=$\mb{none} : \tassert{n=0} \one,$\\
\>\>$\mb{some} : \tassert{n > 0} A \tensor \queue{A}[n-1]}}$
\end{sill}

uses the index refinement $n$ to indicate the size of the queue.
In addition, we introduce a \emph{type constraint} $\tassert{\phi} A$
which can be read as ``\textit{there exists a proof of $\phi$}'' and
is analogous to the \emph{assertion} of $\phi$ in imperative
languages.  Here, the process providing the queue must (conceptually)
send a proof of $n = 0$ after it sends $\mb{none}$, and a proof of
$n > 0$ after it sends $\mb{some}$.  It is therefore constrained in
its choice between the two branches based on the value of the index
$n$.  Because the the index domain from which the propositions $\phi$
are drawn is Presburger arithmetic and hence decidable, no proof of
$\phi$ will actually be sent, but we can nevertheless verify the
constraint statically (which is the subject of this paper) or
dynamically (see~\cite{Gommerstadt18esop,Gommerstadt19phd}).  Although
not used in this example, we also add the dual $\tassume{\phi} A$
(\textit{for all proofs of $\phi$}, analogous to the \emph{assumption}
of $\phi$), and explicit quantifiers $\texists{n} A$ and
$\tforall{n} A$ that send and receive natural numbers, respectively.

Such arithmetic refinements are instrumental in expressing
\emph{sequential} and \emph{parallel complexity bounds}. Prior work on
ergometric~\cite{Das18RAST,Das19Nomos} and temporal session types
\cite{Das18Temporal} rely on index refinements to express the size of
lists, stacks and queue data structures, or the height of trees and
express work and time bounds as a function of these indices.  However,
they do not explore the metatheory or implementation of these
arithmetic refinements or how they integrate with time and work
analysis.

Of course, arithmetic type refinements are not new and have been
explored extensively in functional languages, for example, by
Zenger~\cite{Zenger97}, in DML~\cite{Xi99popl}, or in the form of
\emph{Liquid Types}~\cite{Rondon08pldi}.  Variants have been adapted
to session types as
well~\cite{Griffith13LiqPi,Gommerstadt18esop,Zhou19Fluid},
generally with the implicit assumption that index refinements are
somehow ``orthogonal'' to session types.  In this paper we show that,
upon closer examination, this is \emph{not} the case.  In particular,
unlike in the functional setting, session type equality and therefore
type checking become undecidable.  Remarkably, this is the case
whether we treat session types equirecursively~\cite{Gay2005} or
isorecursively~\cite{Lindley16icfp}, and even in the quantifier-free
fragment.  In response, we develop a new algorithm for type equality
which, though incomplete, easily handles the wide variety of example
programs we have tried (see Appendix~\ref{app:examples}).
Moreover, it is naturally extensible through the additional assertion
of type invariants should the need arise.

With a practically effective type equality algorithm in hand, we then
turn our attention to \emph{type checking}.  It turns out that
assuming an oracle for type equality, type checking is decidable
because it can be reduced to checking the validity of propositions in
Presburger arithmetic.  We define \emph{type checking} over a language
where constructs related to arithmetic constraints ($\texists{n} A$,
$\tforall{n} A$, $\tassert{\phi} A$, and $\tassume{\phi} A$) have
explicit communication counterparts.  Similarly, type constructs for
receiving or sending \emph{potential} ($\getpot^{p}A$ and
$\paypot^{p} A$) for amortized work analysis have corresponding
process constructs to (conceptually) send and receive potential.
Revisiting the queue example, the type

\begin{sill}
  $\queue{A}[n] = \echoice{$\=$\mb{ins} : \textcolor{red}{\getpot^{2n}}
  (A \lolli \queue{A}[n+1]),$\\
  \>\hspace{-3em}$\mb{del} : \textcolor{red}{\getpot^{2}} \ichoice{$\=$\mb{none} :
  \tassert{n=0} \one,$\\
  \>\>$\mb{some} : \tassert{n > 0} A \tensor \queue{A}[n-1]}}$
\end{sill}

expresses that the client has to send $2n$ units of potential to
enqueue an element, and $2$ units of potential to dequeue.
Despite the high theoretical complexity of deciding Presburger
arithmetic, all our examples check very quickly using Cooper's
decision procedure~\cite{cooper1972theorem} with two
optimizations.

Many programs in this explicit language are unnecessarily verbose and
therefore tedious for the programmer to write, because the process
constructs pertaining to the refinement layer contribute only to
verifying its properties, but not its observable computational
outcomes.  As is common for refinement types, we therefore also
designed an \emph{implicit} language for processes where most
constructs related to index refinements and amortized work analysis
are omitted.  The problem of \emph{reconstruction} is then to map such
an implicit program to an explicit one which is sound (the result
type-checks) and complete (if there is a reconstruction, it can be
found).  Interestingly, the nature of Presburger arithmetic makes full
reconstruction impossible.  For example, the proposition
$\forall n.\, \exists k.\, (n = 2k \lor n = 2k+1)$ is true but the
witness for $k$ as a Skolem function of $n$ (namely
$\lfloor n/2\rfloor$) cannot be expressed in Presburger arithmetic.
Since witnesses are critical if we want to understand
the work performed by a computation, we require that type quantifiers
$\tforall{n} A$ and $\texists{n} A$ have explicit witnesses in
processes.  We provide a sound and complete algorithm for the
resulting reconstruction problem.  This algorithm exploits
proof-theoretic properties of the sequent calculus akin to
focusing~\cite{Andreoli92} to avoid backtracking and consequently
provides precise error messages that we have found to be
helpful.

We have implemented our language in SML, where a programmer can
choose explicit or implicit syntax and the exact cost model for
work analysis.  The implementation consists of a lexer, parser, type
checker, and reconstruction engine, with particular attention to
providing precise error messages.  To the best of our knowledge, this
is the first implementation of ergometric session types with
arithmetic refinements.

To summarize, we make the following contributions:
\begin{enumerate}[leftmargin=*]
  \item Design and implementation of a session-typed language with arithmetic
    refinements and ergometric types
  \item Proof of undecidability of type equality for the small quantifier-free fragment
    of this language
  \item A new type equality algorithm that works well in practice
  \item A type checking algorithm that is sound and complete relative to type equality
  \item A sound and complete reconstruction algorithm for a process
    language where most index and ergometric constructs remain implicit
\end{enumerate}

\section{Overview of Refinement Session Types}\label{sec:overview}

The underlying base system of session types is derived from a Curry-Howard
interpretation~\cite{Caires2010CONCUR,Caires16mscs} of intuitionistic linear logic
\cite{Girard87tapsoft}. The key idea is that an intuitionistic linear sequent
\begin{center}
  \begin{minipage}{0cm}
  \begin{tabbing}
  $A_1, A_2, \ldots, A_n \vdash A$
  \end{tabbing}
  \end{minipage}
\end{center}
is interpreted as the interface to a process expression $P$. We label each of the
antecedents with a channel name $x_i$ and the succedent with channel name $z$.
The $x_i$'s are \emph{channels used by} $P$ and $z$ is the \emph{channel provided by} $P$.
\begin{center}
  \begin{minipage}{0cm}
  \begin{tabbing}
  $x_1 : A_1, x_2 : A_2, \ldots, x_n : A_n \vdash P :: (z : C)$
  \end{tabbing}
  \end{minipage}
\end{center}
The resulting judgment formally states that process $P$ provides a service of
session type $C$ along channel $z$, while using the services of session types $A_1,
\ldots,A_n$ provided along channels $x_1, \ldots, x_n$ respectively. All these
channels must be distinct. We abbreviate the antecedent of the sequent by $\Delta$.


We introduce a new type operator $\tassert{\phi} A$ allowing processes
to exchange constraints on these index refinements. The provider of
$\tassert{\phi} A$ sends a proof of $\phi$, which is
received by the client. The dual of this operator is
$\tassume{\phi} A$ where the provider receives a proof of $\phi$ sent
by the client.
Because intrinsic properties of data structures (such as
the number of elements) and also the work performed by a computation
must be nonnegative we work over the natural numbers $0, 1, \ldots$
rather than general integers.  This includes a static validity check
for types to ensure that all index refinements are nonnegative.  For
example, while checking the validity of $\queue{A}[n]$, we encounter
the constraint $n > 0$ in the $\mb{some}$ branch, so we assume it and
then verify that $n - 1 \geq 0$, ensuring the validity of
$\queue{A}[n-1]$.

\begin{figure}[t]
  \includegraphics[width=\linewidth]{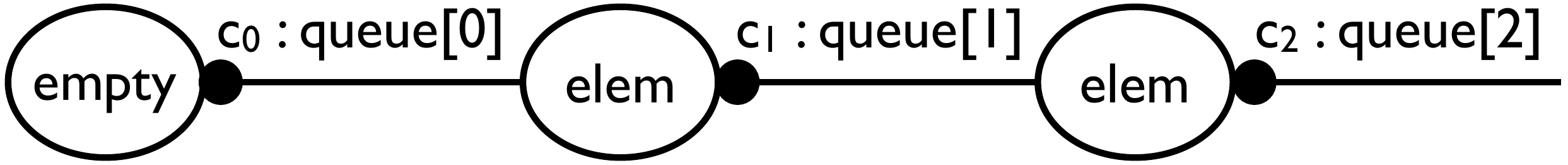}
  
  \caption{Implementation of queue data structure}
  \label{fig:queue}
  
\end{figure}

One parallel implementation of such a queue data structure
(Figure~\ref{fig:queue}) is a sequence of $\mi{elem}$ processes, each
storing an element of the queue, terminated by an $\mi{empty}$
process, representing the empty queue. The $\mi{empty}$ process offers
along $c_0 : \queue{A}[0]$ (indicated by $\bullet$ between
$\mi{empty}$ and $c_0$), while each $\mi{elem}[n]$ process uses
$\queue{A}[n]$, and an element of type $A$ (not shown) and offers on
$\queue{A}[n+1]$. In our notation, the process declarations will be
written as
\[
  \begin{array}{rl}
  \cdot & \vdash empty :: (s : \queue{A}[0]) \\
  (x : A), (t : \queue{A}[n]) & \vdash elem[n] :: (s : \queue{A}[n+1])
  \end{array}
\]
Figure~\ref{fig:queue} shows a queue with 2 elements.
The channel $c_2$ is the only one exposed to the client (not shown)
and receives $\mb{ins}$ and $\mb{del}$ requests.

\begin{figure}[t]
  \vspace{-0.5em}
  \begin{ntabbing}
  \reset
  $\cdot \vdash \mi{empty} :: (s : \queue{A}[0])$ \label{proc:emp_type}\\
  \quad \= $\procdefna{\mi{empty}}{s} =$ \label{proc:emp_def} \\
  \> \; $\m{case}\; s \; ($ \label{proc:emp_case}\\
  \> \quad \= \; $\mb{ins} \Rightarrow$ \= $\erecvch{s}{x} \semi$  \label{proc:emp_recv}
  \hspace{3.4em} \footnotesize{$\%\; (x : A) \vdash (s : \qu{A}[1])$} \\
  \>\>\> $\procdefna{\mi{empty}}{e} \semi$ \label{proc:emp_emp}
  \hspace{-0.9em} \footnotesize{$\%\; (x : A), (e : \qu{A}[0]) \vdash (s : \qu{A}[1])$} \\
  \>\>\> $\procdef{\mi{elem}[0]}{x, e}{s}$ \label{proc:emp_elem} \\
  \>\> $ \mid \mb{del} \Rightarrow$ \> $\esendl{s}{\mb{none}} \semi$
  \label{proc:emp_none}
  \hspace{5.6em} \footnotesize{$\% \; \cdot \vdash (s : \tassert{0 = 0} \one)$}  \\
  \>\>\> $\B{\eassert{s}{0 = 0}} \semi$ \label{proc:emp_assert}
  \hspace{5.4em} \footnotesize{$\% \; \cdot \vdash (s : \one)$} \\
  \>\>\> $\eclose{s})$ \label{proc:emp_close}
  \end{ntabbing}
  
  \begin{ntabbing}
  $(x : A), (t : \queue{A}[n]) \vdash \mi{elem}[n] :: (s : \queue{A}[n+1])$ \label{proc:elem_type}\\
  \quad \= $\procdef{\mi{elem}[n]}{x, t}{s} =$ \label{proc:elem_def} \\
  \> \; $\m{case}\; s \; ($ \label{proc:elem_case}\\
  \> \quad \= \; $\mb{ins} \Rightarrow$ \= $\erecvch{s}{y} \semi$  \label{proc:elem_recv} \\
  \>\>\> $\esendl{t}{\mb{ins}} \semi$ \label{proc:elem_ins} \\
  \>\>\> $\esendch{t}{y} \semi$ \label{proc:elem_y} \\
  \>\>\>\hspace{2.3em} \footnotesize{$\%\; (x : A), (t : \qu{A}[n+1]) \vdash (s : \qu{A}[n+2])$} \\
  \>\>\> $\procdef{\mi{elem}[n+1]}{x, t}{s}$ \label{proc:elem_elem} \\
  \>\> $ \mid \mb{del} \Rightarrow$ \> $\esendl{s}{\mb{some}} \semi$ \label{proc:elem_some} \\
  \>\>\> $\B{\eassert{s}{n+1 > 0}} \semi$ \label{proc:elem_assert} \\
  \>\>\> $\esendch{s}{x} \semi$ \label{proc:elem_x}
  \hspace{2.4em} \footnotesize{$\%\; (t : \qu{A}[n]) \vdash (s : \qu{A}[n])$} \\
  \>\>\> $\fwd{s}{t} )$ \label{proc:elem_fwd}
  \end{ntabbing}
  
  \vspace{-1.25em}
  \caption{Implementations for the $\mi{empty}$ and $\mi{elem}$ processes.}
  \vspace{-1.25em}
  \label{fig:queue_impl}
\end{figure}

Figure~\ref{fig:queue_impl} describes the implementation of $\mi{empty}$
and $\mi{elem}$ processes along with their type derivations on the right
or line below (type $\queue{A}[n]$ abbreviated to $\qu{A}[n]$). Upon
receiving the $\mb{ins}$ label and element $x : A$ (line~\ref{proc:emp_recv}),
the $\mi{empty}$ process spawns a new $\mi{empty}$ process
(line~\ref{proc:emp_emp}), binds it to channel $e$, and tail calls
$\mi{elem}[0]$ (line~\ref{proc:emp_elem}). On inputting the $\mb{del}$ label,
the $\mi{empty}$
process takes the $\mb{none}$ branch (line~\ref{proc:emp_none})
since it stores no elements. Therefore, it needs to send a proof of $n = 0$,
and since it offers $\queue{A}[0]$, it sends the trivial proof of
$0 = 0$ (line~\ref{proc:emp_assert}), and closes the channel terminating
communication (line~\ref{proc:emp_close}). The $\mi{elem}$ process
receives the $\mb{ins}$ label and element $y : A$ (line~\ref{proc:elem_recv}),
passes on these two messages on the tail $t$
(lines~\ref{proc:elem_ins},\ref{proc:elem_y}), and recurses with
$\mi{elem}[n+1]$ (line~\ref{proc:elem_elem}). The type expected by
$\mi{elem}[n+1]$ indeed matches the type of the input and output channels,
as confirmed by the process declaration. On receiving the $\mb{del}$ label,
the $\mi{elem}$ process replies with the $\mb{some}$ label (line~\ref{proc:elem_some})
and the proof of $n+1 > 0$ (line~\ref{proc:elem_assert}), again trivial
since $n$ is a natural number. It terminates with forwarding $s$ along $t$
(line~\ref{proc:elem_fwd}). This forwarding is valid since the types of
$s$ and $t$ exactly match as expected by the $\m{id}$ rule in
Figure~\ref{fig:basic-typing}. The programmer is not burdened with
writing the asserts (in blue) as they are automatically inserted by our
reconstruction algorithm.

At runtime,
each arithmetic proposition will be \emph{closed}, so if it has no
quantifiers it can simply be evaluated.  In the presence of
quantifiers, a decision procedure for Presburger arithmetic can be
applied dynamically (if desired, or if a provider or client may not be
trusted), but no actual proof object needs to be transmitted.

An interesting corner case would be, say, if a process with one
element ($n = 1$) responded with $\mb{none}$ to the $\mb{del}$
request.  It would have to follow up with a proof that $1 = 0$, which
is of course impossible.  Therefore, no further communication along
this channel could take place.

\section{Basic and Refined Session Types}\label{sec:formal}

\begin{table*}[t]
  \begin{tabular}{l l l l l}
  \textbf{Type} & \textbf{Continuation} & \textbf{Process Term} & \textbf{Continuation} & \multicolumn{1}{c}{\textbf{Description}} \\
  \toprule
  $c : \ichoice{\ell : A_\ell}_{\ell \in L}$ & $c : A_k$ & $\esendl{c}{k} \semi P$
  & $P$ & provider sends label $k$ along $c$ \\
  & & $\ecase{c}{\ell}{Q_\ell}_{\ell \in L}$ & $Q_k$ & client receives label $k$ along $c$ \\
  \addlinespace
  $c : \echoice{\ell : A_\ell}_{\ell \in L}$ & $c : A_k$ & $\ecase{c}{\ell}{P_\ell}_{\ell \in L}$
  & $P_k$ & provider receives label $k$ along $c$ \\
  & & $\esendl{c}{k} \semi Q$ & $Q$ & client sends label $k$ along $c$ \\
  \addlinespace
  $c : A \tensor B$ & $c : B$ & $\esendch{c}{w} \semi P$
  & $P$ & provider sends channel $w : A$ along $c$ \\
  & & $\erecvch{c}{y} \semi Q_y$ & $Q_y[w/y]$ & client receives channel $w : A$ along $c$ \\
  \addlinespace
  $c : A \lolli B$ & $c : B$ & $\erecvch{c}{y} \semi P_y$
  & $P_y[w/y]$ & provider receives channel $w : A$ along $c$ \\
  & & $\esendch{c}{w} \semi Q$ & $Q$ & client sends channel $w : A$ along $c$ \\
  \addlinespace
  $c : \one$ & --- & $\eclose{c}$
  & --- & provider sends $\mi{close}$ along $c$ \\
  & & $\ewait{c} \semi Q$ & $Q$ & client receives $\mi{close}$ along $c$ \\
  \bottomrule
  \end{tabular}
  \caption{Basic session types with operational description}
  \label{tab:language}
  \vspace{-2em}
  \end{table*}

This section presents the basic system of session types and its
arithmetic refinement, postponing ergometric types to
Section~\ref{sec:ergo}.  In addition to the type constructors arising
from the connectives of intuitionistic linear logic ($\oplus$,
$\with$, $\tensor$, $\one$ $\lolli$), we have type names, indexed by a
sequence of arithmetic expressions $V \indv{e}$, existential and
universal quantification over natural numbers ($\texists{n} A$,
$\tforall{n} A$) and existential and universal constraints
($\tassert{\phi} A$, $\tassume{\phi} A$).  We write $i$ for constants.

\[
  \begin{array}{lrcl}
    \mbox{Types} & A & ::= & \ichoice{\ell : A}_{\ell \in L}
    \mid \echoice{\ell : A}_{\ell \in L} \\
                 & & \mid & A \tensor A \mid A \lolli A \mid \one \mid V \indv{e} \\
                 & & \mid & \tassert{\phi} A \mid \tassume{\phi} A
                            \mid \texists{n} A \mid \tforall{n} A \\
    \mbox{Arith. Exps.} & e & ::= & i \mid e + e \mid e - e \mid i \times e \mid n \\
    \mbox{Arith. Props.} &
    \phi & ::= & e = e \mid e > e \mid \top \mid \bot
                 \mid \phi \land \phi \\
    & & \mid &  \phi \lor \phi \mid \lnot \phi \mid \texists{n}\phi \mid \tforall{n} \phi
  \end{array}
\]
Our implementation does not support type polymorphism but it is convenient in
some of the examples.  We therefore allow definitions such as
$\queue{A}[n] = \ldots$ and interpret them as a family of definitions,
one for each possible type $A$.

We review a few basic session type operators before introducing the
quantified type constructors. Table~\ref{tab:language} overviews
the session types, their associated process terms and operational
description. Figure~\ref{fig:basic-typing} describes selected
typing rules (ignore the premises and annotation on the turnstile
marked in blue, introduced in Section~\ref{sec:ergo}) leaving the complete
set of rules to Appendix~\ref{app:formal}.

The typing judgment has the form of a sequent
\begin{center}
  \begin{minipage}{0cm}
  \begin{tabbing}
  $\vars \semi \cons \semi \D \entailpot{q}_\Sg P :: (x : A)$
  \end{tabbing}
  \end{minipage}
\end{center}
where $\vars$ are index variables $n$, $\cons$ are constraints over
these variables expressed as a single proposition,
$\D$ are the linear antecedents $x_i : A_i$, $P$ is a process
expression, and $x : A$ is the linear succedent.  The \emph{potential}
$q$ is explained in Section~\ref{sec:ergo}.  We propose and maintain
that the $x_i$ and $x$ are all distinct, and that all free index
variables in $\cons$, $\D$, $P$, and $A$ are contained among $\vars$.
Finally, $\Sigma$ is a fixed signature containing type and process
definitions (explained in Section~\ref{subsec:base}) Because it is
fixed, we elide it from the presentation of the rules.  In addition we
write $\vars \semi \cons \proves \phi$ for semantic entailment
(proving $\phi$ assuming $\cons$) in the
constraint domain where $\vars$ contains all arithmetic variables in
$\cons$ and $\phi$.

We formalize the operational semantics as a system of \emph{multiset rewriting
rules}~\cite{Cervesato09SEM}. We introduce semantic objects $\proc{c}{w, P}$
and $\msg{c}{w, M}$ which mean that process $P$ or message $M$ provide
along channel $c$ and have performed work $w$ 
(as described in Section~\ref{sec:ergo}).
A process configuration is a multiset of such objects, where any two
channels provided are distinct
(formally described at the end of Section~\ref{sec:quantifiers}).

\begin{figure}
  \begin{mathpar}
    \vspace{-0.4em}
    \infer[\m{id}]
      {\vars \semi \cons \semi y : A \entailpot{\B{q}} (\fwd{x}{y}) :: (x : A)}
      {\B{\vars \semi \cons \proves q = 0}}
    %
    \and\vspace{-0.4em}
    \infer[{\with}R]
      {\vars \semi \cons \semi \D \entailpot{\B{q}} \ecase{x}{\ell}{P_\ell}_{\ell \in L} ::
      (x : \echoice{\ell : A_\ell}_{\ell \in L})}
      {(\forall \ell \in L)
       & \vars \semi \cons \semi \D \entailpot{\B{q}} P_\ell :: (x : A_\ell)}
    \and\vspace{-0.4em}
    \infer[{\with}L]
      {\vars \semi \cons \semi \D, (x : \echoice{\ell : A_\ell}_{\ell \in L}) \entailpot{\B{q}}
      (\esendl{x}{k} \semi Q) :: (z : C)}
      {(k \in L) & \vars \semi \cons \semi \D, (x : A_k) \entailpot{\B{q}} Q :: (z : C)}
    
    \and\vspace{-0.4em}
    \infer[{\tensor}R]
      {\vars \semi \cons \semi \D, (y : A) \entailpot{\B{q}} (\esendch{x}{y} \semi P) :: (x : A \tensor B)}
      {\vars \semi \cons \semi \D \entailpot{\B{q}} P :: (x : B)}
    \and\vspace{-0.4em}
    \infer[{\tensor}L]
      {\vars \semi \cons \semi \D, (x : A \tensor B) \entailpot{\B{q}} (\erecvch{x}{y} \semi Q) :: (z : C)}
      {\vars \semi \cons \semi \D, (y : A), (x : B) \entailpot{\B{q}} Q :: (z : C)}
    %
    %
  \end{mathpar}
  \vspace{-2em}
  \caption{Typing Rules for Basic Session Types}
  \vspace{-1.5em}
  \label{fig:basic-typing}
\end{figure}

\subsection{Basic Session Types}\label{subsec:base}

\paragraph{External Choice}
The \emph{external choice} type constructor
$\echoice{\ell : A_{\ell}}_{\ell \in L}$ is an $n$-ary labeled
generalization of the additive conjunction $A \with B$. Operationally,
it requires the provider of
$x : \echoice{\ell : A_{\ell}}_{\ell \in L}$ to branch based on the
label $k \in L$ it receives from the client and continue to provide
type $A_{k}$. The corresponding process term is written as
$\ecase{x}{\ell}{P}_{\ell \in L}$.  Dually, the client must send one
of the labels $k \in L$ using the process term
$(\esendl{x}{k} \semi Q)$ where $Q$ is the continuation.  The typing
for the provider and client are rules ${\with}R$ and ${\with}L$ in
Figure~\ref{fig:basic-typing} respectively.
Communication is asynchronous, so that the client
$\esendl{c}{k} \semi Q$ sends a message $k$ along $c$ and continues as $Q$
without waiting for it to be received. As a technical device to ensure that
consecutive messages on a channel arrive in order, the sender also creates a
fresh continuation channel $c'$ so that the message $k$ is actually represented
as $(\esendl{c}{k} \semi \fwd{c}{c'})$ (read: send $k$ along $c$ and continue along
$c'$). When the message $k$ is received along $c$, we select branch $k$ and
also substitute the continuation channel $c'$ for $c$.
Rules ${\with}S$ and ${\with}C$ below describe the operational behavior of the
provider and client respectively $\fresh{c'}$.

\begin{tabbing}
$({\with}S): \m{proc}(d, w, c.k \semi Q) \;\mapsto\; \m{msg}(c', 0, c.k \semi c' \leftarrow c), $\\
\hspace{12em} $\m{proc}(d, w, Q[c'/c])$ \\
$({\with}C):  \m{proc}(c, w, \m{case}\;c\;(\ell \Rightarrow Q_\ell)_{\ell \in L}),$\\
\hspace{1em}$\m{msg}(c', w', c.k \semi c' \leftarrow c)
\;\mapsto\; \m{proc}(c', w+w', Q_k[c'/c])$
\end{tabbing}


The \emph{internal choice} constructor
$\ichoice{\ell : A_{\ell}}_{\ell \in L}$ is the dual of external
choice requiring the provider to send one of the labels $k \in L$ that
the client must branch on.

\paragraph{Channel Passing}
The \emph{tensor} operator $A \tensor B$ prescribes that the provider of
$x : A \tensor B$
sends a channel $y$ of type $A$ and continues to provide type $B$. The
corresponding process term is $\esendch{x}{y} \semi P$ where $P$ is
the continuation.  Correspondingly, its client must receives a channel
using the term $\erecvch{x}{y} \semi Q$, binding it to variable $y$
and continuing to execute $Q$. The typing rules are
$\tensor R, \tensor L$ in
Figure~\ref{fig:basic-typing}.
The dual operator $A \lolli B$ allows the provider to receive a
channel of type $A$ and continue to provide type $B$.

Finally, the type $\one$ indicates \emph{termination} operationally
denoting that the provider send a \emph{close} message followed by
terminating the communication.
The complete set of static and dynamic semantics
is described in Appendix~\ref{app:formal}.

\paragraph{Forwarding}
A process $\fwd{x}{y}$ identifies the channels $x$ and $y$ so that any
further communication along either $x$ or $y$ will be along the unified
channel. Its typing rule corresponds to the logical rule of identity
($\m{id}$ in Figure~\ref{fig:basic-typing}).  Its computation rules
can be found in Appendix~\ref{app:formal}.



\paragraph{Process Definitions}
Process definitions have the form
$\D \entailpot{\B{q}} f \indv{n} = P :: (x : A)$ where $f$ is the name of the
process and $P$ its definition. In addition, $\overline{n}$ is a
sequence of arithmetic variables that $\D$, $P$ and $A$ can refer to.
All definitions are collected in a fixed global signature $\Sg$. We
require that $\overline{n} \semi \top \semi \D \entailpot{\B{q}} P :: (x : A)$
for every definition, thereby allowing
definitions to be mutually recursive. A new instance of a defined
process $f$ can be spawned with the expression
$\ecut{x}{f \indv{e}}{\overline{y}}{Q}$ where $\overline{y}$ is a
sequence of channels matching the antecedents $\D$ and $\indv{e}$ is a
sequence of arithmetic expression matching the variables
$\indv{n}$. The newly spawned process will use all variables in
$\overline{y}$ and provide $x$ to the continuation $Q$.

\begin{mathpar}
  
  \inferrule*[right=$\m{def}$]
  {\B{\vars \semi \cons \proves r \geq q[\overline{e}/\overline{n}]} \quad
  \overline{y:B} \entailpot{\B{q}} f \indv{n} = P_f :: (x : A) \in \Sg \\
  \D' = \overline{(y:B)}[\overline{e}/\overline{n}] \quad
  \vars {\semi} \cons {\semi} \D, (x {:} A[\overline{e}/\overline{n}]) \entailpot{\B{r-q}} Q :: (z {:} C)}
  {\vars \semi \cons \semi \D, \D' \entailpot{\B{r}} (\ecut{x}{f \indv{e}}{\overline{y}}{Q}) :: (z : C)}
\end{mathpar}
The declaration of $f$ is looked up in the signature $\Sg$, and $\overline{e}$
is substituted for $\overline{n}$ while matching the types in $\D'$ and $\overline{y}$
(second premise). The corresponding semantics rule also performs a similar substitution
$\fresh{a}$
\begin{tabbing}
$(\m{def}C)$ \quad \= $\m{proc}(c, w, \ecut{x}{f \indv{e}}{\overline{d}}{Q}) \; \mapsto \;$ \\
\> $\m{proc}(a, 0, P_f[a/x, \overline{d}/\overline{y}, \overline{e}/\overline{n}]), \;
   \m{proc}(c, w, Q[a/x])$
\end{tabbing}
where $\overline y : B \vdash f \indv{n} = P_f :: (x : A) \in \Sg$.

Sometimes a process invocation is a tail call,
written without a continuation as $\procdef{f \indv{e}}{\overline{y}}{x}$. This is a
short-hand for $\procdef{f \indv{e}}{\overline{y}}{x'} \semi \fwd{x}{x'}$ for a fresh
variable $x'$, that is, we create a fresh channel
and immediately identify it with x.

\paragraph{Type Definitions}
As our queue example already showed, session types can be defined
recursively, departing from a strict Curry-Howard interpretation of
linear logic, analogous to the way pure ML or Haskell depart from
a pure interpretation of intuitionistic logic.  For this purpose we
allow (possibly mutually recursive) type definitions $V \indv{n} = A$
in the signature $\Sg$. Here, $\indv{n}$ denotes a sequence of
arithmetic variables. We also require $A$ to be
\emph{contractive}~\cite{Gay2005} meaning $A$ should not itself be a
type name. Our type definitions are \emph{equirecursive} so we can
silently replace type names $V \indv{e}$ indexed with arithmetic
refinements by $A [\overline{e}/\overline{n}]$ during type checking,
and no explicit rules for recursive types are needed.

All types in a signature must be \emph{valid}, formally written as
$\vars \semi \cons \vdash A\; \mi{valid}$, which requires that all
free arithmetic variables of $\cons$ and $A$ are contained in $\vars$,
and that for each arithmetic expression $e$ in $A$ we can prove
$\vars' \semi \cons' \vdash e : \m{nat}$ for the constraints $\cons'$
known at the occurrence of $e$ (implicitly proving that $e \geq 0$).

\subsection{The Refinement Layer}\label{sec:quantifiers}

We now introduce quantifiers ($\texists{n} A$, $\tforall{n} A$) and
constraints ($\tassert{\phi} A$, $\tassume{\phi} A$).  An overview
of the types, process expressions, and their operational meaning can
be found in Table~\ref{tab:refine}.

\begin{table*}[t]
  \begin{tabular}{l l l l l}
  \textbf{Type} & \textbf{Continuation} & \textbf{Process Term} & \textbf{Continuation} & \multicolumn{1}{c}{\textbf{Description}} \\
  \toprule
  $c : \texists{n} A$ & $c : A[i/n]$ & $\esendn{c}{e} \semi P$
  & $P$ & provider sends the value $i$ of $e$ along $c$ \\
  & & $\erecvn{c}{n} \semi Q$ & $Q[i/n]$ & client receives number $i$ along $c$ \\
  \addlinespace
  $c : \tforall{n} A$ & $c : A[i/n]$ & $\erecvn{c}{n} \semi P$
  & $P[i/n]$ & provider receives number $i$ along $c$ \\
  & & $\esendn{c}{e} \semi Q$ & $Q$ & client sends value $i$ of $e$ along $c$ \\
  \addlinespace
  $c : \tassert{\phi} A$ & $c : A$ & $\eassert{c}{\phi} \semi P$
  & $P$ & provider asserts $\phi$ on channel $c$ \\
  & & $\eassume{c}{\phi} \semi Q$ & $Q$ & client assumes $\phi$ on $c$ \\
  \addlinespace
  $c : \tassume{\phi} A$ & $c : A$ & $\eassume{c}{\phi} \semi P$
  & $P$ & provider assumes $\phi$ on channel $c$ \\
  & & $\eassert{c}{\phi} \semi Q$ & $Q$ & client asserts $\phi$ on $c$ \\
  \bottomrule
  \end{tabular}
  \caption{Refined session types with operational description}
  \label{tab:refine}
  \vspace{-2em}
  \end{table*}

\paragraph{Quantification}
The provider of $(c : \texists{n} A)$ should send a witness $i$ along
channel $c$ and then continue as $A[i/n]$.  The witness is specified
by an arithmetic expression $e$ which, since it must be closed at
runtime, can be evaluated to a number $i$ (although we do not bother
formally representing this evaluation).  From the typing perspective,
we just need to check that the expression $e$ denotes a natural
number, using only the permitted variables in $\vars$.  This is
represented with the auxiliary judgment $\vars \semi \cons \vdash e : \m{nat}$
(implicitly proving that $e \geq 0$ under constraint $\cons$).
\begin{mathpar}
  
  \infer[\exists R]
  {\vars \semi \cons \semi \D \entailpot{\B{q}} \esendn{x}{e} \semi P :: (x : \texists{n} A)}
  {\vars \semi \cons \vdash e : \m{nat} \and
  \vars \semi \cons \semi \D \entailpot{\B{q}} P :: (x : A[e/n])}
  \and
  \infer[\exists L]
  {\vars \semi \cons \semi \D, (x : \texists{n} A) \entailpot{\B{q}} \erecvn{x}{n} \semi Q_n :: (z : C)}
  {\vars,n \semi \cons \semi \D, (x : A) \entailpot{\B{q}} Q_n :: (z : C) & \fresh{n}}
\end{mathpar}
Operationally, the provider sends the arithmetic expression
that the client receives and appropriately substitutes.
\begin{tabbing}
  $({\exists}S)$ \quad \= $\proc{c}{w, \esendn{c}{e} \semi P} \;\mapsto\;$
  $\m{proc}(c', w,$ \\
  \> \qquad $P[c'/c]), \; \msg{c}{0, \esendn{c}{e} \semi \fwd{c}{c'}}$ \\
  $({\exists}C)$ \> $\msg{c}{w', \esendn{c}{e} \semi \fwd{c}{c'}}, \;$
  $\m{proc}(d, w,$\\
  \> \hspace{-2em} $\erecvn{c}{n} \semi Q) \;\mapsto\;$
  $\proc{d}{w+w', Q[e/n][c'/c]}$
\end{tabbing}

The dual type $\tforall{n} A$ reverses the role of the provider and
client.  The client sends (the value of) an arithmetic expression $e$
which the provider receives and binds to $n$.
\begin{mathpar}
  
  \infer[\forall R]
  {\vars \semi \cons \semi \D \entailpot{\B{q}} \erecvn{x}{n} \semi P_n :: (x : \tforall{n} A)}
  {\vars,n \semi \cons \semi \D \entailpot{\B{q}} P_n :: (x : A)}
  \and
  \infer[\forall L]
  {\vars \semi \cons \semi \D, (x : \tforall{n} A) \entailpot{\B{q}} \esendn{x}{e} \semi Q :: (z : C)}
  {\vars \semi \cons \vdash e : \m{nat} \and
  \vars \semi \D, (x : A[e/n]) \entailpot{\B{q}} Q :: (z : C)}
\end{mathpar}

\paragraph{Constraints}
Refined session types also allow constraints over index variables.  As
we have already seen in the examples, these critically govern
permissible messages.
%
From the message-passing perspective, the provider of
$(c : \tassert{\phi} A)$ should send a proof of $\phi$ along $c$ and
the client should receive such a proof.  However, since the index
domain is decidable and future computation cannot depend on the form
of the proof (what is known in type theory as \emph{proof
  irrelevance}) such messages are not actually exchanged.  Instead, it
is the provider's responsibility to ensure that $\phi$ holds, while the
client is permitted to assume that $\phi$ is true.  Therefore,
and in an analogy with imperative languages, we write $\eassert{c}{\phi} \semi P$
for a process that \emph{asserts} $\phi$ for channel $c$ and continues
with $P$, while $\eassume{c}{\phi} \semi Q$ \emph{assumes} $\phi$ and
continues with $Q$.

Thus, the typing rules for this new type constructor are
\begin{mathpar}
  \infer[{\tassertop}R]
  {\vars \semi \cons \semi \D \entailpot{\B{q}} \eassert{x}{\phi} \semi P :: (x : \tassert{\phi} A)}
  {\vars \semi \cons \proves \phi & \vars \semi \cons \semi \D \entailpot{\B{q}} P :: (x : A)}
\end{mathpar}
\begin{mathpar}
  \infer[{\tassertop}L]
  {\vars \semi \cons \semi \D, (x : \tassert{\phi} A) \entailpot{\B{q}} \eassume{x}{\phi} \semi Q :: (z : C)}
  {\vars \semi \cons \land \phi \semi \D, (x : A) \entailpot{\B{q}} Q :: (z : C)}
\end{mathpar}
Notice how the provider must verify the truth of $\phi$ given the
currently known constraints $\cons$ ($\vars \semi \cons \proves \phi$),
while the client assumes $\phi$ by adding it to $\cons$.
%
In well-typed configurations (which arise from executing well-typed processes,
see Appendix~\ref{app:formal}) the constraint $\phi$ in these
rules will always be closed and true so there is no need to check this
explicitly.

The dual operator $\tassume{\phi} A$ reverses the role of provider and
client. The provider of $x : \tassume{\phi} A$ may assume the truth of
$\phi$, while the client must verify it.  The dual rules are
\begin{mathpar}
  
  \infer[{\tassumeop}R]
  {\vars \semi \cons \semi \D \entailpot{\B{q}} \eassume{x}{\phi} \semi P :: (x : \tassume{\phi} A)}
  {\vars \semi \cons \land \phi \semi \D \entailpot{\B{q}} P :: (x : A)}
  \and
  \infer[{\tassumeop}L]
  {\vars \semi \cons \semi \D, (x : \tassume{\phi} A) \entailpot{\B{q}} \eassert{x}{\phi} \semi Q :: (z : C)}
  {\vars \semi \cons \proves \phi & \vars \semi \cons \semi \D, (x : A) \entailpot{\B{q}} Q :: (z : C)}
\end{mathpar}

The remaining issue is how to type-check a branch that is impossible
due to unsatisfiable constraints.  For example, if a client sends
a $\mb{del}$ request to a provider along $c : \queue{A}[0]$, the
type then becomes
\begin{center}
  \begin{minipage}{0cm}
  \begin{tabbing}
  $c : \ichoice{\mb{none} : \tassert{0 {=} 0} \one,
  \mb{some} : \tassert{0 {>} 0} A \tensor \queue{A}[0{-}1]}$
  \end{tabbing}
  \end{minipage}
\end{center}
The client would have to branch on the label received
and then assume the constraint asserted by the provider
\begin{sill}
  $\m{case}\; c\;$ \= $(\, \mb{none} \Rightarrow \eassume{c}{0 = 0} \semi P_1$ \\
  \> $\mid \mb{some} \Rightarrow \eassume{c}{0 > 0} \semi P_2)$
\end{sill}
but what could we write for $P_2$ in the $\mb{some}$ branch?
Intuitively, computation should never get there because the provider
can not assert $0 > 0$.  Formally, we use the process expression
`$\eimposs$' to indicate that computation can never reach this spot:
\begin{sill}
  $\m{case}\; c\;$ \= $(\, \mb{none} \Rightarrow \eassume{c}{0 = 0} \semi P_1$ \\
  \> $\mid \mb{some} \Rightarrow \eassume{c}{0 > 0} \semi \eimposs)$
\end{sill}
In implicit syntax (see Section~\ref{sec:recon}) we could omit
the $\mb{some}$ branch altogether and it would be reconstructed
in the form shown above.
Abstracting away from this example, the typing rule for impossibility
simply checks that the constraints are indeed unsatisfiable
\begin{mathpar}
  
  \infer[\m{unsat}]
  {\vars \semi \cons \semi \D \entailpot{\B{q}} \eimposs :: (x : A)}
  {\vars \semi \cons \proves \bot}
\end{mathpar}
There is no operational rule for this scenario since in well-typed configurations
the process expression `$\eimposs$' is dead code and can never be reached.

The extension of session types with index refinements is type safe,
expressed using the usual proofs of preservation and progress on
configurations. A well-typed configuration,
represented using the
judgment $\D_1 \Vdash_{\Sg} \config :: \D_2$ denotes a set of
semantic objects $\config$ using channels in $\D_1$ and offering
channels $\D_2$. The proof of preservation proceeds by induction on
the operational semantics and inversion on the configuration and
process typing judgment.

To state progress, we need the notion of a poised
process~\cite{Pfenning15fossacs}.  A process $\proc{c}{w, P}$ is
poised if it is trying to receive a message on $c$. Dually, a message
$\msg{c}{w, M}$ is poised if it is sending along $c$. A configuration
is poised if every message or process in the configuration is
poised. Intuitively, this means that every process in the
configuration is trying to interact with the outside world.

\begin{theorem}[Type Safety]\label{thm:type_safety}
  For a well-typed configuration $\D_1 \Vdash_{\Sg} \config :: \D_2$:
  \begin{enumerate}
    \item[(i)] (Progress) Either $\config$ is poised, or $\config \step
    \config'$. 

    \item[(ii)] (Preservation) If $\config \step \config'$, then
    $\D_1 \Vdash_{\Sg} \config' :: \D_2$
  \end{enumerate}
\end{theorem}

\section{Type Equality}\label{sec:tpeq}

At the core of an algorithm for type checking is type equality.  It is
necessary for the rule of identity (operationally: forwarding) as well
as the channel-passing constructs for types $A \tensor B$ and
$A \lolli B$.
Informally, two types are equal if they permit exactly the same
communication behaviors.  For example, if

\begin{sill}
  $\m{nat} = \ichoice{\mb{zero} : \one, \mb{succ} : \m{nat}}$ \\
  $\m{nat}' = \ichoice{\mb{zero} : \one, \mb{succ} : \ichoice{\mb{zero} : \one, \mb{succ} : \m{nat}'}}$
\end{sill}

we have that $\m{nat} \equiv \m{nat}'$ because each type allows any
sequence of $\mb{succ}$ labels, followed by $\mb{zero}$
followed by $\mi{close}$.  This carries over to
indexed types.  If we define

\begin{sill}
  $\m{nat}[n] = \ichoice{\mb{zero} : \tassert{n = 0} \one, \mb{succ} : \tassert{n > 0} \m{nat}[n-1]}$ \\
  $\m{pos}[n] = \ichoice{\mb{succ} : \tassert{n > 0} \m{nat}[n-1]}$
\end{sill}

then $\m{pos}[n+1] \equiv \m{nat}[n+1]$ for all $n$.

Following seminal work by Gay and Hole~\cite{Gay2005} type equality is
formally captured as a coinductive definition.  Gay and Hole actually
define subtyping ($A \leq B$ if every behavior permitted by $A$ is
also permitted by $B$) and derive type equality from it; we simplify
matters by directly defining equality.  Subtyping follows the
same pattern and presents no additional complications.
Our definitions capture \emph{equirecursive} type equality, but they
can easily be adapted to an \emph{isorecursive} equality with explicit
$\m{fold}$ messages~\cite{Lindley16icfp,Derakhshan19corr}.

\begin{definition}\label{def:unfold}
  We define $\unfold{A}$ as
  \begin{mathpar}
    
    \infer[\m{def}]
    {\unfold{V \indv{e}} = A[\overline{e}/\overline{n}]}
    {V \indv{n} = A \in \Sg}
    \and
    \infer[\m{str}]
    {\unfold{A} = A}
    {A \not= V\indv{e}}
  \end{mathpar}
\end{definition}

Like Gay and Hole, we require type definitions to be contractive, so
the result of unfolding is never a type variable.  Let
$\mi{Type}$ be the set of all closed type expressions (no free
variables).  An interesting point in this definition is that
constraint types with unsatisfiable constraints are related
because neither can communicate any further messages.

\begin{definition}\label{def:rel}
  A relation $\rel \subseteq
  \mi{Type} \times \mi{Type}$ is a type bisimulation if $(A, B) \in
  \rel$ implies (analogous cases omitted):
  \begin{itemize}[leftmargin=*]
    \item If $\unfold{A} = \ichoice{\ell : A_\ell}_{\ell \in L}$, then $\unfold{B} =
    \ichoice{\ell : B_\ell}_{\ell \in L}$ and $(A_\ell, B_\ell) \in \rel$ for
    all $\ell \in L$.

    \item If $\unfold{A} = A_1 \lolli A_2$, then $\unfold{B} =
    B_1 \lolli B_2$ and $(A_1, B_1) \in \rel$ and
    $(A_2, B_2) \in \rel$.


    \item If $\unfold{A} = \tassume{\phi}{A'}$, then $\unfold{B} = \tassume{\psi}{B'}$
    and either $\proves \phi$, $\proves \psi$, and
    $(A', B') \in \rel$, or $\proves \lnot \phi$ and $\proves \lnot \psi$.

    \item If $\unfold{A} = \texists{m} A'$, then $\unfold{B} = \texists{n} B'$
    and for all $i \in \mathbb{N}$, $(A'[i/m], B'[i/n]) \in \rel$.

  \end{itemize}
\end{definition}

\begin{definition}\label{def:tpeq}
  Two types $A$ and $B$ are equal ($A \equiv B$) iff there exists a type
  bisimulation $\rel$ such that $(A, B) \in \rel$.
\end{definition}

This definition only applies to ground types with no free variables. Since
we allow quantifiers and arithmetic constraints, we need to define equality
in the presence of free variables and arbitrary constraints. To this end,
we define the notation $\forall \vars. \; \cons \Rightarrow A \equiv B$
under the presupposition that $A$ and $B$ are valid type under assumption
$\cons$. 
Interestingly, this definition implies that if $\cons$ is unsatisfiable, then
$A \equiv B$ for all valid types $A$ and $B$.

\begin{definition}\label{def:tpeq_vars}
  We define $\forall \vars. \; \cons \Rightarrow A \equiv B$ iff for all
  ground substitutions $\sigma : \vars$ satisfying $\cons$ (that is,
  $\cdot \proves \cons[\sigma]$), we have $A[\sigma] \equiv B[\sigma]$.
\end{definition}

\begin{theorem}\label{thm:undec}
  Checking $\forall \vars. \; \cons \Rightarrow A \equiv B$ is
  undecidable.
\end{theorem}

\begin{proof}
  Given a two counter machine, we construct two types $A$ and $B$ such
  that the computation of the machine is infinite iff $A \equiv B$. Thus, we
  establish a reduction from non-halting problem to type equality.
\end{proof}  

The type system allows us to simulate the two counter
machine. Intuitively, the quantified type constructors allow us to
branch depending on arithmetic constraints. This coupled with
arbitrary recursion in the language establishes
undecidability. Remarkably, the small fragment of our language
containing only type definitions, internal choice ($\oplus$) and
assertions ($\tassert{\phi} A$) constructors is sufficient to prove
undecidability.  Moreover, the proof still applies if we treat types
isorecursively.  Appendix~\ref{app:tpeq} contains the machine
construction and details of the undecidability proof.

\paragraph{Algorithm}
Despite its undecidability, we have designed a coinductive algorithm
for soundly approximating type equality. Like Gay and Hole's
algorithm, it proceeds by attempting to construct a bisimulation and
it can terminate in three different states: (1) we have succeeded in
constructing a bisimulation, (2) we have found a counterexample to
type equality by finding a place where the types may exhibit different
behavior, or (3) we have terminated the search without a definitive
answer.  From the point of view of type-checking, both (2) and (3) are
interpreted as a failure to type-check.  The algorithm is expressed as
a set of inference rule where the execution of the algorithm
corresponds to construction of a deduction.  The algorithm is
deterministic (no backtracking) and the implementation is
quite efficient in practice.

We explain the algorithm with an illustrative example.
\begin{sill}
  $\m{ctr}[x,y] = \ichoice{$\=$\mb{lt} : \tassert{x < y} \m{ctr}[x+1,y],$\\
  \>$\mb{ge} : \tassert{x \ge y} \one}$
\end{sill}
The $\m{ctr}$ type outputs $\mb{lt}$ or $\mb{ge}$ based on
comparing $x$ and $y$ and recurses with $\m{ctr}[x+1,y]$
if $x < y$.
Compare types $\m{ctr}[x,y]$ and
$\m{ctr}[x+1,y+1]$. They will both output $\mb{le}$ exactly $\m{max}(0, y-x)$
number of times terminating
with $\mb{ge}$, and thus are equal according to our definition.
Suppose we wish to prove $\m{ctr}[x,y] \equiv \m{ctr}[x+1,y+1]$
for all $x,y \in \mathbb{N}$.
We use the algorithmic equality judgment $\vars \semi \cons \semi \G
\vdash A \equiv B$ to denote checking given variables $\vars$
satisfying constraint $\cons$, whether types $A$ and $B$ are equal.
Since the algorithm is coinductive, $\G$ stores equality constraints
encountered so far. We initiate the algorithm with an
empty $\G$ and $\cons = \true$, thereby checking
$x,y \semi \true \semi \cdot \vdash \m{ctr}[x,y] \equiv \m{ctr}[x+1,y+1]$.
\begin{mathpar}
  
  \inferrule*[right = $\m{expd}$]
  {V_1 \indv{v_1} = A \in \Sg \and
  V_2 \indv{v_2} = B \in \Sg \\\\
  \gamma = \forall \vars. \; \cons \Rightarrow V_1 \indv{e_1} \equiv V_2 \indv{e_2}\\\\
  \vars \semi \cons \semi \G, \gamma
  \vdash A[\overline{e_1} / \overline{v_1}] \equiv
  B[\overline{e_2}/\overline{v_2}]}
  {\vars \semi \cons \semi \G \vdash
  V_1 \indv{e_1} \equiv V_2 \indv{e_2}}
\end{mathpar}
The $\m{expd}$ rule adds the equality constraint $V_1 \indv{e_1} \equiv V_2 \indv{e_2}$
to $\G$ and expands the two sides by replacing each type name with its definition.
In our example, we add $\forall X,Y. \;
\m{ctr}[X,Y] \equiv \m{ctr}[X+1,Y+1]$ to $\G$ ($\alpha$-renamed to avoid
confusion) and replace the type names with their definition.
Next, we use the $\oplus$ rule below to explore each branch.
\begin{mathpar}
  
  \infer[\ichoiceop]
  {
    \vars \semi \cons \semi \G \vdash \ichoice{\ell : A_\ell}_{\ell \in L}
    \equiv \ichoice{\ell : B_\ell}_{\ell \in L}
  }
  {
    \vars \semi \cons \semi \G \vdash A_\ell \equiv B_\ell \quad (\forall \ell \in L)
  } 
\end{mathpar}
First, we check the $\mb{lt}$ branch: $\tassert{x < y} \m{ctr}[x+1,y] \equiv
\tassert{x+1 < y+1} \m{ctr}[x+2,y+1]$. We use the $\tassertop$ rule
to check equivalence of the two constraints
\begin{mathpar}
  
  \infer[\tassertop]
  {
    \vars \semi \cons \semi \G \vdash \tassert{\phi} A
    \equiv \tassert{\psi} B
  }
  {
    \vars \semi \cons \proves \phi \leftrightarrow \psi \and
    \vars \semi \cons \land \phi \semi \G \vdash A \equiv B
  } 
\end{mathpar}
If the assumption of $\phi$ were to make the constraints
contradictory we would succeed at this point, using the rule
\begin{mathpar}
  
  \infer[\bot]
  {\Sg \semi \vars \semi \cons \semi \G \vdash
  A \equiv B}
  {\vars \semi \cons \proves \bot}
\end{mathpar}
in accordance with Definition~\ref{def:tpeq_vars}.
In our example, since $x < y \leftrightarrow x+1 < y+1$ and $x < y$ is
consistent, we compare $\m{ctr}[x+1,y] \equiv
\m{ctr}[x+2,y+1]$. At this point, we check if this equality is
entailed by one of the stored equality constraints stored.  The
simplest case of such an entailment is witnessed by a substitution,
applied to one of the stored equality constraints. And yes, since we
stored $(\forall X,Y. \; \m{ctr}[X,Y] \equiv \m{ctr}[X+1,Y+1])$, we can
substitute $x+1$ for $X$ and $y$ for $Y$ to satisfy the desired
equality constraint.  This is the coinductive aspect of the algorithm
formalized in the $\m{def}$ rule.
\begin{mathpar}
  
  \inferrule*[right=$\m{def}$]
  {
    \forall \vars'. \; \cons' \Rightarrow V_1 \indv{E_1} \equiv V_2 \indv {E_2} \in \G \\
    \vars \semi \cons \proves \exists \vars'. \; \cons' \land \indv{E_1}
    = \indv{e_1} \land \indv{E_2} = \indv{e_2}
  }
  {
    \vars \semi \cons \semi \G \vdash V_1 \indv{e_1} \equiv V_2 \indv{e_2}
  }
\end{mathpar}
For our example, this reduces to checking the validity of
$\forall x,y. \; \exists X,Y. \; X = x+1 \land Y = y \land X+1=x+2 \land
Y+1=y+1$.
Similarly, for the $\mb{ge}$ branch, we check if $\tassert{x \ge y} \one
\equiv \tassert{x+1 \ge y+1} \one$. Since the two constraints are equivalent
and $\one \equiv \one$, we infer that the types match in both branches.
Thus, we conclude that $\m{ctr}[x,y] \equiv \m{ctr}[x+1,y+1]$.
The rest of the rules of the type equality algorithm are similar to the
ones presented (full rules in Appendix~\ref{app:tpeq}).

The system so far is potentially nonterminating because when
encountering variable/variable equations, we can use $\m{expd}$
indefinitely.  To ensure termination we use two techniques.  The first
is to introduce internal names for every subexpression of type
definitions.  This means the algorithm alternates between comparing
two type names and two type constructors.
The second is to restrict the $\m{expd}$ rule to the case where no
assumption of the form
$\forall \vars'.\, \cons' \Rightarrow V_1\indv{e_1'} \equiv
V_2\indv{e_2'}$ is already present in $\Gamma$.  This means that for
the case where there are no index expressions, the algorithm behaves
exactly like Gay and Hole's and is terminating: we close a branch
when we find the equation $V_1 \equiv V_2$ in $\Gamma$.

We prove that the type equality algorithm is sound with respect to the
declarative equality definition. The soundness is proved by
constructing a type bisimulation from a derivation of the algorithmic
type equality judgment (Appendix~\ref{app:tpeq}).
\begin{theorem}\label{thm:tpeq_sound}
  If $\vars \semi \cons \semi \cdot \vdash A \equiv B$, then
  $\forall \vars. \; \cons \Rightarrow A \equiv B$.
\end{theorem}

\paragraph{Reflexivity}
Because it is not always derivable with the rules so far, we found it
necessary to add one more rule to the algorithm, namely reflexivity
on type names when indices are provably equal.
This is still sound since the reflexive closure of a type
bisimulation is still a type bisimulation.
\begin{mathpar}
  
  \infer[\m{refl}]
  {\vars \semi \cons \semi \G \vdash V\indv{e} = V\indv{e'}}
  {\vars \semi \cons \proves e_1 = e_1' \land \ldots \land e_n = e_n'}
\end{mathpar}

Traditional refinement languages such as DML~\cite{Xi99popl}
only use reflexivity as a criterion for equality of indexed type
names.  However, as exemplified in the $\m{ctr}$ example, our
algorithm extends beyond reflexivity.

\section{Ergometric Session Types}\label{sec:ergo}

An important application of refinement types is
complexity analysis.
%
To describe the resource contracts for inter-process communication,
the type system is enhanced to support amortized resource analysis
\cite{Tarjan85AARA}. The key idea is that \emph{processes store
  potential} and \emph{messages carry potential}. This potential can
either be consumed to perform \emph{work} or exchanged using special
messages. The type system provides the programmer with the flexibility
to specify what constitutes work. Thus, the programmer can choose to
count the resource they are interested in, and the type system
provides the corresponding upper bound.  Our current examples assign
unit cost to message sending operations ($\esendl{c}{k}$,
$\eclose{c}$, $\esendch{c}{d}$) (exempting those for index objects
$\esendn{c}{e}$, $\eassert{c}{\phi}$, and potential $\epay{c}{r}$, see
below) effectively counting the total number of ``real'' messages
exchanged during a computation.

Two dual type constructors $\tpaypot{A}{r}$ and $\tgetpot{A}{r}$
are used to exchange potential. The provider of $x : \tpaypot{A}{r}$
must \emph{pay} $r$ units of potential along $x$ using process
term $(\epay{x}{r} \semi P)$, and continue to provide $A$ by
executing $P$. These $r$ units are deducted from the potential
stored inside the sender. Dually, the client must receive the
$r$ units of potential using the term $(\eget{x}{r} \semi Q)$
and add this to its internal stored potential. Finally, since
processes are allowed to store potential, the typing judgment
is enhanced by adding a natural number on the turnstile denoting
its internal potential.
\begin{center}
  \begin{minipage}{0cm}
  \begin{tabbing}
  $\vars \semi \cons \semi \D \entailpot{q}_{\Sg} P :: (x : A)$
  \end{tabbing}
  \end{minipage}
\end{center}
We allow potential $q$ to refer to index variables in $\vars$.
The typing rules for $\tpaypot{A}{r}$ are
\begin{mathpar}
  
  \infer[{\paypot}R]
  {\vars \semi \cons \semi \D \entailpot{q} \epay{x}{r_1} \semi P :: (x : \tpaypot{A}{r_2})}
  {\vars \semi \cons \proves q \geq r_1 = r_2 &
  \vars \semi \cons \semi \D \entailpot{q-r_1} P :: (x : A)}
  \and
  \infer[{\paypot}L]
  {\vars \semi \cons \semi \D, (x : \tpaypot{A}{r_2}) \entailpot{q} \eget{x}{r_1} \semi Q :: (z : C)}
  {\vars \semi \cons \proves r_1 = r_2 &
  \vars \semi \cons \semi \D, (x : A) \entailpot{q+r_1} Q :: (z : C)}
\end{mathpar}
In both cases,
we check that the exchanged potential in the expression and type matches
($r_1 = r_2$), and while paying, we ensure that the sender has sufficient potential to pay.
Operationally, the provider creates a special message containing the potential
that is received by the client.
\begin{tabbing}
$({\paypot}S)$ \quad \= $\proc{c}{w, \epay{c}{r} \semi P} \;\mapsto\;$ \\
\> $\proc{c'}{w, P[c'/c]}, \; \msg{c}{0, \epay{c}{r} \semi \fwd{c}{c'}}$ \\
$({\paypot}C)$ \> $\msg{c}{w', \epay{c}{r} \semi \fwd{c}{c'}}, \;$
$\m{proc}(d, w,$\\
\> \qquad $\eget{c}{r} \semi Q) \;\mapsto\;\proc{d}{w+w', Q[c'/c]}$
\end{tabbing}
The dual type $\tgetpot{A}{r}$ enables the provider to receive potential
that is sent by its client.
\begin{mathpar}
  
  \infer[{\getpot}R]
  {\vars \semi \cons \semi \D \entailpot{q} \eget{x}{r_1} \semi P :: (x : \tgetpot{A}{r_2})}
  {\vars \semi \cons \proves r_1 = r_2 &
  \vars \semi \cons \semi \D \entailpot{q+r_1} P :: (x : A)}
\end{mathpar}
\begin{mathpar}
  
  \infer[{\getpot}L]
  {\vars \semi \cons \semi \D, (x : \tgetpot{A}{r_2}) \entailpot{q} \epay{x}{r_1} \semi Q :: (z : C)}
  {\vars {\semi} \cons \proves q \geq r_1 {=} r_2 &
  \vars {\semi} \cons {\semi} \D, (x : A) \entailpot{q-r_1} Q :: (z : C)}
\end{mathpar}
The work counter $w$ in the semantic objects $\proc{c}{w, P}$
and $\msg{c}{w, M}$ tracks the work done by the process. We use a
special expression $\ework{r} \semi P$ to increment the work counter.
None of the other rules affect this counter but simply preserve their
total sum. The programmer can insert
these flexibly in the program to count a specific resource. For example,
to count the total number of messages sent, we insert $\ework{1}$
just before sending every message.
\begin{tabbing}
  $({\m{work}})$ \quad $\proc{c}{w, \ework{r} \semi P} \; \mapsto \;
  \proc{c}{w+r, P}$
\end{tabbing}
Statically, type checking this construct requires potential.
No other typing rule affects the potential.
\begin{mathpar}
  \infer[\m{work}]
  {\vars \semi \cons \semi \D \entailpot{q} \ework{r} \semi P :: (x : A)}
  {\vars \semi \cons \proves q \geq r &
  \vars \semi \cons \semi \D \entailpot{q-r} P :: (x : A)}
\end{mathpar}
Since the amount of potential consumed to type check this expression
is equal to the amount of work performed by it, the type safety theorem
expresses that the total work done by a system can never exceed its
initial potential.

The algorithms for type checking
and equality extend easily to these new
constructors since we already track variables and constraints over
natural numbers in the judgments.

\paragraph{Queue Example}
Revisiting the ergometric session type of the queue data structure, we
count the total number of messages exchanged in the system.
\begin{sill}
  $\queue{A}[n] = \echoice{$\=$\mb{ins} : \textcolor{red}{\getpot^{2n}}
  (A \lolli \queue{A}[n+1]),$\\
  \>\hspace{-3em}$\mb{del} : \textcolor{red}{\getpot^{2}} \ichoice{$\=$\mb{none} :
  \tassert{n=0} \one,$\\
  \>\>$\mb{some} : \tassert{n > 0} A \tensor \queue{A}[n-1]}}$\\
  $\cdot \entailpot{0} \mi{empty} :: (s : \queue{A}[0])$\\
  $(x : A), (t : \queue{A}[n]) \entailpot{0} \mi{elem}[n] :: (s : \queue{A}[n+1])$
\end{sill}
Thus, the queue requires $2n$ units of potential to insert and $2$
units of potential to delete an element. The $2n$ units are used to
carry the $\mb{ins}$ message and the element to insert to the end of
the queue. While deletion, the process sends $2$ messages, either
$\mb{none}$ and $\m{close}$ (lines~\ref{proc:emp_none},\ref{proc:emp_close}
in Figure~\ref{fig:queue_impl}) or $\mb{some}$ and the element stored
(lines~\ref{proc:elem_some},\ref{proc:elem_x}). Operationally, when the
$\mi{elem}$ process receives $2(n+1)$ units on $s$, it consumes $2$
units to send 2 messages (lines~\ref{proc:elem_ins},\ref{proc:elem_y}
in Figure~\ref{fig:queue_impl}) and sends the remaining $2n$ units on
$t$ in accordance with the prescribed type. And on deletion, the $2$
units are consumed to send the two messages (lines~\ref{proc:elem_some},\ref{proc:elem_x}).

\section{Constraint and Work Reconstruction}
\label{sec:recon}

The process expressions introduced so far in the language
follow simple syntax-directed typing rules. This means they are
immediately amenable to be interpreted as an algorithm for
type-checking, calling upon a decision procedure where arithmetic
entailments and type equalities need to be verified.  However, this
requires the programmer to write a significant number of
$\m{assume}, \m{assert}, \m{pay}, \m{get}$ and $\m{work}$ expressions
in their code; constructs corresponding to proof constraints and
potential (quantifiers are still explicit). Relatedly, this hinders
reuse: we are unable to provide multiple types to the same program
so that it can be used in different contexts.

\begin{figure}
  \begin{mathpar}
    \vspace{-0.4em}
    \infer[{\tassertop}R]
    {\vars \semi \cons \semi \D \ivdash{q} P :: (x : \tassert{\phi} A)}
    {\vars \semi \cons \proves \phi & \vars \semi \cons \semi \D \ivdash{q} P :: (x : A)}
    \and\vspace{-0.4em}
    \infer[{\tassertop}L]
    {\vars \semi \cons \semi \D, (x : \tassert{\phi} A) \ivdash{q} Q :: (z : C)}
    {\vars \semi \cons \land \phi \semi \D, (x : A) \ivdash{q} Q :: (z : C)}
  \end{mathpar}
  \vspace{-2em}
  \caption{Implicit Typing Rules}
  \vspace{-1.5em}
  \label{fig:implicit-typing}
\end{figure}

This section introduces an \emph{implicit type system} in which the
source program never contains the
$\m{assume}, \m{assert}, \m{pay},$ $\m{get}$ and $\m{work}$ constructs.
Moreover, impossible branches may be omitted from $\m{case}$ expressions.  The
missing branches and other constructs are restored by a type-directed
process of \emph{reconstruction}.  In the first phase, a $\m{case}$
expression with a missing branch for label $\ell$ is extended by a
branch $\ell \Rightarrow \eimposs$ so that type checking later
verifies that the omitted branch is indeed impossible.  Then assume
and asserts, and finally pay and gets are inserted according to a
reconstruction algorithm described in this section. Finally, since
the potential is treated as linear (must be $0$ before termination),
$\m{work}$ constructs are inserted just before the terminating
expression to consume the leftover potential.

Following branch reconstruction, the resulting process expression is
typechecked with the implicit typing judgment
$\vars \semi \cons \semi \D \ivdash{q} P :: (x : A)$ using the rules
in Figure~\ref{fig:implicit-typing} (analogous rules $\tassumeop R,
\tassumeop L, \getpot R, \getpot L, \paypot R, \paypot L$ omitted).
The only difference from
the explicit system is that the process expressions do not change on
application of these rules. The remaining rules exactly match the
explicit system (see Appendix~\ref{app:formal}).

The implicit rules are trivially sound and complete with respect to
the explicit system, since from an implicit typing derivation we can
read off the corresponding explicit process expression and vice versa.
The rules are also manifestly decidable since the types in the premise
are smaller than the conclusion for all the rules presented.

However, the implicit type system is highly nondeterministic.  Given
an implicit source program, there may be many different corresponding
explicit programs depending on when the rules in
Figure~\ref{fig:implicit-typing} are applied.  The necessary
backtracking would greatly complicate error messages and could also be
inefficient.  To solve this problem, we introduce a novel
\emph{forcing calculus} which enforces an order among these implicit
constructs. The core idea of this calculus is to follow the structure
of each type, but within that \emph{ $\m{assume}$ and $\m{get}$ should
  be inserted as early as possible, and $\m{assert}$ and $\m{pay}$
  should be inserted as late as possible.} This reasoning is sound
since the constraints obey a \emph{monotonicity property}: if a
constraint is true at a program point, it will always be true later in
the program.  Thus, eagerly assuming and lazily asserting constraints
is sound: if a constraint can be proved now, it can be proved later.
It is also complete under the mild assumption that the types
can be polarized (explained below).
A similar reasoning
holds for the potential: a sequence of potential exchanges can always
be reordered as eagerly receiving all the potential, and then sending
it only when required.
Logically, the
$\tassumeop R, \tassertop L, \getpot R, \paypot L$ rules are
invertible, and are applied eagerly while their dual rules are
applied lazily. The result is again complete under the assumption
of polarizable types.

This strategy is formally realized in the forcing calculus using the
judgment
$\vars \semi \cons \semi \D \semi \W \entailpot{q} P :: (x : A)$.  The
context is split into two: the linear context $\D$ contains stable
propositions on which the invertible left rules have been applied,
while the ordered context $\W$ stores channels on which invertible
rules can possibly still be applied to. First, we assign polarities to
the type operators with implicit expressions, a notion borrowed from
focusing~\cite{Andreoli92} with a similar function here.  Type
definitions are unfolded in order to determine their polarity, which
is always possible since type definitions are contractive. The types
that involve communication are called \emph{structural} and
represented by $S$.

\[
  \begin{array}{lcl}
    A^+ & ::= & S \mid \tassert{\phi} A^+ \mid \tpaypot{A^+}{r} \\
    A^- & ::= & S \mid \tassume{\phi} A^- \mid \tgetpot{A^-}{r} \\
    A & ::= & A^+ \mid A^- \\
    S & ::= & \ichoice{\ell : A}_{\ell \in L}
    \mid \echoice{\ell : A}_{\ell \in L}
    \mid A \tensor A \mid \one \mid A \lolli A \\
    & & \mid \texists{n} A \mid \tforall{n} A
  \end{array}
\]
Not all types can be polarized in this manner.  For example,
$\tassume{\phi} \tassert{\psi} A$ or $\tpaypot{\tgetpot{A}{q}}{r}$.
When checking the validity of types before performing reconstruction
we reject such types with alternating polarities since our
deterministic algorithm would be incomplete and we have found no need
for them. We also require that all process declarations contain
only structural types at the top-level.

Thus, $\tassertop$ and $\paypot$ operators are positive, while
$\tassumeop$ and $\getpot$ are negative. The structural types, denoted
by $S$ are considered neutral. In the forcing calculus, the invertible
rules are applied first (analogous $\tassertop L, \paypot L$ omitted).
\begin{mathpar}
  
  \infer[\tassumeop R]
  {\vars \semi \cons \semi \D^- \semi \W \entailpot{q}
  P :: (x : \tassume{\phi} A^-)}
  {\vars \semi \cons \land \phi \semi \D^- \semi \W
  \entailpot{q} P :: (x : A^-)}
  \and
  \infer[\getpot R]
  {\vars \semi \cons \semi \D^- \semi \W \entailpot{q}
  P :: (x : \tgetpot{A^-}{r})}
  {
  \vars \semi \cons \semi \D^- \semi \W
  \entailpot{q+r} P :: (x : A^-)}
\end{mathpar}
If a negative type is encountered in the ordered context,
it is considered stable (invertible rules applied)
and moved to $\D^-$.
\begin{mathpar}
  \infer[\m{move}]
  {\vars \semi \cons \semi \D^- \semi \W \cdot
  (x : A^-) \entailpot{q} P :: (z : C^+)}
  {\vars \semi \cons \semi \D^-, (x : A^-) \semi
  \W \entailpot{q} P :: (z : C^+)}
\end{mathpar}
The ordered context $\W$ imposes an order on the channels
on which these invertible rules are applied.

Once all the invertible rules are applied, we reach a
stable sequent of the form $\vars \semi \cons \semi
\D^- \semi \cdot \entailpot{q} P :: (x : A^+)$, i.e.,
the ordered context is empty and the offered type $A^+$
is positive. A stable sequent implies that all constraints
and potential have been received. We send a constraint
or potential lazily, i.e., just before communicating on
that channel. We realize this by \emph{forcing} the
channel just before communicating on it. As an example, while
sending (or receiving) a label on channel $x$, we force it.
\begin{mathpar}
  
  \infer[\oplus F_R]
  {\vars \semi \cons \semi \D^- \semi \cdot \entailpot{q} \esendl{x}{k} \semi P ::
  (x : A^+)}
  {\vars \semi \cons \semi \D^- \semi \cdot \entailpot{q} \esendl{x}{k} \semi P ::
  \focus{x : A^+}}
\end{mathpar}
\begin{mathpar}
  
  \infer[\oplus F_L]
  {\vars{\semi}\cons {\semi} \D, (x : A^-) {\semi} \cdot \entailpot{q}
  \ecase{x}{\ell}{Q_\ell}_{\ell \in L} :: (z : C^+)}
  {\vars{\semi}\cons {\semi} \D, \focus{x : A^-} {\semi} \cdot \entailpot{q}
  \ecase{x}{\ell}{Q_\ell}_{\ell \in L} :: (z : C^+)}
\end{mathpar}
The square brackets $[\cdot]$ indicates that the channel is
forced, indicating that a communication is about to happen
on it. If there are $\m{assert}$ or $\m{pay}$ constructs pending
on the forced channel, they are applied now
(analogous $\tassumeop L, \getpot L$ omitted).
\begin{mathpar}
  
  \infer[\tassertop R]
  {\vars \semi \cons \semi \D^- \semi \cdot \entailpot{q} P :: \focus{x : \tassert{\phi} A^+}}
  {\vars \semi \cons \proves \phi \and
  \vars \semi \cons \semi \D^- \semi \cdot \entailpot{q} P :: \focus{x : A^+}}
  \and
  \infer[\paypot R]
  {\vars \semi \cons \semi \D^- \semi \cdot \entailpot{q} P :: \focus{x : \tpaypot{A^+}{r}}}
  {\vars \semi \cons \proves q \geq r &
  \vars \semi \cons \semi \D^- \semi \cdot \entailpot{q-r} P :: \focus{x : A^+}}
\end{mathpar}
Finally, if a forced channel has a structural type, we apply the
corresponding structural rule and \emph{lose the forcing}. Again,
as an example, we consider the internal choice operator.
\begin{mathpar}
  \infer[\oplus R_k]
  {\vars \semi \cons \semi \D^- \semi \cdot \entailpot{q} (\esendl{x}{k} \semi P) ::
  \focus{x : \ichoice{\ell : A_\ell}_{\ell \in L}}}
  {(k \in L) \quad
  \vars \semi \cons \semi \D^- \semi \cdot \entailpot{q} P :: (x : A_k)}
  \and
  \infer[\oplus L]
  {\vars {\semi} \cons {\semi} \D, \focus{x {:} \ichoice{\ell : A_\ell}} {\semi} \cdot \entailpot{q}
  \ecase{x}{\ell}{Q_\ell} {::} (z {:} C^+)}
  {(\forall \ell \in L) \quad
  \vars \semi \cons \semi \D \semi (x : A_\ell) \entailpot{q} Q_\ell :: (z : C^+)}
\end{mathpar}
In either case, applying the structural rule creates a possibly unstable
sequent, thereby restarting the inversion phase.

Remarkably, \emph{the forcing calculus is sound and complete with respect
to the implicit type system}, assuming types can be polarized.
Since every rule in the forcing calculus is also present in the implicit system,
it is trivially sound. Moreover, applying $\m{assume}$ and $\m{get}$ eagerly,
and $\m{assert}$ and $\m{pay}$ lazily also turns out to be complete due
to the monotonicity property of constraints and potential.

\begin{theorem}[Soundness and Completeness]\label{thm:forcing}
  For (valid) polarized types $A$ and context $\D$ we have:
  \begin{enumerate}[leftmargin=*]
    \item If $\vars {\semi} \cons {\semi} \D \ivdash{q} P :: (x {:} A)$, then
    $\vars {\semi} \cons {\semi} \cdot {\semi} \D \entailpot{q} P :: (x {:} A)$.

    \item If $\vars {\semi} \cons {\semi} \cdot {\semi} \D \entailpot{q} P :: (x {:} A)$, then
    $\vars {\semi} \cons {\semi} \D \ivdash{q} P :: (x {:} A)$.
  \end{enumerate}
\end{theorem}

\begin{proof}
  Proof of 1. follows by induction on the implicit typing judgment.
  Proof of 2. follows by induction on the forcing judgment. See
  Appendix~\ref{app:forcing} for the full rules and proof.
\end{proof}


If a process is well-typed in the implicit system, it is well-typed
using the forcing calculus. Reading off the process expression from the
typing derivation in the forcing calculus results in the corresponding
explicit program. Thus, if a reconstruction is possible, the forcing calculus
will find it! We use this calculus to reconstruct the explicit program, which
is then typechecked using the explicit typing system.

\section{Implementation and Evaluation}

We have implemented a protoype for the language in Standard ML (about
6500 lines of code) available open-source that closely adheres to the
theory presented here.  Command line options determine whether to use
explicit or implicit syntax, and the result of reconstruction can be
displayed if desired.  We use a straightforward implementation of
Cooper's algorithm~\cite{cooper1972theorem} to decide Presburger
arithmetic with two small but significant optimizations.  One takes
advantage of the fact that we are working over natural numbers rather
than integers which bounds possible solutions from below, and the
other is to eliminate constraints of the form $x = e$ by substituting
$e$ for $x$ in order to reduce the number of variables.  For the implementation of
type equality we assign internal names to subexpressions parameterized
over their free index variables in order to facilitate efficient loop
detection.
After checking the validity of types, the implementation reconstructs
missing branches, then constraints, and finally work.  Verifying
constraints is postponed to the final pass of type-checking the
reconstructed process expression.

\begin{table}[t]
  \centering
  \begin{tabular}{l r r r r r}
  \textbf{Module} & \textbf{iLOC} & \textbf{eLOC} & \textbf{\#Defs} & \textbf{R (ms)} & \textbf{T (ms)} \\
  \toprule
  arithmetic & 69 & 143 & 8 & 0.353 & 1.325 \\
  integers & 90 & 114 & 8 & 0.200 & 1.074 \\
  linlam & 54 & 67 & 6 & 0.734 & 4.003 \\
  list & 244 & 441 & 29 & 1.534 & 3.419 \\
  primes & 90 & 118 & 8 & 0.196 & 1.646 \\
  segments & 48 & 65 & 9 & 0.239 & 0.195 \\
  ternary & 156 & 235 & 16 & 0.550 & 1.967 \\
  theorems & 79 & 141 & 16 & 0.361 & 0.894 \\
  tries & 147 & 308 & 9 & 1.113 & 5.283 \\
  \midrule
  \textbf{Total} & \textbf{977} & \textbf{1632} & \textbf{109} & \textbf{5.280} & \textbf{19.806} \\
  \bottomrule
  \end{tabular}
  \caption{Case Studies}
  \label{tab:case_study}
  \vspace{-2em}
  \end{table}

We have a variety of 9 case studies implemented, totaling about 1000 lines
of code type-checked in our implementation. Table~\ref{tab:case_study}
describes the results: iLOC describes the lines of source code in
implicit syntax, eLOC describes the lines of code after reconstruction
(which inserts implicit constructs), \#Defs shows the number of process
definitions, R (ms) and T (ms) show the reconstruction and type-checking
time in milliseconds respectively. The experiments
were run on an Intel Core i5 2.7 GHz processor with 16 GB 1867 MHz DDR3
memory. We briefly describe each case study.
\begin{enumerate}[leftmargin=*]
\item \textbf{arithmetic}: natural numbers in unary and binary
  representation indexed by their value and processes implementing
  standard arithmetic operations.

\item \textbf{integers}: an integer counter represented using two indices
  $x$ and $y$ with value $x-y$.

\item \textbf{linlam}: expressions in the linear $\lambda$-calculus
  indexed by their size with an \emph{eval} process to evaluate them.

\item \textbf{list}: natural number lists indexed by their size,
  and their standard operations: \emph{append, reverse, map, fold}, etc.

\item \textbf{primes}: implementation of the sieve of Eratosthenes to
  classify numbers as prime or composite.

\item \textbf{segments}: type $\m{seg}[n] =
  \forall k. \m{list}[k] \lolli \m{list}[n+k]$ representing partial lists
  with constant-work append operation.

\item \textbf{ternary}: natural numbers represented in balanced
  ternary form with digits $0, 1, -1$, indexed by their value, and
  some standard operations on them.

\item \textbf{theorems}: processes representing (circular~\cite{Derakhshan19corr})
  proofs of simple theorems such as $n(k+1) = nk+n$.

\item \textbf{tries}: a trie data structure to store multisets of
  binary numbers, with constant amortized work insertion and deletion
  verified with ergometric types.
\end{enumerate}

We give details for some of these examples in Appendix~\ref{app:examples}.

\paragraph{Queue as Two Lists}
The running $\queue{A}[n]$ example has a linear insertion cost,
since the element travels to the end of the queue where it is inserted.
However, a more efficient implementation of a queue
using two lists (or stacks)~\cite{Das18RAST} has a
constant amortized cost for insertion and deletion. Elements are
inserted into an \emph{input list}, and removed from an \emph{output list}. If
the output list is empty at removal, the input list is reversed and
made the output list. The input list stores extra potential (4 units)
which is consumed during reversal. Since ergometric types
support amortized analysis, we obtain the constant cost type

\begin{sill}
  $\queue{A}[n] = \echoice{$\=$\mb{ins} : \textcolor{red}{\getpot^{6}}
  (A \lolli \queue{A}[n+1]),$\\
  \>\hspace{-3em}$\mb{del} : \textcolor{red}{\getpot^{4}} \ichoice{$\=$\mb{none} :
  \tassert{n=0} \one,$\\
  \>\>$\mb{some} : \tassert{n > 0} A \tensor \queue{A}[n-1]}}$
\end{sill}

\paragraph{Binary Numbers}
We can represent binary numbers indexed by their value as sequences
of bits $\mb{b0}$ and $\mb{b1}$ followed by $\mb{e}$ to indicate the end of the
sequence.

\begin{sill}
  $\m{bin}[n] = \ichoice{$ \= $\mb{b0} : \tassert{n > 0} \texists{k} \tassert{n = 2*k} \m{bin}[k],$ \\
    \> $\mb{b1} : \tassert{n > 0} \texists{k} \tassert{n = 2*k+1} \m{bin}[k],$ \\
    \> $\mb{e} : \tassert{n = 0} \one\,}$
\end{sill}

A binary number on outputting $\mb{b0}$ (or $\mb{b1}$) must send a proof
of $n > 0$ (no leading zeros) and a witness $k$ such that
$n = 2k$ (resp. $n = 2k+1$) before continuing with $\m{bin}[k]$.
While outputting $\mb{e}$, it must send
a proof that $n = 0$ and then terminate, as described in the type. Moreover, since we
reconstruct impossible branches, when a programmer implements
\emph{a predecessor} process declared as

\begin{tabbing}
  $(m : \m{bin}[n+1]) \vdash \mi{pred}[n] :: (n : \m{bin}[n])$
\end{tabbing}

they can skip the impossible case of label $\mb{e}$ (since $n+1 \neq 0$).

\paragraph{Linear $\lambda$-calculus}
We demonstrate an implementation of the (untyped) linear $\lambda$-calculus
in which the index objects track the size of the expression. 

\begin{sill}
  $\m{exp}[n] = \ichoice{$\=$\mb{lam} : \tassert{n > 0}
    \tforall{n_1} \m{exp}[n_1] \lolli \m{exp}[n_1+n-1],$\\
    \hspace{2em}$\mb{app} : \texists{n_1} \texists{n_2} \tassert{n = n_1+n_2+1} \m{exp}[n_1] \tensor \m{exp}[n_2]}$
\end{sill}

An expression is either a $\lambda$ (label $\mb{lam}$) or an application
(label $\mb{app}$). In case of $\mb{lam}$, it expects a number $n_1$
and an argument of size $n_1$ and
then behaves like the body of the $\lambda$-abstraction of size $n_1+n-1$.
In case of $\mb{app}$, it will send $n_1$ and $n_2$ such that $n = n_1
+n_2+1$, then an expression
of size $n_1$ and then behaves as an expression of size $n_2$.

A value can only be a $\lambda$ expression

\begin{sill}
$\m{val}[n] = \ichoice{$\=$\mb{lam} : \tassert{n > 0}
  \tforall{n_1} \m{exp}[n_1] \lolli \m{exp}[n_1+n-1]}$
\end{sill}
so the $\mb{app}$ label is not permitted.
Type checking verifies that that the result of evaluating
a linear $\lambda$-term is no larger than the original term.
\begin{sill}
$(e : \m{exp}[n]) \vdash \mi{eval}[n] :: (v : \texists{k} \tassert{k \leq n} \m{val}[k])$
\end{sill}

\section{Further Related Work}\label{sec:related}


Refinement types were introduced to allow specification and verification
of recursively defined subtypes of user-defined types
\cite{Xi99popl,Freeman91Refinement}, but have since been applied for
complexity analysis~\cite{SynthPLDI19,CraryPOPL00}. \citet{LagoLDT11}
designed a system of linear dependent types for the $\lambda$-calculus
with higher-order recursion and use it for complexity analysis~\cite{LagoGeometry13}.
Refinement~\cite{LiqCost20} and dependent types~\cite{DanielssonLazy08} have
also been employed to reason about efficiency of lazy functional programs.
Refinement type and effect systems have been proposed for incremental
computational complexity~\cite{Cicek15CostIt} and relational cost analysis of
functional~\cite{EzgiRelCost} and functional-imperative programs~\cite{QuFuncImpRelCost}.
Automatic techniques for complexity analysis of sequential~\cite{HoffmannW15}
and parallel programs~\cite{HoffmannESOP15} that do not rely on refinements
have also been studied. In contrast to these articles
that use nominal types and apply to functional programs, the structural
type system of session types poses additional theoretical and practical
challenges for deciding type equality, type checking, and reconstruction.

Label-dependent session types~\cite{ThiemannLDST} use a
limited form of dependent types where values can depend on labels
drawn from a finite set. They use this
to encode general binary session types and also extend the types
with primitive recursion over natural numbers,
although unlike our work, they do not support general recursive types.
\citet{ToninhoFOSSACS18} propose a dependent type theory combining
functions and session types through a contextual monad allowing
processes to depend on functions and vice-versa.
Unlike our equality algorithm that involves no type-level
computation, they rely on term (or process) equality to
define type equality.
\citet{ToninhoPPDP11} develop an interpretation of linear
type theory as dependent session types for a term passing extension
of the $\pi$-calculus to express interface contracts
and proof-carrying certification. However, they do not discuss a type
equality algorithm nor provide an implementation. 
\citet{WuX17DST} propose a dependent session type system of
DML style based on ATS~\cite{Xi98appliedtype} formalizing type equality
in terms of subtyping and regular constraint relations.
In contrast to our refinement layer, none of these dependent type
systems are applied for complexity analysis.

LiquidPi~\cite{Griffith13LiqPi} applies the idea of refinements
to session types to describe and validate safety properties of
distributed systems. They also present an algorithm for
inferring these refinements when they are restricted to a finite
set of templates. However, they do not specifically explore the
fragment of arithmetic refinements, nor apply them to study
resource analysis. Linearly refined session
types~\cite{Baltazar12LRST} extend the $\pi$-calculus with capabilities
from a fragment of multiplicative linear logic. These capabilities
encode authorization logic enabling fine-grained specifications,
i.e., a process can take an action only if it contains certain
capabilities. \citet{Franco14LRST} implement these linearly refined
session types in a language called SePi. In this work, we explore
arithmetic refinements that are more general than a multiset of
uninterpreted formulae~\cite{Baltazar12LRST}. \citet{Bocchi19AsyncTimedMPST}
present asynchronous timed session types to model timed protocols,
ensuring processes perform actions in the time frame prescribed
by their protocol. \citet{Zhou19Fluid} refine base types with arithmetic
propositions~\cite{ZhouThesis}
in the context of multiparty session types without recursive types.
In this restricted setting, subtyping and therefore type equality
is decidable and much simpler than in our setting.
Finally, session types with limited arithmetic refinements
(only base types could be refined) have been proposed for the purpose of
runtime monitoring~\cite{Gommerstadt18esop,Gommerstadt19phd},
which is complementary to our uses for static verification.
They have also been proposed to capture work~\cite{Das18RAST,Das19Nomos}
and parallel time~\cite{Das18Temporal}, but parameterization
over index objects was left to an informal meta-level and not
part of the object language.  Consequently, these languages
contain neither constraints nor quantifiers, and the metatheory
of type equality, type checking, and reconstruction in the presence
of index variables was not developed.

\section{Conclusion}

This paper explored the metatheory of session types extended with
arithmetic refinements. The type system was enhanced with quantifiers
and type constraints and applied to verify sequential complexity
bounds (characterizing the total work) of session-typed programs.

In the future we plan to pursue several natural generalizations.  In
multiple examples we have noted that even nonlinear arithmetic
constraints that arise have simple proofs, despite their general
undecidability, so we want to develop a heuristic
nonlinear solver.  Secondly, much of the theory in this paper is modular 
relying on a few simple properties of quantified linear arithmetic and
could easily be generalized to other domains such as quantifier-free
index domains with SMT solvers, arbitrary integers, modular arithmetic, and
fractional potentials. We would also like to generalize our
approach to a mixed linear/nonlinear language~\cite{Benton94csl} or
all the way to adjoint session
types~\cite{Pfenning15fossacs,Pruiksma19places}.


We also plan to explore automatic inference of potential
annotations. Currently, programmers have to compute work bounds,
express them in the type, and let the type-checker verify them. With
inference some of this work may be automated, although the tradeoff
between automation and precision of error messages will have to be
carefully weighed.
%
Finally, prior work has explored a temporal linear type system for
parallel complexity analysis~\cite{Das18Temporal} and we would like to
explore if similar type-checking and reconstruction algorithms can be
devised.  However, its proof-theoretic properties are not as uniform
as those for quantifiers, constraints, and ergometric types.



\bibliography{refs}

\appendix

\clearpage

\section{Further Examples}
\label{app:examples}
\lstset{basicstyle=\ttfamily\small, keepspaces=true, columns=fullflexible}

We present several different kinds of example from varying domains
illustrating different features of the type system and algorithms.
To transition from abstract syntax to the concrete syntax from the
implementation, we use Table~\ref{tab:syntax} as a guide.
In addition, types and processes are declared and defined as follows.
\begin{lstlisting}
type V{n} = A
decl f{n1}{n2}: (x1 : A1) ... (xn : An) |{q}- (x : A)
proc x <- f{n1}{n2} <- x1 ... xn = P
\end{lstlisting}
The first line shows a type definition of $V$ with one index $n$
defined with type expression $A$. The second line shows a process declaration:
process $f$ with two indices $n_1$ and $n_2$, channels $x_1 \ldots
x_n$ in its context with types $A_1 \ldots A_n$ respectively,
storing potential $q$, and offering along channel $x$ of type $A$.
The last line shows the process definition of $f$ offering channel
$x$ and using channels $x_1 \ldots x_n$ with $P$ being a process
expression.

\begin{table}[t]
  \centering
  \begin{tabular}{l l}
  \textbf{Abstract Syntax} & \textbf{Concrete Syntax} \\
  \toprule
  $\ichoice{l_1 : A_1, l_2 : A_2}$ & \verb|+{l1 : A1, l2 : A2}| \\
  $\echoice{l_1 : A_1, l_2 : A_2}$ & \verb|&{l1 : A1, l2 : A2}| \\
  $A \tensor B$ & \verb|A * B| \\
  $A \lolli B$ & \verb|A -o B| \\
  $\one$ & \verb|1| \\
  $\texists{n} A$ & \verb|?n. A| \\
  $\tforall{n} A$ & \verb|!n. A| \\
  $\tassert{n = 0} A$ & \verb|?{n = 0}. A| \\
  $\tassume{n = 0} A$ & \verb|!{n = 0}. A| \\
  $\tpaypot{A}{r}$ & \verb!|{r}> A! \\
  $\tgetpot{A}{r}$ & \verb!<{r}| A! \\
  $V[n_1,n_2]$ & \verb|V{n1}{n2}| \\
  $\esendl{x}{k}$ & \verb|x.k| \\
  $\ecase{x}{l}{P}$ & \verb!case x (l => P)! \\
  $\esendch{x}{w}$ & \verb|send x w| \\
  $\erecvch{x}{y}$ & \verb|y <- recv x| \\
  $\eclose{x}$ & \verb|close x| \\
  $\ewait{x}$ & \verb|wait x| \\
  $\esendn{x}{e}$ & \verb|send x {e}| \\
  $\erecvn{x}{n}$ & \verb|{n} <- recv x| \\
  $\eassert{x}{n = 0}$ & \verb|assert x {n = 0}| \\
  $\eassume{x}{n = 0}$ & \verb|assume x {n = 0}| \\
  $\epay{x}{r}$ & \verb|pay x {r}| \\
  $\eget{x}{r}$ & \verb|get x {r}| \\
  \bottomrule
  \end{tabular}
  \caption{Abstract and Corresponding Concrete Syntax}
  \label{tab:syntax}
  \vspace{-2em}
  \end{table}

\subsection{Lists with Potential}

The type $\m{list}[n,p]$ is the type of lists of length $n$ where
each element carries potential $p$.  Since we do not have polymorphism,
we use elements of type $\m{nat}$ to stand in for an arbitrary type.
We start with the basic, unindexed version
\begin{lstlisting}
type nat = +{zero : 1, succ : nat}

type list = +{ cons : nat * list,
               nil : 1 }
\end{lstlisting}
$\m{list}$ is a purely positive type, so the provider of a list
just sends messages.  For example, ignoring constraints and
potentials for now, the list $3,4,5$ would be the following
sequence of messages:
\begin{lstlisting}
cons, a3, cons, a4, cons, a5, nil, close
\end{lstlisting}
Here, $a_3$, $a_4$, $a_5$ are \emph{channels} along which
representations of the numbers $3$, $4$, and $5$ are sent,
respectively.

Next we index lists by their length $n$ and assign a uniform potential
$p$ to every element.  This potential is transfered to the client of
the list before each element.
\begin{lstlisting}
type list{n}{p} = 
+{ cons : ?{n > 0}. |{p}> nat * list{n-1}{p},
   nil : ?{n = 0}. 1 }
\end{lstlisting}
Now our example list where each element has potential 2, provided
along an $l : \m{list}[3,2]$ would send the following messages, if all
implicit information were actually transmitted.
\begin{lstlisting}
cons, assert {3 > 0}, pay 2, a3,
cons, assert {2 > 0}, pay 2, a4,
cons, assert {1 > 0}, pay 2, a5,
nil, assert {0 = 0}, close
\end{lstlisting}

\paragraph{Nil and Cons Processes}
Next we would like to program a \emph{process} $\m{nil}$ that sends
the messages for an empty list.  $\m{nil}$ does not have to send any
potential, according to the type, but it must send two messages:
$\mb{nil}$ and $\mi{close}$.  Since our cost model assigns unit cost
to every send operation, the $\m{nil}$ process must carry a potential
of 2.  Nevertheless, it can be at the end of a list of any potential
$p$.
\begin{lstlisting}
decl nil{p} : . |{2}- (l : list{0}{p})

proc l <- nil{p} <- = l.nil ; close l
\end{lstlisting}
Reconstruction will add the assertion of \verb!0 = 0! to this implicit
code.  It will also add the necessary work before every send
operation, according to our cost model.  We will show the
reconstruction only for this example.
\begin{lstlisting}
proc l <- nil{p} <-  = 
  work ;
  l.nil ;
  assert l {0 = 0} ;
  work ;
  close l
\end{lstlisting}

A $\m{cons}$ \emph{process} should take an element $x$ and a list
$t$ and provide the list that sends $\mb{cons}$, then $x$, and
then behaves like $t$.  This time, let's first examine the code
before the type:
\begin{lstlisting}
proc l <- cons{n}{p} <- x t =
   l.cons ;
   send l x ;
   l <- t
\end{lstlisting}
Regarding typing, we see that if $t$ is a list of length $n$,
then $l$ should be a list of length $n+1$.  How much potential
does $\m{cons}$ need?  We need 2 units to send $\mb{cons}$
and the element $x$, but we also need $p$ units because we
are constructing a list where each element has potential $p$.
So, overall, $\m{cons}$ requires potential $p+2$.  Putting these
together, we get
\begin{lstlisting}
decl cons{n}{p} :
(x:nat) (t : list{n}{p}) |{p+2}- (l : list{n+1}{p})
\end{lstlisting}
If we make a mistake, for example, state the potential as
$p+1$ we get an error message:
\begin{lstlisting}
error:insufficient potential: true |/= p+1-1 >= 1
proc l <- cons{n}{p} <- x t =
l.cons ; send l x ; l <- t
         ~~~~~~~~ 
\end{lstlisting}
which pinpoints the source of the error: we have insufficient
potential to send the element $x$.

\paragraph{Appending Two Lists}
For the $\m{append}$ process, let's start again with the code.  We
receive the lists along $l_1$ and $l_2$ and send the result of
appending them along $l$.  We branch on $l_1$.  If it is the label
$\mb{cons}$, we receive the element $x$, then we send on $\mb{cons}$
and $x$ along $l$ and recurse.  When recursing, the length of $l_1$
(which is $n$) is reduced by one, while the length of $l_2$ (which is
$k$) stays the same.  When $l_1$ is empty (we receive
$\mb{nil}$), we wait for $l_1$ to finish and then forward $l_2$
to $l$.
\begin{lstlisting}
proc l <- append{n}{k}{p} <- l1 l2 =
  case l1
    ( cons => x <- recv l1 ;
              l.cons ; send l x ;
              l <- append{n-1}{k}{p} <- l1 l2
    | nil => wait l1 ; l <- l2 )
\end{lstlisting}
Considering the parallelism inherent in this process, we see that it
implements a pipeline from $l_1$ to $l$ with a constant delay.  The
list $l_2$ can be computed in parallel with this pipeline and
eventually is connected directly to the end of $l_1$.

It remains to reason about the potential.  We need two units of
potential to send $\mb{cons}$ and $x$, and this for each element of
$l_1$.  Forwarding as needed in the second branch is cost-free, so no
additional potential is needed.  Therefore we obtain the type
\begin{lstlisting}
decl append{n}{k}{p} :
(l1 : list{n}{p+2}) (l2 : list{k}{p}) |-
      (l : list{n+k}{p})
\end{lstlisting}
that is, each element of $l_1$ must have a potential $p+2$ and each
element of $l_2$ only potential $p$.  Also, the length of the output
list is $n+k$.

If we want a more symmetric form of $\m{append}$ where all three lists
carry potential $p$ per element, we can ``pre-pay'' the cost of
sending the $\mb{cons}$ and $x$ messages for each element of $l_1$.
Since $l_1$ has length $n$, this means the $\m{append}$ process
requires $2n$ units of potential.
\begin{lstlisting}
decl append{n}{k}{p} :
(l1 : list{n}{p}) (l2 : list{k}{p}) |{2*n}-
      (l : list{n+k}{p})
\end{lstlisting}
Note that this type can be assigned to \emph{exactly the same program}
as the first one: it is our choice if we want to require that $l_1$
carry the potential or that the process invoking append pre-pay the
total cost when invoking $\m{append}$.

\paragraph{Reversing a List}
The process $\m{rev}$ for reversing a list is quite similar to
$\m{append}$.  It uses an accumulator $a$ to which it sends all the
elements from the incoming list $l$.  When $l$ is empty, the reversed
list $r$ is just the accumulator.  The work analysis is analogous
to $\m{append}$, so we only show the result.
\begin{lstlisting}
decl rev{n}{k}{p} :
(l : list{n}{p+2}) (a : list{k}{p}) |-
      (r : list{n+k}{p})

proc r <- rev{n}{k}{p} <- l a =
  case l ( cons => x <- recv l ;
                   a' <- cons{k}{p} <- x a ;
                   r <- rev{n-1}{k+1}{p} <- l a'
         | nil => wait l ; r <- a )
\end{lstlisting}
Even though similar in work to append, its concurrent behavior is
quite different.  I cannot send an element along the output $r$ until
the whole input list $l$ has been processed.

To just reverse a list we have to create an empty accumulator
to start, which requires 2 units of potential, just once for
the whole list.
\begin{lstlisting}
decl reverse{n}{p} : (l : list{n}{p+2}) |{2}-
      (r : list{n}{p})
proc r <- reverse{n}{p} <- l =
  a <- nil{p} <- ;
  r <- rev{n}{0}{p} <- l a
\end{lstlisting}
Again, we could assign a different potential if we would be willing to
prepay for operations instead of expecting the necessary potential to
be stored with the elements of the input list $l$.

\paragraph{Recharging Potential}
We can also ``recharge'' the potential of a list by adding $2$ units
to every element.  For a list of length $n$, this requires up-front
potential of $4n+2$: a portion $2n$ goes to boost the potential of
each element, another portion $2n$ goes to actually send each
$\mb{cons}$ label and element $x$ from the input list, and $2$ units
go to sending $\mb{nil}$ and $\mi{close}$.
\begin{lstlisting}
decl charge2{n}{p} :
    (k : list{n}{p}) |{4*n+2}- (l : list{n}{p+2})

proc l <- charge2{n}{p} <- k =
  case k ( cons => x <- recv k ;
                   l.cons ; send l x ;
                   l <- charge2{n-1}{p} <- k
         | nil => wait k ; l.nil ; close l )
\end{lstlisting}
While this requires a lot of work, its parallel complexity is good
since it is a pipeline with a constant delay between input and output.

It would also be correct to recharge every element of the list with $q$ units
of potential.
\begin{lstlisting}
decl charge2{n}{p}{q} :
    (k : list{n}{p}) |{(q+2)*n+2}- (l : list{n}{p+q})
\end{lstlisting}
This, unfortunately, requires the nonlinear arithmetic expression
$(q+2)n+2$.  While a simple solver for polynomial constraint could
handle this, we currently reject this as nonlinear.

\paragraph{Map}
In a functional setting, mapping a function $f$ over a list of length
$n$ requires $n$ uses of $f$ and is therefore not linear.  However,
reuse can be replaced by recursion.  In this
example, a \emph{mapper} from type $A$ to type $B$ is a process that
provides the choice between two labels, $\mb{next}$ and $\mb{done}$.
When receiving $\mb{next}$ it then receives an element of type $A$,
responds with an element of type $B$ and recurses to wait for the next
label.  If it receives $\mb{done}$, it terminates.  We have
\[
  \m{mapper}_{AB} = \echoice{\mb{next} : A \lolli B \tensor \m{mapper}_{AB},
    \mb{done} : \one}
\]
Here, we use $\m{nat}$ for both $A$ and $B$.  We further assume for
simplicity that the mapper has enough internal potential so it does
not require any potential from the $\m{map}$ process.  We then have
\begin{lstlisting}
proc l <- map{n}{p} <- k m =
  case k ( cons => x <- recv k ;
                   m.next ; send m x ;
                   y <- recv m ;
                   l.cons ; send l y ;
                   l <- map{n-1}{p} <- k m
         | nil => wait k ;
                  m.done ; wait m ;
                  l.nil ; close l )
\end{lstlisting}
We see there that for each element, $\m{map}$ needs to send 4 messages
($\mb{next}$, $x$, $\m{cons}$, and $y$) so the input list should have
potential $p+4$. There is also a constant overhead of $3$ for the
empty list ($2$ for $\mb{nil}$ and $\mi{close}$, and $1$ to notify the
mapper that we are done).
\begin{lstlisting}
decl map{n}{p} :
(k : list{n}{p+4}) (m:mapper) |{3}- (l : list{n}{p})
\end{lstlisting}
Folding a list can be done in a similar fashion.

\paragraph{Filter}
Filtering elements from an input list is interesting because we cannot
statically predict the length of the output list.  So for the first
time in the list examples we require a quantifier.  However, we know
that the output list is not longer than the input list so we define a
new type
$\m{bdd\_list}[n,p] = \texists{m} \tassert{m \leq n} \m{list}[m,p]$.
\begin{lstlisting}
type bdd_list{n}{p} = ?m. ?{m <= n}. list{m}{p}
\end{lstlisting}
Even in implicit form, this type requires communication of the witness
$m$.  While not strictly needed, it is helpful to define bounded
versions of $\m{nil}$ and $\m{cons}$ which have the same ergometric
properties (definitions elided).  More interesting is the cost-free
$\m{bdd\_resize}$ which we can use to inform the type checker that a
list bounded by $n$ is also bounded by $n+1$: the type checker just
has to verify that $m \leq n$ implies $m \leq n+1$.
\begin{lstlisting}
decl bdd_nil{p} : . |{2}- (l : bdd_list{0}{p})
decl bdd_cons{n}{p} :
  (x : nat) (k : bdd_list{n}{p}) |{p+2}-
     (l : bdd_list{n+1}{p})

decl bdd_resize{n}{p} :
  (k : bdd_list{n}{p}) |- (l : bdd_list{n+1}{p})

proc l <- bdd_resize{n}{p} <- k =
  {m} <- recv k ;
  send l {m} ;
  l <- k
\end{lstlisting}
A \emph{selector} responds $\mb{false}$ for elements to be excluded
from the result, and it responds $\mb{true}$ and then returns the
element itself so it can be included in the output list.  Either way,
it recurses so the next element can be tested.
\begin{tabbing}
  $\m{selector}_A = \echoice{$\=$\mb{next} : A \lolli
    \ichoice{$\=$\mb{false} : \m{selector}_A,$\\
  \>\>$\mb{true} : A \tensor \m{selector}_A}$\\
  \>$\mb{done} : \one}$
\end{tabbing}
Concretely (using $\m{nat}$ for $A$):
\begin{lstlisting}
type selector =
&{ next : nat -o +{ false : selector,
                    true : nat * selector },
   done : 1 }
\end{lstlisting}
Then the filter process has the type and definition
as described in Figure~\ref{fig:filter_proc}.
\begin{figure*}
\begin{lstlisting}
decl filter{n}{p} : (s : selector) (k : list{n}{p+4}) |{3}- (l : bdd_list{n}{p})

proc l <- filter{n}{p} <- s k =
  case k ( cons => x <- recv k ;
                   s.next ; send s x ;
                   case s ( false => l' <- filter{n-1}{p} <- s k ;
                                     l <- bdd_resize{n-1}{p} <- l'
                          | true => x' <- recv s ;
                                    l' <- filter{n-1}{p} <- s k ;
                                    l <- bdd_cons{n-1}{p} <- x' l' )
          | nil => wait k ; s.done ; wait s ;
                   l <- bdd_nil{p} <- )
\end{lstlisting}
\caption{Type and Definition of \emph{filter} process}
\label{fig:filter_proc}
\end{figure*}

\paragraph{A Queue as Two Lists}
Finally, we return to the queue we used earlier as a running example
An efficient functional implementation uses two lists to implement the
queue.  Actually, its efficiency is debatable if used
non-linearly~\cite{okasaki_1995}, but here the type checker establishes
constant-time amortized cost for enqueuing and dequeuing messages.

The algorithm is straightforward: the queue process maintains two
lists, $\mi{in}$ and $\mi{out}$.  When elements are enqueued, they are
put into the $\mi{in}$ queue, where each element has a potential of
$4$.  When elements are dequeued, we take them from the $\mi{out}$
list where each element has a potential of $2$, which is enough to
send back the element to the client.  When the output list is empty
when a dequeue request is received, we first reverse the input list
and make it the output list.  This reversal costs $2$ units of
potential for each element, but we have stored $4$ so this is
sufficient.

The amortized cost of an enqueue is therefore $6$: $2$ to construct
the new element of the input list, plus $4$ because each element of
the input list must have potential $4$ to account for its later
reversal.  When calculating the cost of a dequeue we see it should
be $4$, leading us to the type
\begin{lstlisting}
type queue{n} = &{ enq : <{6}| nat -o queue{n+1},
                   deq : <{4}| deq_reply{n} }
type deq_reply{n} =
  +{ none : ?{n = 0}. 1,
     some : ?{n > 0}. nat * queue{n-1} }
\end{lstlisting}
Note that uses of $\tgetpot{}{6}$ and $\tgetpot{}{4}$ which mean that the
client has to transfer this potential, while in lists we used
$\tpaypot{}{p}$ so the list provider payed the potential $p$.

We have two processes definitions, one with an output list,
and one where we have noted that the output list is empty.
Figure~\ref{fig:queue2} shows the declarations and definitions.
\begin{figure*}
\begin{lstlisting}
decl queue_lists{n1}{n2} : (in : list{n1}{4}) (out : list{n2}{2}) |- (q : queue{n1+n2})
decl queue_rev{n1} : (in : list{n1}{4}) |{4}- (q : deq_reply{n1})

proc q <- queue_lists{n1}{n2} <- in out =
  case q ( enq => x <- recv q ;
                  in' <- cons{n1}{4} <- x in ;
                  q <- queue_lists{n1+1}{n2} <- in' out
         | deq => case out ( cons => x <- recv out ;
                                     q.some ; send q x ;
                                     q <- queue_lists{n1}{n2-1} <- in out
                           | nil => wait out ;
                                    q <- queue_rev{n1} <- in ) )

proc q <- queue_rev{n1} <- in =
  out <- reverse{n1}{2} <- in ;
  case out ( cons => x <- recv out ;
                     q.some ; send q x ;
                     in0 <- nil{4} <- ;
                     q <- queue_lists{0}{n1-1} <- in0 out
           | nil => wait out ;
                    q.none ; close q )
\end{lstlisting}
\caption{Queue as 2 Lists}
\label{fig:queue2}
\end{figure*}

To create a new queue from scratch we need $4$ units of potential in
order to create two empty lists.
\begin{lstlisting}
decl queue_new : . |{4}- (q : queue{0})
proc q <- queue_new <- =
  in0 <- nil{4} <- ;
  out0 <- nil{2} <- ;
  q <- queue_lists{0}{0} <- in0 out0
\end{lstlisting}

\subsection{Linear \texorpdfstring{$\lambda$}{Lambda}-Calculus}

We demonstrate an implementation of the (untyped) linear $\lambda$-calculus,
including evaluation, in which the index objects track the size of the
expression.  Type-checking verifies that that the result of evaluating
a linear $\lambda$-term is no larger than the original term.  Our
representation uses linear higher-order abstract syntax.

Ignoring issues of size for the moment, $\lambda$-calculus expressions
are represented with the type $\m{exp}$, values are of type $\m{val}$.
\begin{sill}
  $\m{exp} = \ichoice{$ \= $\mb{lam} : \m{exp} \lolli \m{exp},$ \\
    \> $\mb{app} : \m{exp} \tensor \m{exp}\,}$ \\
  $\m{val} = \ichoice{\,\mb{lam} : \m{exp} \lolli \m{exp}\,}$
\end{sill}
An expression is a process sending either the label $\mb{lam}$ or
$\mb{app}$.  In case of $\mb{lam}$ it then expects a channel along
which it receives the argument to the function and then behaves like
the body of the $\lambda$-abstraction.  In case of $\mb{app}$ it will
send a channel along which the function part of an application will be
sent and then behaves like the argument.  While this kind of encoding
may seem strange, it works miraculously.

For example, the representation of the identity function
$\lambda x.\, x$ is
\begin{lstlisting}
decl id : . |- (e : exp)
proc e <- id <- =
  e.lam ;
  x <- recv e ;
  e <- x
\end{lstlisting}
The application of the identity function to itself:
\begin{lstlisting}
decl idid : . |- (e : exp)
proc e <- idid <- =
  i1 <- id <- ;
  i2 <- id <- ;
  e.app ; send e i1 ; e <- i2
\end{lstlisting}
Similar to the $\m{nil}$ and $\m{cons}$ processes for lists, we have
two processes for constructing expressions and values.
\begin{lstlisting}
decl apply : (e1 : exp) (e2 : exp) |- (e : exp)
proc e <- apply <- e1 e2 =
  e.app ; send e e1 ; e <- e2

decl lambda : (f : exp -o exp) |- (v : val)
proc e <- lambda <- f =
  e.lam ; e <- f
\end{lstlisting}
Evaluation of a $\lambda$-abstraction immediately returns it as a
value.  For an application $e_1\, e_2$ we first evaluate the function.
By typing, this value is a process that will send the $\mb{lam}$ label
and then expects argument expression as a message.  So we send
$e_2$.
\begin{lstlisting}
decl eval : (e : exp) |- (v : val)
proc v <- eval <- e =
  case e ( lam => v <- lambda <- e
         | app => e1 <- recv e ;       % e = e2
                  v1 <- eval <- e1 ;
                  case v1 ( lam => send v1 e ;
                                   v <- eval <- v1 ))
\end{lstlisting}
The trickiest part here is the line marked with the comment
\verb'e = e2'.  An $e$ representing $e_1\, e_2$ will first send
$e_1 : \m{exp}$ and then behave like $e_2$.  So what we send to $v_1$
two lines below in fact represents $e_2$.

To track the sizes of $\lambda$-expressions we need quantifiers.  For
example, an application $e_1\, e_2$ has size $n$ if there exist sizes
$n_1$ and $n_2$ such that $n = n_1+n_2+1$ and of $e_1$ and $e_2$ have
size $n_1$ and $n_2$, respectively.

The case of $\lambda$-abstractions is more complicated.  The issue is
that we cannot assign a fixed size (say 1) to the variable in a
$\lambda$-abstraction because we may apply it to an argument of any
size.  Instead, the body can receive an arbitrary size $n_1$ and then
an expression of size $n_1$.

These considerations lead us to the types
\begin{sill}
  $\m{exp}[n] = $\\
  $\ichoice{$ \= $\mb{lam} : \tassert{n > 0}
    \tforall{n_1} \m{exp}[n_1] \lolli \m{exp}[n_1+n-1],$\\
    \> $\mb{app} : \texists{n_1} \texists{n_2} \tassert{n = n_1+n_2+1} \m{exp}[n_1] \tensor \m{exp}[n_2] \,}$ \\
  $\m{val}[n] = $\\
  $\ichoice{$ \= $\mb{lam} : \tassert{n > 0}
    \tforall{n_1} \m{exp}[n_1] \lolli \m{exp}[n_1+n-1]\,}$
\end{sill}
The types of $\m{apply}$, $\m{lambda}$, and $\m{eval}$ change
accordingly
\begin{lstlisting}
decl apply{n1}{n2} :
(e1 : exp{n1}) (e2 : exp{n2}) |- (e : exp{n1+n2+1})

decl lambda{n2} :
(f : !n1. exp{n1} -o exp{n1+n2}) |- (v : val{n2+1})

decl eval{n} :
(e : exp{n}) |- (v : ?k. ?{k <= n}. val{k})
\end{lstlisting}
Interesting here is only that the result type of evaluation contains
an existential quantifier since we do not know the precise size of the
value---we just know it is bounded by $n$.  Side calculation shows that
every reduction that takes place during evaluation reduces the size of
the output value by 2 since both the $\lambda$-abstraction and the
application disappear. The $\mi{eval}$ process is implemented in Figure
\ref{fig:eval_linlam}.
\begin{figure*}
\begin{lstlisting}
proc v <- eval{n} <- e =
  case e
    ( lam => send v {n} ;
             v <- lambda{n-1} <- e
    | app => {n1} <- recv e ;
             {n2} <- recv e ;           % n = n1 + n2 + 1
             e1 <- recv e ;             % e1 : exp{n1}, e = e2 : exp{n2}
             v1 <- eval{n1} <- e1 ;
             {k2} <- recv v1 ;          % v1 : val{k2} for some k2 <= n1
             case v1
              ( lam => send v1 {n2} ;
                       send v1 e ;   % v1 : exp{n2+k2-1}
                       v2 <- eval{n2+k2-1} <- v1 ; % v2 : val{l} for some l <= n2+k2-1 <= n2+n1-1 = n-2
                       {l} <- recv v2 ;
                       send v {l} ; v <- v2
              )
    )
\end{lstlisting}
\caption{Evaluating expressions in linear $\lambda$-calculus}
\label{fig:eval_linlam}
\end{figure*}

\subsection{Binary Numbers}\label{subsec:binary}
We can represent binary numbers as sequences of bits $\mb{b0}$
and $\mb{b1}$ followed by $\mb{e}$ to indicate the end of the
list.  The representation function $\overline{n}$ is given by
\[
  \begin{array}{lcll}
    \overline{0} & = & \mb{e}, \mi{close} \\
    \overline{2n} & = & \mb{b0}, \overline{n} & \mbox{n > 0} \\
    \overline{2n+1} & = & \mb{b1}, \overline{n}
  \end{array}
\]
The restriction on $n > 0$ in the second line ensures uniqueness of
representation (no leading zeros).  For the operations on them it is
important that the representation be in \emph{little endian} for, that
is, the least significant bit comes first.

Ignoring indexing for now, we would represent this as the type
\begin{sill}
  $\m{bin} = \ichoice{\mb{b0} : \m{bin}, \mb{b1} : \m{bin}, \mb{e} : \one}$
\end{sill}
However, this representation does not enforce the absence of leading
zeros.  One way to enforce that is by indexing this type by
the natural number it represents.  However, we would have to be
able to express that a number is even (which is possible in Presburger
arithmetic) since $\m{even}(n) \triangleq \exists k.\, 2*k = n$.
\begin{sill}
  $\m{bin}[n] = \ichoice{\mb{b0} : \tassert{n > 0 \land \m{even}(n)} \m{bin}[n/2], \ldots}$
\end{sill}
However, the value of the remaining sequence of bits would have to be
$n/2$ and that function is neither available nor expressible.  Instead,
we use a quantifier at the type level to give a name to $n/2$:
\begin{sill}
  $\m{bin}[n] = \ichoice{$ \= $\mb{b0} : \tassert{n > 0} \texists{k} \tassert{n = 2*k} \m{bin}[k],$ \\
    \> $\mb{b1} : \texists{k} \tassert{n = 2*k+1} \m{bin}[k],$ \\
    \> $\mb{e} : \tassert{n = 0} \one\,}$
\end{sill}

\paragraph{Zero and Successor Processes}
Now we can straightforwardly express zero and the successor processes
and verify their correctness, but we pay the price of having to send
and receive $k$.  In the code below name this index variable $n'$ to
represent $\lfloor n/2\rfloor$.  We receive the bits along channel $x$
and send the bits for $n+1$ along channel $y$.  If we receive $\mb{b0}$
we just flip the bit to $\mb{b1}$, send it to the output, and forward
all the remaining bit.  If we receive $\mb{b1}$ we flip it to $\mb{b0}$
but we have to call $\m{succ}$ recursively to account for the
``carry''.  If we receive $\mb{e}$ we output the representation of
$1$.
\begin{lstlisting}
decl zero : . |- (x : bin{0})
decl succ{n} : (x : bin{n}) |- (y : bin{n+1})

proc x <- zero <- = x.e ; close x

proc y <- succ{n} <- x =
  case x ( b0 => {n'} <- recv x ;
                 y.b1 ; send y {n'} ;
                 y <- x
         | b1 => {n'} <- recv x ;
                 y.b0 ; send y {n'+1} ;
                 y <- succ{n'} <- x
         | e => y.b1 ; send y {0} ;
                y.e ; wait x ; close y )
\end{lstlisting}
As can be seen from its type, type checking verifies that, for
example, $\m{succ}$ implements the successor function.  It is
plausible that we could omit the sending a receiving of $n'$, collect
constraints in a manner analogous to the other implicit constructs,
and check their validity, making this program even more compact.
However, two problems arise.  The first is that the quality of the
error messages degrades because first we need to collect all
constraints and then solve them together and it is no longer clear
were the error might have arisen.  And secondly we can no longer
perform reconstruction in general because in
$\forall n.\, \exists k.\, \phi$ propositions of Presburger arithmetic
the Skolem function that computes $k$ given $n$ may not be
expressible.  This means we would have to add something like Hilbert's
$\epsilon$ operator or definite description $\iota$, which is an
interesting item for future work.  So far in our examples the benefits
of good error messages have outweighed the few additional messages
that are sometimes required for quantified types.

\paragraph{Domain Restrictions}
We would like to apply the predecessor process only to
strictly positive numbers so we don't have to worry about computing
$0-1$.  We can express this with an explicit constraint in process
typing, which we have omitted in the body of the paper for the sake of
brevity.  This is quite similar to standard idioms in Liquid
types~\cite{Rondon08pldi,Griffith13LiqPi}.
\begin{lstlisting}
decl pred{n|n > 0} : (x : bin{n}) |- (y : bin{n-1})
proc y <- pred{n} <- x =
  case x
    ( b0 => {n'} <- recv x ;    % 2*k-1 = 2*(k-1)+1
            y.b1 ; send y {n'-1} ;
            y <- pred{n'} <- x
    | b1 => {n'} <- recv x ;
            y <- dbl0{n'} <- x  % 2*k+1-1 = 2*k
    % no case for e
    )
\end{lstlisting}
This is an example where the case for $\mb{e}$ is impossible (since  
$n > 0$ is incompatible with $n = 0$) so we omit this branch and it is  
reconstructed as shown in Section~\ref{sec:recon}.

Here $\m{dbl0}$ is a process with type
\begin{lstlisting}
decl dbl0{n} : (x : bin{n}) |- (y : bin{2*n})
\end{lstlisting}
which is not as trivial as one might think.  In fact, if we try
\begin{lstlisting}
proc y <- dbl0{n} <- x =
  y.b0 ; send y {n} ;
\end{lstlisting}
we get the error message
\begin{lstlisting}
error:assertion not entailed: true |/= 2*n > 0
  y.b0 ; send y {n} ;
         ~~~~~~~~~~ 
\end{lstlisting}
which means that this function is not correct because we cannot prove
that $2n > 0$ because we don't know if $n > 0$, possibly introduce
leading zeros!  The correct code requires one bit of look-ahead on
the input channel $x$.  Computing $2n+1$ is much simpler.
\begin{lstlisting}
proc y <- dbl0{n} <- x =
  case x ( b0 => {n'} <- recv x ;
                 y.b0 ; send y {n} ;
                 y.b0 ; send y {n'} ;
                 y <- x
         | b1 => {n'} <- recv x ;
                 y.b0 ; send y {n} ;
                 y.b1 ; send y {n'} ;
                 y <- x
         | e => y.e ; wait x ; close y )

decl dbl1{n} : (x : bin{n}) |- (y : bin{2*n+1})
proc y <- dbl1{n} <- x =
  y.b1 ; send y {n} ;
  y <- x
\end{lstlisting}

\paragraph{Discarding and Copying Numbers}
Even though we are in a purely linear language, message sequences of
purely positive type can be dropped or duplicated by explicit
programs.  It is plausible we could build this into the language as a
derivable idiom, maybe along the lines of equality types or type
classes in languages like Standard ML or Haskell.  However, this has
not been a high priority, so here are the processes that consume
a binary number without producing output, or duplicating a binary
number.
\begin{lstlisting}
decl drop{n} : (x : bin{n}) |- (u : 1)
proc u <- drop{n} <- x =
  case x ( b0 => {n'} <- recv x ; u <- drop{n'} <- x
         | b1 => {n'} <- recv x ; u <- drop{n'} <- x
         | e => wait x ; close u )

decl dup{n} :
(x : bin{n}) |- (xx : bin{n} * bin{n} * 1)
proc xx <- dup{n} <- x =
  case x ( b0 => {n'} <- recv x ;
                 xx' <- dup{n'} <- x ;
                 x1' <- recv xx' ;
                 x2' <- recv xx' ; wait xx' ;
                 x1 <- dbl0{n'} <- x1' ; send xx x1 ; 
                 x2 <- dbl0{n'} <- x2' ; send xx x2 ;
                 close xx
         | b1 => {n'} <- recv x ;
                 xx' <- dup{n'} <- x ;
                 x1' <- recv xx' ;
                 x2' <- recv xx' ; wait xx' ;
                 x1 <- dbl1{n'} <- x1' ; send xx x1 ;
                 x2 <- dbl1{n'} <- x2' ; send xx x2 ;
                 close xx
         | e => wait x ;
                x1 <- zero <- ; send xx x1 ;
                x2 <- zero <- ; send xx x2 ;
                close xx )
\end{lstlisting}
Note that the correctness of duplication (we get two copies of the
number $n$) is expressed in the indexed type and verified by the type
checker.

\paragraph{Comparisons}
Comparison is one place where we need to discard part of a binary
number because once the shorter number has been exhausted, we need to
consume the remaining bits of the longer number.  For comparison, it
is very helpful to have a representation that enforces the absence of
leading zeros.  But how do we express the outcome and verify the
correctness?  For this purpose we define a new type, $\m{ord}[m,n]$
to represent the outcome of the comparison of the representations
of $m$ and $n$.
\begin{lstlisting}
type ord{m}{n} = +{ lt : ?{m < n}. 1,
                    eq : ?{m = n}. 1,
                    gt : ?{m > n}. 1 }
\end{lstlisting}
Then the type
\begin{lstlisting}
decl compare{m}{n} :
(x : bin{m}) (y : bin{n}) |- (o : ord{m}{n})
\end{lstlisting}
guarantees the correctness of $\m{compare}$ because it can produce
$\mb{lt}$ only if $m < n$, and similarly for $\mb{eq}$ and $\mb{gt}$.
The code is somewhat tedious, because we do not have nested pattern
matching, but we show the beginning of the process definition because
it brings up an interesting fact regarding the type checker.
\begin{lstlisting}
proc o <- compare{m}{n} <- x y =
  case x
    ( b0 => {m'} <- recv x ;
            case y
              ( b0 => {n'} <- recv y ;
                      o <- compare{m'}{n'} <- x y ;
                      ... ) )
\end{lstlisting}
For this recursive call to type-check we have to make sure that
the outcome of $\m{compare}[m,n]$ (of type $\m{ord}[m,n]$)
is the same as the outcome of $\m{compare}[m',n']$ (which is of
type $\m{ord}[m',n']$.  This requires show that
\begin{tabbing}
  $m,n,m',n' \semi $\\
  \quad$m' > 0 \land m = 2*m' \land n' > 0 \land n = 2*n'
  \semi \cdot \vdash$\\
  \qquad$\m{ord}[m,n] \equiv \m{ord}[m',n']$
\end{tabbing}
This is certainly true declaratively (according to the definition of
type bisimulation) because $2*m' < 2*n'$ iff $m' < n'$ and similarly
for equality and greater-than comparison.  However, it goes beyond
reflexivity of indexed types and the generality of the $\m{def}$ rule
in the type equality algorithm (see Section~\ref{sec:tpeq}) is needed
in this and other examples.

We have other processes, such as multiplication (which require
nonlinear constraints and utilizes the $\m{dup}$ process) as well as
conversion to and from the unary representation of numbers.
Unfortunately we cannot easily characterize the \emph{work} of many of
these operations because they depend on the number of digits in the
binary representation, which is logarithmic in its value.  This source
of nonlinearity makes it difficult to track work unless we consider
numbers of a fixed width (such as 32-bit or 64-bit numbers).  However,
then we would be working in modular arithmetic and have to overcome
the lack of $\m{div}$ and $\m{mod}$ function in Presburger arithmetic.

\subsection{Prime Sieve}
As a slightly more complex example we use an implementation prime
sieve.  It turns out to be convenient to use an implementation of
unary natural numbers, which we elide, which satisfies the
following signature.
\begin{lstlisting}
type nat{n} = +{succ : ?{n > 0}. nat{n-1},
                zero : ?{n = 0}. 1}

decl zero : . |- (x : nat{0})
decl succ{n} : (x : nat{n}) |- (y : nat{n+1})
decl drop{n} : (x : nat{n}) |- (u : 1)
decl dup{n} :
(x : nat{n}) |- (xx : nat{n} * nat{n} * 1)
\end{lstlisting}
The output of the prime sieve is a stream of bits (which we call
$\mb{prime}$ and $\mb{composite}$) of length $k$ giving the status of
the numbers $2,3,4,5,6,...,k+2$.  It would start with
$\mb{prime}, \mb{prime}, \mb{composite}, \mb{prime}, \mb{composite},
\ldots$.
\begin{lstlisting}
type stream{k} = +{prime : ?{k > 0}. stream{k-1},
                   composite : ?{k > 0}. stream{k-1},
                   end : ?{k = 0}. 1}
\end{lstlisting}
The underlying algorithmic idea is to set up a chain of filters that
mark multiples of all the prime number we produce as composite.  So
when the output stream has proclaimed $p$ numbers to be prime, there
will be $ap+b$ active processes small constants $a$ and $b$.  This
algorithm and its implementation has a good amount of parallelism
since all the filters can operate concurrently.  Moreover, we do not
need to perform any multiplication or division, we just set up some
cyclic counters along the way.

At one end of the chain is a process $\m{candidates}$ that produces
a sequence of bits $\mb{prime}$, because initially all natural
numbers starting at 2 are candidates to be primes.
\begin{lstlisting}
decl candidates{n} : (x : nat{n}) |- (s : stream{n})
proc s <- candidates{n} <- x =
  case x ( succ => s.prime ;
                   s <- candidates{n-1} <- x
         | zero => wait x ;
                   s.end ; close s )
\end{lstlisting}
The candidates are fed into a process $\m{head}$ that remains at the
other end of the stream, after all multiples of primes so far have
been marked as $\mb{composite}$.  The bit $\mb{prime}$ it sees must
indeed be a prime number, so it sets up a filter for multiples of the
current list index (which is $x : \m{nat}$).  If the number is
composite, that information is simply passed on.

Two interesting aspects of this process are that the current list
index must be duplicated to set up the filter which contains a cyclic
counter.  The other is that the sum $n+k$ of the current list index
$n$ and the length of the remaining stream $k$ must be invariant.  We
express this by passing in a redundant argument, constrained to be
equal to the sum which never changes.  This would catch errors if we
incorrectly forgot to increment the list index.
\begin{lstlisting}
decl head{k}{n}{kn|kn = k+n} :
(t : stream{k}) (x : nat{n}) |- (s : stream{k})

proc s <- head{k}{n}{kn} <- t x =
  case t
    ( prime => s.prime ; % not divisible: new prime
               z <- zero <- ;
               xx <- dup{n} <- x ;
               x1 <- recv xx ;
               x2 <- recv xx ;
               wait xx ;
               f <- filter{k-1}{n}{0}{n} <- t x1 z ;
               x' <- succ{n} <- x2 ;
               s <- head{k-1}{n+1}{kn} <- f x'
    | composite => s.composite ;
                   x' <- succ{n} <- x ;
                   s <- head{k-1}{n+1}{kn} <- t x'
    | end => wait t ;
             u <- drop{n} <- x ; wait u ;
             s.end ; close s )
\end{lstlisting}
Finally, the filter process $\m{filter}$ itself.  It implements a
cyclic counter as two natural numbers $c$ and $d$ whose sum remains
invariant.  When $c$ reaches $0$ we know we have reach a multiple of
the original $d+1$, we mark the position in the stream as
$\mb{composite}$ and reinitialize the counter $c$ with the value of
$d$ and set $d$ to $0$.  When the stream is empty (implying $k = 0$)
we have to deallocate the counters. Figure~\ref{fig:prime_sieve}
shows the process declaration and definition.
\begin{figure*}
\begin{lstlisting}
decl filter{k}{n1}{n2}{n|n = n1+n2} : (t : stream{k}) (c : nat{n1}) (d : nat{n2}) |- (s : stream{k})
proc s <- filter{k}{n1}{n2}{n} <- t c d =
  case t ( prime => case c ( succ => s.prime ; % not divisible by n
                             d' <- succ{n2} <- d ;
                             s <- filter{k-1}{n1-1}{n2+1}{n} <- t c d'
                           | zero => wait c ; % divisible by n: not prime
                             s.composite ;
                             z <- zero <- ;
                             s <- filter{k-1}{n2}{0}{n} <- t d z ) % cyclic loop
         | composite => s.composite ;  % already composite
                        case c ( succ => d' <- succ{n2} <- d ;
                                         s <- filter{k-1}{n1-1}{n2+1}{n} <- t c d'
                               | zero => wait c ;
                                 z <- zero <- ;
                                 s <- filter{k-1}{n2}{0}{n} <- t d z )
         | end => wait t ;
                  u <- drop{n1} <- c ; wait u ;
                  u <- drop{n2} <- d ; wait u ;
                  s.end ; close s )

\end{lstlisting}
\caption{Filtering prime numbers}
\label{fig:prime_sieve}
\end{figure*}

\subsection{Trie}
We illustrate the data structure of a trie to maintain multisets of
natural numbers.  Remarkably, the data structure is linear although in
one case we have to deallocate an integer.  There is a fair amount of
parallelism since consecutive requests to insert numbers into the trie
can be carried out concurrently.  We also obtain a good
characterization of the necessary work---the data structure is quite
efficient (in theoretical terms).

\paragraph{A Binary Counter}
We start with a binary counter that can receive $\mb{inc}$ messages to
increment the counter and $\mb{val}$ to retrieve the current value of
the counter.  For the latter, we use binary numbers where each bit
carries potential $p$.  When retrieving the number each bit carries
potential $0$, for simplicity.
\begin{lstlisting}
type bin{n}{p} =
  +{ b0 : ?{n > 0}. ?k. ?{n = 2*k}. |{p}> bin{k}{p},
     b1 : ?{n > 0}. ?k. ?{n = 2*k+1}. |{p}> bin{k}{p},
     e : ?{n = 0}. 1 }

type ctr{n} =
  &{ inc : <{3}| ctr{n+1},
     val : <{2}| bin{n}{0} } 
     % potential 0 here for simplicity
\end{lstlisting}
An increment of the binary counter has $3$ units of amortized cost.
$1$ of these units accounts for a future carry, and $2$ units account
for the cost of converting the counter to a sequence of bits should
the value be required.  Note that amortized analysis is necessary
here, because in the worst case a single increment message could
require $\mathrm{log}(n)$ carry bits to be propagated through the
number.

The implementation of the counter is similar to the implementation of
a queue.  It is a chain of processes $\m{bit0}$ and $\m{bit1}$, each
holding a single bit, terminated by an $\m{empty}$ process.  When
converted to a message sequence of bits of type $\m{bin}$, $\m{bit0}$
corresponds to $\mb{b0}$, $\m{bit1}$ to $\mb{b1}$, and $\m{empty}$ to
$\mb{e}$.
\begin{lstlisting}
decl empty :
. |-  (c : ctr{0})

decl bit0{n|n > 0} :
(d : ctr{n}) |{2}- (c : ctr{2*n})

decl bit1{n} :
(d : ctr{n}) |{3}- (c : ctr{2*n+1})

proc c <- empty <- =
  case c ( inc => c0 <- empty <- ;
                  c <- bit1{0} <- c0
         | val => c.e ; close c )

proc c <- bit0{n} <- d =
  case c ( inc => c <- bit1{n} <- d
         | val => c.b0 ; send c {n};
                  d.val ; c <- d )

proc c <- bit1{n} <- d =
  case c ( inc => d.inc ;
                  c <- bit0{n+1} <- d
         | val => c.b1 ; send c {n} ;
                  d.val ; c <- d )
\end{lstlisting}
The type-checker was useful in this example because it discovered a
bug in an earlier, unindexed version of the code: $\m{bit0}[n]$ must
require $n > 0$, otherwise its response to a value request might have
leading zeros.  We also see that each $\m{bit0}$ process carries a
potential of $2$ (to process a future $\mb{val}$ request) and each
$\m{bit1}$ process carries a potential of $3$ (to also be able to pass
on a carry as an $\mb{inc}$ message when incremented).

\paragraph{The Trie Interface}
A trie is represented by the type $\m{trie}[n]$ where $n$ is the
number of elements in the current multiset.  When inserting a number
it becomes a trie $\m{trie}[n+1]$.  When we delete a number $x$ from
the trie we delete all copies of $x$ and return its multiplicity.  For
example, if $x$ was not in the trie at all, we respond with 0; if $x$
was added to the trie 3 times we respond with 3.  If $m$ is the
multiplicity of the number, then after deletion the trie will have
$\m{trie}[n-m]$ elements.  This requires the constraint that
$m \leq n$: the multiplicity of an element cannot be greater than the
total number of elements in the multiset.

When inserting a binary number into the trie that number can be of any
value.  Therefore, we must pass the index $k$ representing that value,
which is represented by a universal quantifier in the type.
Conversely, when responding we need to return the unique binary number
$m$ which is of course not known statically and therefore is an
existential quantifier.

The way we insert the binary number is starting at the root with the
least significant bit and recursively insert the number into the left
or right subtrie, depending on whether the bit is $\mb{b0}$ or
$\mb{b1}$.  When we reach the end of the sequence of bits ($\mb{e}$)
we increase the multiplicity at the leaf we have reached.  This relies
on the uniques of representation, that is, our binary numbers may not
have any leading zeros (which is enforced by the indices, as explained
in an earlier example in Appendix~\ref{subsec:binary}).  As we traverse the
trie, we need to construct new intermediate nodes in case we encounter
a leaf.  It turns out these operations require 4 messages per bit, so
the input number should have potential of 4 per bit.  For deletion, it
turns out we need one more because we need to communicate the answer
back to the client, so 5 units per bit.  For simplicity, we therefore
uniformly require 5 units of potential per bit when adding a number to
the trie and ``burn'' the extra unit during insertion.
\begin{sill}
  $\m{trie}[n] = $\\
  $\echoice{$ \= $\mb{ins} : \tgetpot{\tforall{k} \m{bin}[k,5]
      \lolli \m{trie}[n+1]}{4},$ \\
  \> $\mb{del} : \tgetpot{\tforall{k} \m{bin}[k,5] \lolli
      \texists{m} \tassert{m \leq n}$\\
  \> \hspace{12em} $\m{bin}[m,0] \tensor \m{trie}[n-m]}{5}\,}$
\end{sill}
\begin{lstlisting}
type trie{n} =
  &{ ins : <{4}| !k. bin{k}{5} -o trie{n+1},
     del : <{5}| !k. bin{k}{5} -o ?m. ?{m <= n}.
                              bin{m}{0} * trie{n-m} }
\end{lstlisting}

\paragraph{Trie Implementation}
We have two kinds of nodes: leaf nodes (process $\m{leaf}[0]$) not holding
any elements and element nodes (process $\m{nodes}[n_0,m,n_1]$) representing an element
of multiplicity $m$ with $n_0$ and $n_1$ elements in the left and
right subtries, respectively.  A node therefore has type
$\m{trie}[n_0+m+n_1]$.  Neither process carries any potential.
\begin{lstlisting}
decl leaf : . |- (t : trie{0})
decl node{n0}{m}{n1} :
(l : trie{n0}) (c : ctr{m}) (r : trie{n1}) |-
      (t : trie{n0+m+n1})
\end{lstlisting}
The code then is a bit tricky in the details (nice to have a
type checker verifying both indices and potential!) but not
conceptually difficult.  We present it here without further
comments.
\begin{lstlisting}
proc t <- leaf <- =
  case t
    ( ins => {k} <- recv t ;
             x <- recv t ;
             case x
              ( b0 => {k'} <- recv x ;
                      l <- leaf <- ;
                      c0 <- empty <- ;
                      r <- leaf <- ;
                      l.ins ; send l {k'} ; send l x ;
                      t <- node{1}{0}{0} <- l c0 r
              | b1 => {k'} <- recv x ;
                      l <- leaf <- ;
                      c0 <- empty <- ;
                      r <- leaf <- ;
                      r.ins ; send r {k'} ; send r x ;
                      t <- node{0}{0}{1} <- l c0 r
              | e => wait x ;
                     l <- leaf <- ;
                     c0 <- empty <- ;
                     c0.inc ;
                     r <- leaf <- ;
                     t <- node{0}{1}{0} <- l c0 r )
    | del => {k} <- recv t ;
             x <- recv t ;
             u <- drop{k}{5} <- x ; wait u ;
             send t {0} ;
             c0 <- empty <- ;
             c0.val ;
             send t c0 ;
             t <- leaf <-
    )

proc t <- node{n0}{m}{n1} <- l c r =
  case t
    ( ins => {k} <- recv t ;
             x <- recv t ;
             case x
              ( b0 => {k'} <- recv x ;
                      l.ins ; send l {k'} ; send l x ;
                      t <- node{n0+1}{m}{n1} <- l c r
              | b1 => {k'} <- recv x ;
                      r.ins ;
                      send r {k'} ;
                      send r x ;
                      t <- node{n0}{m}{n1+1} <- l c r
              | e => wait x ;
                     c.inc ;
                     t <- node{n0}{m+1}{n1} <- l c r)
    | del => {k} <- recv t ;
             x <- recv t ;
             case x
              ( b0 => {k'} <- recv x ;
                      l.del ; send l {k'} ; send l x ;
                      {m1} <- recv l ;
                      a <- recv l ; send t {m1} ;
                      send t a ;
                      t<- node{n0-m1}{m}{n1} <- l c r
              | b1 => {k'} <- recv x ;
                      r.del ; send r {k'} ; send r x ;
                      {m2} <- recv r ;
                      a <- recv r ;
                      send t {m2} ; send t a ;
                      t<- node{n0}{m}{n1-m2} <- l c r
              | e => wait x ;
                     send t {m} ;
                     c.val ; send t c ;
                     c0 <- empty <- ;
                     t <- node{n0}{0}{n1} <- l c0 r
              )
    )
\end{lstlisting}

\clearpage

\section{Type Equality}
\label{app:tpeq}
\paragraph{Syntax}
The types follow the syntax below; $n$ is a variable, and
$i$ is a constant.
\[
  \begin{array}{lrcl}
    \mbox{Types} & A & ::= & \ichoice{\ell : A}_{\ell \in L}
    \mid \echoice{\ell : A}_{\ell \in L} \\
                 & & \mid & A \tensor A \mid A \lolli A \mid \one \mid V \indv{e} \\
                 & & \mid & \tassert{\phi} A \mid \tassume{\phi} A
                            \mid \texists{n} A \mid \tforall{n} A \\
    \mbox{Arith. Exps.} & e & ::= & i \mid e + e \mid e - e \mid i \times e \mid n \\
    \mbox{Arith. Props.} &
    \phi & ::= & e = e \mid e > e \mid \top \mid \bot
                 \mid \phi \land \phi \\
    & & \mid &  \phi \lor \phi \mid \lnot \phi \mid \texists{n}\phi \mid \tforall{n} \phi
  \end{array}
\]

\subsection{Type Equality Definition}

\begin{definition}\label{def:app_unfold}
  For a type $A$, define $\unfold{A}$ in the presence of
  signature $\Sg$:
  \begin{mathpar}
    \infer[\m{def}]
    {\unfold{V \indv{e}} = A[\overline{e}/\overline{v}]}
    {V \indv{v} = A \in \Sg}
    \and
    \infer[\m{str}]
    {\unfold{A} = A}
    {A \not= V \indv{e}}
  \end{mathpar}
\end{definition}

\begin{definition}\label{def:app_rel}
  Let $\mi{Type}$ be the set of all ground session types. A relation $\rel
  \subseteq \mi{Type} \times \mi{Type}$ is a type simulation if $(A, B) \in
  \rel$ implies the following conditions:
  \begin{itemize}
    \item If $\unfold{A} = \ichoice{\ell : A_\ell}_{\ell \in L}$, then $\unfold{B} =
    \ichoice{\ell : B_\ell}_{\ell \in L}$ and $(A_\ell, B_\ell) \in \rel$ for
    all $\ell \in L$.

    \item If $\unfold{A} = \echoice{\ell : A_\ell}_{\ell \in L}$, then $\unfold{B} =
    \echoice{\ell : B_\ell}_{\ell \in L}$ and $(A_\ell, B_\ell) \in \rel$ for
    all $\ell \in L$.

    \item If $\unfold{A} = A_1 \lolli A_2$, then $\unfold{B} =
    B_1 \lolli B_2$ and $(A_1, B_1) \in \rel$ and
    $(A_2, B_2) \in \rel$.

    \item If $\unfold{A} = A_1 \tensor A_2$, then $\unfold{B} =
    B_1 \tensor B_2$ and $(A_1, B_1) \in \rel$ and
    $(A_2, B_2) \in \rel$.

    \item If $\unfold{A} = \one$, then $\unfold{B} = \one$.

    \item If $\unfold{A} = \tassert{\phi}{A'}$, then $\unfold{B} = \tassert{\psi}{B'}$
    and either $\proves \phi$, $\proves \psi$, and
    $(A', B') \in \rel$ or $\proves \lnot \phi$ and $\proves \lnot \psi$.

    \item If $\unfold{A} = \tassume{\phi}{A'}$, then $\unfold{B} = \tassume{\psi}{B'}$
    and either $\proves \phi$, $\proves \psi$, and
    $(A', B') \in \rel$, or $\proves \lnot \phi$ and $\proves \lnot \psi$.

    \item If $\unfold{A} = \texists{m} A'$, then $\unfold{B} = \texists{n} B'$
    and for all $i \in \mathbb{N}$, $(A'[i/m], B'[i/n]) \in \rel$.

    \item If $\unfold{A} = \tforall{m} A'$, then $\unfold{B} = \tforall{n} B'$
    and for all $i \in \mathbb{N}$, $(A'[i/m], B'[i/n]) \in \rel$.
  \end{itemize}

\end{definition}

\begin{definition}\label{def:app_decl_releq}
  The type equality relation $\equiv_\rel$ is defined by
  $A \equiv_\rel B$ iff $\rel$ is a type simulation and $(A, B) \in \rel$.
\end{definition}

\begin{definition}\label{def:app_decl_eq}
  The coinductive type equality definition $\equiv$ is defined by
  $A \equiv B$ iff there exists a relation $\rel$ such that $A \equiv_\rel B$.
\end{definition}

\begin{lemma}\label{lem:refl}
  The relation $\equiv$ is reflexive.
\end{lemma}

\begin{proof}
  Consider the reflexive relation $\rel_r$ such that $(A, A) \in \rel_r$
  for any $A$. To establish $\rel_r$ as a type simulation, consider
  a representative case. If $\unfold{A} = \tassert{\phi}{A'}$, then the same holds
  and either $\proves \phi = \phi$ and $(A', A') \in \rel_r$, by definition of $\rel_r$
  or $\proves \lnot \phi = \lnot \psi$. If $\unfold{A} = \texists{n} A'$,
  then the same holds and $(A'[i/n], A'[i/n]) \in \rel_r$ by definition.
\end{proof}

\begin{lemma}\label{lem:symm}
  The relation $\equiv$ is symmetric.
\end{lemma}

\begin{proof}
  Consider $A \equiv B$. There exists a type simulation $\rel_1$
  such that $(A, B) \in \rel_1$. Define $\rel_2$ as follows.
  \[
    \rel_2 = \{(B, A) \mid (A, B) \in \rel_1\}
  \]
  It is easy to see that $\rel_2$ is a type simulation since all
  invariants in Definition~\ref{def:app_rel} are symmetric.
\end{proof}

\begin{lemma}\label{lem:trans}
  The relation $\equiv$ is transitive.
\end{lemma}

\begin{proof}
  Suppose $A_1 \equiv A_2$ and $A_2 \equiv A_3$.
  Thus, there exist type simulations $\rel_1$ and $\rel_2$ such that
  $(A_1, A_2) \in \rel_1$ and $(A_2, A_3) \in \rel_2$.
  Denoting relational composition by $\rel_1 \cdot \rel_2$,
  \begin{eqnarray*}
    \rel & = & (\rel_1 \cdot \rel_2) \\
    & = & \{(A, C) \mid \exists \; B. \;
    (A, B) \in \rel_1, (B, C) \in \rel_2\}
  \end{eqnarray*}
  Clearly, $(A_1, A_3) \in \rel$. To establish $\rel$ as a type
  simulation, consider $(A, C) \in \rel$ and the case $A = \tassert{\phi_1}{A'}$. Since
  $(A, B) \in \rel_1$ for some $B$, we get that $B = \tassert{\phi_2}{B'}$.
  As a first case, assume $\proves \phi_1$ and $\proves \phi_2$ and $(A', B') \in \rel_1$.
  Similarly, since $(B, C) \in \rel_2$, we get that $C = \tassert{\phi_3}{C'}$.
  And since $\proves \phi_2$, we get that $\proves \phi_3$ and $(B', C') \in \rel_2$.
  Using the definition of $\rel$, we get that $(A', C') \in \rel$. Thus,
  we have $A = \tassert{\phi_1}{A'}$ and $C = \tassert{\phi_3}{C'}$ and
  and $\proves \phi_1$ and $\proves \phi_3$ and $(A',C') \in \rel$.
  Therefore, $\rel$ is a type simulation. A similar argument holds
  if $\proves \lnot \phi_1$. 
\end{proof}

\subsection{Type Equality Algorithm}
We start with defining a judgment for defining type equality. The
judgment is defined as $\vars \semi \cons \semi \G
\vdash A \equiv B$.
$\Sg$ defines the signature containing the type definitions.
$\vars$ stores the list of free variables in $A$ and $B$, and
$\cons$ stores the constraints that the variables in $\vars$
satisfy. Finally, $\G$ stores the equality constraints that we are
collecting while following the algorithm recursively. The algorithm
is initiated with an empty $\G$.

We perform a simple type transformation that assigns a name to
each intermediate type expression. After the transformation, the type
grammar becomes
\[
\begin{array}{rcl}
  A & ::= & \ichoice{\ell : T}_{\ell \in L}
  \mid \echoice{\ell : T}_{\ell \in I}
  \mid T \lolli T \mid T \tensor T \mid \one \\
  & \mid & \tassert{\phi}{T} \mid \tassume{\phi}{T}
  \mid \texists{n} T \mid \tforall{n} T \\
  T & ::= & V\indv{e}
\end{array}
\]
The type equality algorithm is performed on transformed type expressions.
This simplifies the algorithm and improves its completeness. Figure
\ref{fig:tpeq_rules} presents the algorithmic rules for type equality.

\begin{figure}
\begin{mathpar}
  \infer[\oplus]
  {\vars \semi \cons \semi \G \vdash
  \ichoice{l : T_l}_{l \in L} \equiv \ichoice{l : U_l}_{l \in L}}
  {\vars \semi \cons \semi \G \vdash
  T_l \equiv U_l \quad (\forall l \in L)}
  \and
  \infer[\with]
  {\vars \semi \cons \semi \G \vdash
  \echoice{l : T_l}_{l \in L} \equiv \echoice{l : U_l}_{l \in L}}
  {\vars \semi \cons \semi \G \vdash
  T_l \equiv U_l \quad (\forall l \in L)}
  \and
  \infer[\tensor]
  {\vars \semi \cons \semi \G \vdash
  T_1 \tensor T_2 \equiv U_1 \tensor U_2}
  {\vars \semi \cons \semi \G \vdash T_1 \equiv U_1 \qquad
  \vars \semi \cons \semi \G \vdash T_2 \equiv U_2}
  \and
  \infer[\lolli]
  {\vars \semi \cons \semi \G \vdash
  T_1 \lolli T_2 \equiv U_1 \lolli U_2}
  {\vars \semi \cons \semi \G \vdash T_1 \equiv U_1 \qquad
  \vars \semi \cons \semi \G \vdash T_2 \equiv U_2}
  \and
  \infer[\one]
  {\vars \semi \cons \semi \G \vdash \one \equiv \one}
  {}
  \and
  \infer[\tassertop]
  {\vars \semi \cons \semi \G \vdash
  \tassert{\phi}{T} \equiv \tassert{\psi}{U}}
  {\vars \semi \cons \proves \phi \leftrightarrow \psi \qquad
  \vars \semi \cons \wedge \phi \semi \G \vdash T \equiv U}
  \and
  \infer[\tassumeop]
  {\vars \semi \cons \semi \G \vdash
  \tassume{\phi}{T} \equiv \tassume{\psi}{U}}
  {\vars \semi \cons \proves \phi \leftrightarrow \psi \qquad
  \vars \semi \cons \wedge \phi \semi \G \vdash T \equiv U}
  \and
  \infer[\exists]
  {\vars \semi \cons \semi \G \vdash
  \texists{m}{T} \equiv \texists{n}{U}}
  {\vars, k \semi \cons \semi \G \vdash
  T[k/m] \equiv U[k/n]}
  \and
  \infer[\forall]
  {\vars \semi \cons \semi \G \vdash
  \tforall{m}{T} \equiv \tforall{n}{U}}
  {\vars, k \semi \cons \semi \G \vdash
  T[k/m] \equiv U[k/n]}
  \and
  \infer[\bot]
  {\vars \semi \cons \semi \G \vdash
  T \equiv U}
  {\vars \semi \cons \proves \bot}
  \and
  \inferrule*[right=$\m{def}$]
  {
    \forall \vars'. \; \cons' \Rightarrow V_1 \indv{E_1} \equiv V_2 \indv {E_2} \in \G \\
    \vars \semi \cons \proves \exists \vars'. \; \cons' \land \indv{E_1}
    = \indv{e_1} \land \indv{E_2} = \indv{e_2}
  }
  {
    \vars \semi \cons \semi \G \vdash V_1 \indv{e_1} \equiv V_2 \indv{e_2}
  }
  \and
  \inferrule*[right = $\m{expd}$]
  {V_1 \indv{v_1} = A \in \Sg \and
  V_2 \indv{v_2} = B \in \Sg \\\\
  \gamma = \forall \vars. \; \cons \Rightarrow V_1 \indv{e_1} \equiv V_2 \indv{e_2}\\\\
  \vars \semi \cons \semi \G, \gamma
  \vdash A[\overline{e_1} / \overline{v_1}] \equiv
  B[\overline{e_2}/\overline{v_2}]}
  {\vars \semi \cons \semi \G \vdash
  V_1 \indv{e_1} \equiv V_2 \indv{e_2}}
\end{mathpar}
\caption{Algorithmic Rules for Type Equality}
\label{fig:tpeq_rules}
\end{figure}

\begin{lemma}\label{lem:alt}
  Consider a goal $\vars \semi \cons \semi \G \vdash A \equiv B$.
  Either $A$ and $B$ are both structural or both type names.
\end{lemma}

\begin{proof}
  By induction on the algorithmic type equality judgment, and the fact that
  type definitions are contractive, and the intermediate transformation where
  every continuation is replaced by a type name.
\end{proof}

\subsection{Soundness of Type Equality}

\begin{definition}\label{def:app_algo_tpeq}
  The notation $\forall \vars. \cons \Rightarrow A \equiv_\rel B$ denotes that two types
  $A$ and $B$ quantified over $\vars$ are related in $\rel$ under the constraint $\cons$.
  This holds when for all ground substitutions $\sigma$ over $\vars$
  satisfying $\cons$, $A[\sigma] \equiv_\rel B[\sigma]$.
\end{definition}

\begin{definition}\label{def:decl_tpeq}
  The declarative definition for equality, denoted by $\forall \vars. \cons \Rightarrow
  A \equiv B$ holds when there exists a type simulation $\rel$ such that
  $\forall \vars. \cons \Rightarrow A \equiv_\rel B$.
\end{definition}

\begin{lemma}\label{lem:app_gen_sim}
  Suppose $\forall \vars'. \cons' \Rightarrow V_1 \indv{e_1'} \equiv_\rel V_2 \indv{e_2'}$ holds.
  And suppose, $\vars \semi \cons \proves \exists \vars'. \cons' \wedge
  \overline{e_1'} = \overline{e_1} \wedge \overline{e_2'} = \overline{e_2}$
  for some $\vars, \cons, \indv{e_1}, \indv{e_2}$. Then, $\forall \vars. \cons
  \Rightarrow V_1 \indv{e_1} \equiv_\rel V_2 \indv{e_2}$ holds.
\end{lemma}

\begin{proof}
  To prove $\forall \vars. \cons \Rightarrow V_1 \indv{e_1} \equiv_\rel V_2 \indv{e_2}$, it
  is sufficient to show that $V_1 \indv{e_1[\sigma]} \equiv_\rel V_2 \indv{e_2[\sigma]}$
  for any substitution $\sigma : \vars$ such that $\cdot \proves \cons[\sigma]$.
  Applying this substitution to $\vars \semi \cons \proves \exists \vars'.
  \cons' \wedge \overline{e_1'} = \overline{e_1} \wedge \overline{e_2'} = \overline{e_2}$,
  we get that $\exists \vars'. \cons' \wedge \overline{e_1'} =
  \overline{e_1[\sigma]} \wedge \overline{e_2'} = \overline{e_2[\sigma]}$
  since $\cons[\sigma]$ is true. Thus, there exists $\sigma' : \vars'$ such that
  $\cdot \proves \cons'[\sigma']$ holds, and $\overline{e_1'[\sigma']} =
  \overline{e_1[\sigma]}$ and $\overline{e_2'[\sigma']} = \overline{e_2[\sigma]}$.
  And since $\forall \vars'. \cons' \Rightarrow V_1 \indv{e_1'} \equiv_\rel V_2 \indv{e_2'}$,
  we can use substitution $\sigma'$ to get that $V_1 \indv{e_1'[\sigma']} \equiv_\rel
  V_2 \indv{e_2'[\sigma']}$. This implies that $V_1 \indv{e_1[\sigma]} \equiv_\rel V_2
  \indv{e_2[\sigma]}$ since $\overline{e_1'[\sigma']} =
  \overline{e_1[\sigma]}$ and $\overline{e_2'[\sigma']} = \overline{e_2[\sigma]}$.
\end{proof}

\begin{definition}
  Given $\vars \semi \cons \semi \G \vdash A \equiv B$, define
  the set of conclusions $\mathcal{S}$ by the judgment $\concs{\vars
  \semi \cons \semi \G \vdash A \equiv B}{\mathcal{S}}$. This judgment
  simply collects the set of conclusions in a derivation tree. It is
  defined using the rules in Figure~\ref{fig:concs-judgment}.

  \begin{figure*}
  \centering
  \begin{mathpar}
    \infer[\oplus]
    {\concs{\Sg \semi \vars \semi \cons \semi \G \vdash
    \ichoice{l : T_l}_{l \in L} \equiv \ichoice{l : U_l}_{l \in L}}{\mathcal{S},
    (\forall \vars. \cons \Rightarrow \ichoice{l : T_l}_{l \in L} \equiv \ichoice{l :
    U_l}_{l \in L})}}
    {\concs{\Sg \semi \vars \semi \cons \semi \G \vdash
    T_l \equiv U_l}{\mathcal{S}_l} \quad (\forall l \in L)
    \qquad \mathcal{S} = \bigcup_{l \in L} \mathcal{S}_l}
    \and
    \infer[\with]
    {\concs{\Sg \semi \vars \semi \cons \semi \G \vdash
    \echoice{l : T_l}_{l \in L} \equiv \echoice{l : U_l}_{l \in L}}{\mathcal{S},
    (\forall \vars. \cons \Rightarrow \echoice{l : T_l}_{l \in L} \equiv \echoice{l :
    U_l}_{l \in L})}}
    {\concs{\Sg \semi \vars \semi \cons \semi \G \vdash
    T_l \equiv U_l}{\mathcal{S}_l} \quad (\forall l \in L)
    \qquad \mathcal{S} = \bigcup_{l \in L} \mathcal{S}_l}
    \and
    \infer[\tensor]
    {\concs{\Sg \semi \vars \semi \cons \semi \G \vdash
    T_1 \tensor T_2 \equiv U_1 \tensor U_2}{\mathcal{S},
    (\forall \vars. \cons \Rightarrow T_1 \tensor T_2 \equiv U_1 \tensor U_2)}}
    {\concs{\Sg \semi \vars \semi \cons \semi \G \vdash T_1 \equiv U_1}
    {\mathcal{S}_1} \qquad
    \concs{\Sg \semi \vars \semi \cons \semi \G \vdash T_2 \equiv U_2}
    {\mathcal{S}_2} \qquad \mathcal{S} = \mathcal{S}_1 \cup \mathcal{S}_2}
    \and
    \infer[\lolli]
    {\concs{\Sg \semi \vars \semi \cons \semi \G \vdash
    T_1 \lolli T_2 \equiv U_1 \lolli U_2}{\mathcal{S},
    (\forall \vars. \cons \Rightarrow T_1 \lolli T_2 \equiv U_1 \lolli U_2)}}
    {\concs{\Sg \semi \vars \semi \cons \semi \G \vdash T_1 \equiv U_1}
    {\mathcal{S}_1} \qquad
    \concs{\Sg \semi \vars \semi \cons \semi \G \vdash T_2 \equiv U_2}
    {\mathcal{S}_2} \qquad \mathcal{S} = \mathcal{S}_1 \cup \mathcal{S}_2}
    \and
    \infer[\one]
    {\concs{\Sg \semi \vars \semi \cons \semi \G \vdash \one \equiv \one}
    {\{(\vars \semi \cons \semi \one \equiv \one)\}}}
    {}
    \and
    \infer[\tassertop]
    {\concs{\Sg \semi \vars \semi \cons \semi \G \vdash
    \tassert{\phi}{T} \equiv \tassert{\psi}{U}}{\mathcal{S},
    (\forall \vars. \cons \Rightarrow \tassert{\phi}{T} \equiv \tassert{\psi}{U})}}
    {\vars \semi \cons \proves \phi \leftrightarrow \psi \qquad
    \concs{\Sg \semi \vars \semi \cons \wedge \phi \semi \G \vdash T \equiv U}
    {\mathcal{S}}}
    \and
    \infer[\tassumeop]
    {\concs{\Sg \semi \vars \semi \cons \semi \G \vdash
    \tassume{\phi}{T} \equiv \tassume{\psi}{U}}{\mathcal{S},
    (\forall \vars. \cons \Rightarrow \tassume{\phi}{T} \equiv \tassume{\psi}{U})}}
    {\vars \semi \cons \proves \phi \leftrightarrow \psi \qquad
    \concs{\Sg \semi \vars \semi \cons \wedge \phi \semi \G \vdash T \equiv U}
    {\mathcal{S}}}
    \and
    \infer[\exists]
    {\concs{\Sg \semi \vars \semi \cons \semi \G \vdash
    \texists{m} T \equiv \texists{n} U}{\mathcal{S},
    (\forall \vars. \cons \Rightarrow \texists{m} T \equiv \texists{n} U)}}
    {\concs{\Sg \semi \vars,k \semi \cons \semi \G \vdash T[k/m] \equiv U[k/n]}
    {\mathcal{S}}}
    \and
    \infer[\forall]
    {\concs{\Sg \semi \vars \semi \cons \semi \G \vdash
    \tforall{m} T \equiv \tforall{n} U}{\mathcal{S},
    (\forall \vars. \cons \Rightarrow \tforall{m} T \equiv \tforall{n} U)}}
    {\concs{\Sg \semi \vars,k \semi \cons \semi \G \vdash T[k/m] \equiv U[k/n]}
    {\mathcal{S}}}
    \and
    \infer[\m{def}]
    {\concs{\Sg \semi \vars \semi \cons \semi \G \vdash
    V_1 \indv{e_1} \equiv V_2 \indv{e_2}}{\{\}}}
    {\forall \vars'. \; \cons' \Rightarrow V_1 \indv{E_1} \equiv V_2
    \indv {E_2} \in \G \qquad
    \forall \vars. \cons \rightarrow \exists \vars'. \cons' \wedge
    \overline{E_1} = \overline{e_1} \wedge \overline{E_2} = \overline{e_2}}
    \and
    \inferrule*[right = $\m{exp}$]
    {V_1 \indv{v_1} = A \in \Sg \qquad
    V_2 \indv{v_2} = B \in \Sg \qquad
    \gamma = \forall \vars. \; \cons \Rightarrow V_1 \indv{e_1} \equiv V_2 \indv{e_2} \\\\
    \concs{\Sg \semi \vars \semi \cons \semi \G, \gamma
    \vdash A[\overline{e_1} / \overline{v_1}] \equiv
    B[\overline{e_2}/\overline{v_2}]}{\mathcal{S}}}
    {\concs{\Sg \semi \vars \semi \cons \semi \G \vdash
    V_1 \indv{e_1} \equiv V_2 \indv{e_2}}{\mathcal{S},
    (\forall \vars. \cons \Rightarrow V_1 \indv{e_1} \equiv V_2 \indv{e_2})}}
  \end{mathpar}
  \caption{Judgment to collect all the equality conclusions}
  \label{fig:concs-judgment}
\end{figure*}
\end{definition}

\begin{theorem}
  If $\vars \semi \cons \semi \cdot \vdash A \equiv B$,
  then $\forall \vars. \cons \Rightarrow A \equiv B$.
\end{theorem}

\begin{proof}
  Let $\mathcal{S}$ be the set such that $\concs{\vars \semi
  \cons \semi \cdot \vdash A \equiv B}{\mathcal{S}}$. Define a relation
  $\rel$ such that
  \begin{eqnarray*}
    \rel = \{(A[\sigma], B[\sigma]) \mid (\forall \vars. \cons \Rightarrow A \equiv B)
    \in \mathcal{S}, \\ \forall \sigma : \vars. \; \cdot \proves \cons[\sigma]\}
  \end{eqnarray*}
  We prove that $\rel$ is a type simulation. Consider $(A[\sigma], B[\sigma]) \in \rel$
  where $(\vars \semi \cons \semi A \equiv B) \in \mathcal{S}$ for some $\sigma : \vars$
  and $\cdot \proves \cons[\sigma]$. We case analyze on the structure of $A[\sigma]$.

  Consider the case where $A = \ichoice{l : A_l}_{l \in L}$. Since $A$ and $B$ are both
  structural, $B = \ichoice{l : B_l}_{l \in L}$. Since $(\forall \vars. \cons \Rightarrow
  A \equiv B) \in \mathcal{S}$, by definition of $\mathcal{S}$, we get $(\forall \vars.
  \cons \Rightarrow A_l \equiv B_l) \in \mathcal{S}$ for all $l \in L$. Then, by the definition
  of $\rel$, we get that $(A_l[\sigma], B_l[\sigma]) \in \rel$. Also,
  $A[\sigma] = \ichoice{l : A_l[\sigma]}_{l \in L}$ and similarly,
  $B[\sigma] = \ichoice{l : B_l[\sigma]}_{l \in L}$.
  Hence, $\rel$ satisfies Definition~\ref{def:app_rel} and is indeed a type simulation.

  Next, consider the case where $A = \tassert{\phi} A'$. Since $A$ and $B$ are both
  structural, $B = \tassert{\psi} B'$. Since $(\forall \vars. \cons \Rightarrow A \equiv B)
  \in \mathcal{S}$, we get that $\vars \semi \cons \proves \phi \leftrightarrow \psi$
  and $(\forall \vars. \cons \land \phi \Rightarrow A' \equiv B') \in \mathcal{S}$. Thus, for
  any substitution $\sigma$ such that $\cdot \proves \cons[\sigma] \land \phi [\sigma]$,
  we get that $(A'[\sigma], B'[\sigma]) \in \rel$. And $A[\sigma] = \tassert{\phi[\sigma]}
  A'[\sigma]$, and $B[\sigma] = \tassert{\psi[\sigma]} B'[\sigma]$. Since $\cdot \proves
  \phi[\sigma]$, we get that $\phi[\sigma] = \true$ and since $\vars \semi \cons \proves
  \phi \leftrightarrow \psi$, we get that $\psi[\sigma] = \true$. Thus, $\rel$ is a
  type simulation.

  Next, consider the case where $A = \texists{m} A'$. Since $A$ and $B$ are both
  structural, $B = \texists{n} B'$. Since $(\forall \vars. \cons \Rightarrow A \equiv B)
  \in \mathcal{S}$, we get that $(\forall \vars,k. \cons \Rightarrow A'[k/m] \equiv B'[k/n])
  \in \mathcal{S}$. Thus, for any substitution $\sigma$ over $\vars, k$ such that
  $\cdot \proves \cons[\sigma]$, we get that $(A'[k/m][\sigma], B'[k/n][\sigma])
  \in \rel$. Since $\sigma$ is over $\vars, k$, we know that $k$ is substituted
  over a ground term, and since $\cons$ does not mention $k$, we know that $k$
  can be substituted with any ground term. Thus, $\forall i \in \mathbb{N}$,
  $(A'[i/k][k/m][\sigma], B'[i/k][k/n][\sigma]) \in \rel$. Thus, $\rel$ is a type
  simulation.

  The only case where a conclusion is not added to $\mathcal{S}$ is the $\m{def}$
  rule. In this case, adding $(\forall \vars. \cons \Rightarrow V_1 \indv{e_1} \equiv V_2
  \indv{e_2})$ is redundant. Because Lemma~\ref{lem:app_gen_sim} states that
  $V_1 \indv{e_1[\sigma]} \equiv_\rel V_2 \indv{e_2[\sigma]}$, this implies
  $(V_1 \indv{e_1[\sigma]}, V_2 \indv{e_2[\sigma]}) \in \rel$. Thus, $\rel$
  is indeed a type simulation.

\end{proof}

\subsection{Undecidability of Type Equality}
We prove the undecidability of type equality by exhibiting a
reduction from an undecidable problem about two counter
machines.

A two counter machine $\tcm$ is given a sequence
of instructions $\iota_1, \iota_2, \ldots, \iota_m$ where each
instruction is one of the following.

\begin{itemize}
\item ``$\inc{c_j} ; \goto \; k$'' (increment counter $j$ by 1 and
go to instruction $k$)
\item ``$\zeroc{c_j} \; \goto \; k : \dec{c_j} ; \goto \; l$'' (if the value
of the counter $j$ is 0, go to instruction $k$, else decrement the
counter by 1 and go to instruction $l$)
\item ``$\halt$'' (stop computation)
\end{itemize}

A configuration $C$ of the machine $\tcm$ is defined as a
triple $(i, c_1, c_2)$, where $i$ denotes the number of the
current instruction and $c_j$'s denote the value of the counters.
A configuration $C'$ is defined as the successor configuration
of $C$, written as $C \mapsto C'$ if $C'$ is the result of
executing the $i$-th instruction on $C$. If $\ins_i = \halt$, then
$C = (i, c_1, c_2)$ has no successor configuration.

The computation of $\tcm$ is the unique maximal sequence
$\rho = \rho(0) \rho(1) \ldots$ such that $\rho(i) \mapsto
\rho(i+1)$ and $\rho(0) = (1, 0, 0)$. Either $\rho$ is infinite,
or ends in $(i, c_1, c_2)$ such that $\ins_i = \halt$.

The \emph{halting problem} refers to determining whether
the computation of a two counter machine $\tcm$ is finite.
Both the halting problem and its dual, the non-halting problem
are undecidable.

We prove the undecidability of the type equality problem by
reducing the non-halting problem of an instance of the two
counter machine $\tcm$ to an instance of the type equality
problem.

Consider the instruction $\ins_i$. It can have one of the three
forms shown above. We will define a recursive type based
on the type of instruction. First, we define a type $T_{\inf} =
\ichoice{\ell : T_{\inf}}$ and $T_{\inf}' =
\ichoice{\ell' : T_{\inf}'}$.

\begin{itemize}
\item Case ($\ins_i = \inc{c_1} ; \goto \; k$) : In this case, we
define the type $T_i [c_1, c_2] = \ichoice{\m{inc}_1 : T_k
[c_1 + 1, c_2]}$ for all $c_1, c_2 \in \mathbb{N}$. We also
define $T_i' [c_1, c_2] = \ichoice{\m{inc}_1 : T_k' [c_1 + 1, c_2]}$
for all $c_1, c_2 \in \mathbb{N}$.

\item Case ($\ins_i = \inc{c_2} ; \goto \; k$) : In this case, we
define the type $T_i [c_1, c_2] = \ichoice{\m{inc}_2 : T_k
[c_1, c_2 + 1]}$ for all $c_1, c_2 \in \mathbb{N}$. We also
define $T_i' [c_1, c_2] = \ichoice{\m{inc}_2 : T_k' [c_1, c_2 + 1]}$
for all $c_1, c_2 \in \mathbb{N}$.

\item Case ($\ins_i = \zeroc{c_1} \; \goto \; k : \dec{c_j} ; \goto \; l$) :
In this case, we define the type $T_i [c_1, c_2] = \ichoice{\m{zero}_1 : 
\tassert{c_1 = 0} T_k[c_1, c_2], \m{dec}_1 : \tassert{c_1 > 0}
T_l [c_1 - 1, c_2]}$ for all $c_1, c_2 \in \mathbb{N}$. We also
define $T_i' [c_1, c_2] = \ichoice{\m{zero}_1 : 
\tassert{c_1 = 0} T_k' [c_1, c_2], \m{dec}_1 : \tassert{c_1 > 0}
T_l' [c_1 - 1, c_2]}$.

\item Case ($\ins_i = \zeroc{c_2} \; \goto \; k : \dec{c_j} ; \goto \; l$) :
In this case, we define the type $T_i [c_1, c_2] = \ichoice{\m{zero}_2 : 
\tassert{c_2 = 0} T_k [c_1, c_2], \m{dec}_2 : \tassert{c_2 > 0}
T_l [c_1, c_2 - 1]}$ for all $c_1, c_2 \in \mathbb{N}$. We also
define $T_i' [c_1, c_2] = \ichoice{\m{zero}_2 : 
\tassert{c_2 = 0} T_k' [c_1, c_2], \m{dec}_2 : \tassert{c_2 > 0}
T_l' [c_1, c_2 - 1]}$ for all $c_1, c_2 \in \mathbb{N}$

\item Case ($\ins_i = \halt$) : In this case, we define
$T_i [c_1, c_2] = T_{inf}$ for all $c_1, c_2 \in \mathbb{N}$. We
also define $T_i' [c_1, c_2] = T_{\inf}'$ for all $c_1, c_2 \in \mathbb{N}$.

\end{itemize}

Suppose, the counter $\tcm$ is initialized in the state $(1, m, n)$.
The type equality question we ask is $T_1 [m, n] \equiv T_1' [m, n]$.
The two types only differ at the halting instruction. If and only if
$\tcm$ halts, the first type sends infinitely many $\ell$ labels
while the second type sends infintely many $\ell'$ labels. While if
$\tcm$ does not halt, the two types produce exactly the same communication
behavior, since the halting instruction is never reached. Hence,
the two types are equal iff $\tcm$ does not halt.

Hence, the type equality problem can be reduced to the
non-halting problem.

\clearpage

\section{Formal Typing Rules}
\label{app:formal}
\subsection{Explicit Type System}
This section formalizes the algorithmic explicit type system, providing a
compact collection of all the rules. Figure~\ref{fig:app-explicit} presents
the algorithmic typing rules, Figure~\ref{fig:opsem} presents the
operational semantics and Figure~\ref{fig:conf-typing} presents the configuration
typing rules.

\subsection{Implicit Type System}
The rules of the implicit system are similar, except they do not change the
process terms for the rules involving proof constraints and potential. The
rest of the rules exactly match the explicit system. Figure~\ref{fig:implicit_rules}
contains the implicit typing rules that differ from the explicit system.
\begin{figure}[H]
\begin{mathpar}
  \infer[\tassertop R]
  {\vars \semi \cons \semi \D \ivdash{q} P :: (x : \tassert{\phi} A)}
  {\vars \semi \cons \proves \phi \and
  \vars \semi \cons \semi \D \ivdash{q} P :: (x : A)}
  \and
  \infer[\tassertop L]
  {\vars \semi \cons \semi \D, (x : \tassert{\phi} A)
  \ivdash{q} Q :: (z : C)}
  {\vars \semi \cons \land \phi \semi \D, (x : A)
  \ivdash{q} Q :: (z : C)}
  \and
  \infer[\tassumeop R]
  {\vars \semi \cons \semi \D \ivdash{q}
  P :: (x : \tassume{\phi} A)}
  {\vars \semi \cons \land \phi \semi \D
  \ivdash{q} P :: (x : A)}
  \and
  \infer[\tassumeop L]
  {\vars \semi \cons \semi \D, (x : \tassume{\phi} A)
  \ivdash{q} Q :: (z : C)}
  {\vars \semi \cons \proves \phi \and
  \cons \semi \D, (x : A) \ivdash{q} Q :: (z : C)}
  \and
  \infer[\paypot R]
  {\vars \semi \cons \semi \D \ivdash{q}
  P :: (x : \tpaypot{A}{r})}
  {\vars \semi \cons \proves q \geq r \and
  \vars \semi \cons \semi \D \ivdash{q-r} P :: (x : A)}
  \and
  \infer[\paypot L]
  {\vars \semi \cons \semi \D, (x : \tpaypot{A}{r}) \ivdash{q}
  Q :: (z : C)}
  {\vars \semi \cons \semi \D, (x : A) \ivdash{q+r} Q :: (z : C)}
  \\
  \infer[\getpot R]
  {\vars \semi \cons \semi \D \ivdash{q}
  P :: (x : \tgetpot{A}{r})}
  {\vars \semi \cons \semi \D \ivdash{q+r} P :: (x : A)}
  \and
  \infer[\getpot L]
  {\vars \semi \cons \semi \D, (x : \tgetpot{A}{r}) \ivdash{q}
  Q :: (z : C)}
  {\vars \semi \cons \proves q \geq r \\
  \vars \semi \cons \semi \D, (x : A) \ivdash{q-r} Q :: (z : C)}
\end{mathpar}
\caption{Implicit Typing Rules}
\label{fig:implicit_rules}
\end{figure}

\begin{figure}
  \begin{mathpar}
    \infer[\m{emp}]
    {(\cdot) \potconf{0} (\cdot) :: (\cdot)}
    {\mathstrut}
    \and
    \infer[\m{compose}]
    {\D \potconf{E + E'}(\config \; \config') :: \D''}
    {\D \potconf{E} \config :: \D' \and \Sg \semi \D' \potconf{E'} \config' :: \D''}
    \and
    \infer[\m{proc}]
    {\D, \D_1 \potconf{p + w}(\proc{x}{w, P}) :: (\D, (x:A) )}
    {\cdot \semi \cdot \semi \D_1 \entailpot{p} P :: (x : A)}
    \and
    \infer[\m{msg}]
    {\D, \D_1 \potconf{p + w}(\msg{x}{w, P}) :: (\D, (x:A) )}
    {\cdot \semi \cdot \semi \D_1 \entailpot{p} P :: (x : A)}
  \end{mathpar}
  \caption{Configuration Typing Rules}
  \label{fig:conf-typing}
\end{figure}
\begin{figure*}
  \begin{mathpar}
    \infer[\m{id}]
    {\vars \semi \cons \semi y : A \entailpot{q} \fwd{x}{y} :: (x : B)}
    {\vars \semi \cons \proves q = 0 \qquad
    \vars \semi \cons \semi \cdot \vdash A \equiv B}
    \and
    \inferrule*[right=$\m{spawn}$]
    {q = p[\overline{e_c}/\overline{v}] \and
    \vars \semi \cons \proves r \geq q \and
    \mathcal{X} \indv{v} : \overline{y_i : A_i} \entailpot{p} x' : A \in \Sg \\\\
    \vars \semi \cons \semi \cdot \vdash \D_1 \equiv_{\alpha}
    \overline{y_i : A_i [e_c/v]} \\\\
    \Sg \semi \vars \semi \cons \semi \D_2, (x :
    A[\overline{e_c} / \overline{v}]) \entailpot{r - q} Q_x :: (z : C)}
    {\Sg \semi \vars \semi \cons \semi \D_1, \D_2 \entailpot{r}
    (\ecut{x}{\mathcal{X} \indv{e_c}}{y}{Q_x}) :: (z : C)}
    \and
    \infer[\oplus R_k]
    {\vars \semi \cons \semi \D \entailpot{q} (\esendl{x}{k} \semi P) ::
    (x : \ichoice{\ell : A_\ell}_{\ell \in L})}
    {\vars \semi \cons \semi \D \entailpot{q} P :: (x : A_k) \quad (k \in L)}
    \and
    \infer[\oplus L]
    {\vars \semi \cons \semi \D, (x : \ichoice{\ell : A_\ell}_{\ell \in L}) \entailpot{q}
    \ecase{x}{\ell}{Q_\ell}_{\ell \in L} :: (z : C)}
    {\vars \semi \cons \semi \D, (x : A_\ell) \entailpot{q} Q_\ell :: (z : C) \quad (\forall \ell \in L)}
    \and
    \infer[\with R]
    {\vars \semi \cons \semi \D \entailpot{q} \ecase{x}{\ell}{P_\ell}_{\ell \in L} ::
    (x : \echoice{\ell : A_\ell}_{\ell \in L})}
    {\vars \semi \cons \semi \D \entailpot{q} P_\ell :: (x : A_\ell) \quad (\forall \ell \in L)}
    \and
    \infer[\with L_k]
    {\vars \semi \cons \semi \D, (x : \echoice{\ell : A_\ell}_{\ell \in L}) \entailpot{q}
    \esendl{x}{k} \semi Q :: (z : C)}
    {\vars \semi \cons \semi \D, (x : A_k) \entailpot{q} Q :: (z : C)}
    \and
    \infer[\lolli R]
    {\vars \semi \cons \semi \D \entailpot{q} (\erecvch{x}{y} \semi P_y) :: (x : A \lolli B)}
    {\vars \semi \cons \semi \D, (y : A) \entailpot{q} P_y :: (x : B)}
    \and
    \infer[\lolli L]
    {\vars \semi \cons \semi \D, (w : A_1), (x : A_2 \lolli B)
    \entailpot{q} (\esendch{x}{w} \semi Q) :: (z : C)}
    {\vars \semi \cons \semi \cdot \vdash A_1 \equiv A_2 \and
    \vars \semi \cons \semi \D, (x : B) \entailpot{q} Q :: (z : C)}
    \and
    \infer[\tensor R]
    {\vars \semi \cons \semi (w : A_1), \D \entailpot{q}
    \esendch{x}{w} \semi P :: (x : A_2 \tensor B)}
    {\vars \semi \cons \semi \cdot \vdash A_1 \equiv A_2 \qquad
    \vars \semi \cons \semi \D \entailpot{q} P :: (x : B)}
    \qquad
    \infer[\tensor L]
    {\vars \semi \cons \semi \D, (x : A \tensor B) \entailpot{q} \erecvch{x}{y} \semi Q_y :: (z : C)}
    {\vars \semi \cons \semi \D, (y : A), (x : B) \entailpot{q} Q_y :: (z : C)}
    \and
    \infer[\one R]
    {\vars \semi \cons \semi \cdot \entailpot{q}
    \eclose{x} :: (x : \one)}
    {\vars \semi \cons \proves q = 0}
    \and
    \infer[\one L]
    {\vars \semi \cons \semi \D, (x : \one) \entailpot{q} \ewait{x} \semi Q :: (z : C)}
    {\vars \semi \cons \semi \D \entailpot{q} Q :: (z : C)}
    \and
    \infer[\exists R]
    {\vars \semi \cons \semi \D \vdash \esendn{x}{e} \semi P :: (x : \texists{n} A)}
    {\vars \semi \cons \vdash e : \m{nat} \and
    \vars \semi \cons \semi \D \vdash P :: (x : A[e/n])}
    \and
    \infer[\exists L]
    {\vars \semi \cons \semi \D, (x : \texists{n} A) \vdash \erecvn{x}{n} \semi Q_n :: (z : C)}
    {\vars,n \semi \cons \semi \D, (x : A) \vdash Q_n :: (z : C)}
    \and
    \infer[\forall R]
    {\vars \semi \cons \semi \D \vdash \erecvn{x}{n} \semi P_n :: (x : \tforall{n} A)}
    {\vars,n \semi \cons \semi \D \vdash P_k :: (x : A)}
    \and
    \infer[\forall L]
    {\vars \semi \cons \semi \D, (x : \tforall{n} A) \vdash \esendn{x}{e} \semi Q :: (z : C)}
    {\vars \semi \cons \vdash e : \m{nat} \and
    \vars \semi \cons \semi \D, (x : A[e/n]) \vdash Q :: (z : C)}
    \and
    \infer[\m{work}]
    {\vars \semi \cons \semi \D \entailpot{q} \ework{c} \semi P :: (x : A)}
    {\vars \semi \cons \proves q \geq c \qquad
    \vars \semi \cons \semi \D \entailpot{q-c} P :: (x : A)}
    \\
    \infer[\paypot R]
    {\vars \semi \cons \semi \D \entailpot{q}
    \epay{x}{r_1} \semi P :: (x : \tpaypot{A}{r_2})}
    {\vars \semi \cons \proves q \geq r_1 = r_2 \qquad
    \vars \semi \cons \semi \D \entailpot{q-r_1} P :: (x : A)}
    \and
    \infer[\paypot L]
    {\vars \semi \cons \semi \D, (x : \tpaypot{A}{r_2}) \entailpot{q}
    \eget{x}{r_1} \semi Q :: (z : C)}
    {\vars \semi \cons \proves r_1 = r_2 \qquad
    \vars \semi \cons \semi \D, (x : A) \entailpot{q+r_1} Q :: (z : C)}
    \\
    \infer[\getpot R]
    {\vars \semi \cons \semi \D \entailpot{q}
    \eget{x}{r_1} \semi P :: (x : \tgetpot{A}{r_2})}
    {\vars \semi \cons \proves r_1 = r_2 \qquad
    \vars \semi \cons \semi \D \entailpot{q+r_1} P :: (x : A)}
    \and
    \infer[\getpot L]
    {\vars \semi \cons \semi \D, (x : \tgetpot{A}{r_2}) \entailpot{q}
    \epay{x}{r_1} \semi Q :: (z : C)}
    {\vars \semi \cons \proves q \geq r_1 = r_2 \qquad
    \vars \semi \cons \semi \D, (x : A) \entailpot{q-r_1} Q :: (z : C)}
    \\
    \infer[\tassertop R]
    {\vars \semi \cons \semi \D \entailpot{q}
    \eassert{x}{\phi} \semi P :: (x : \tassert{\phi'} A)}
    {\vars \semi \cons \proves \phi \qquad
    \vars \semi \cons \proves \phi' \qquad
    \vars \semi \cons \semi \D \entailpot{q} P :: (x : A)}
    \and
    \infer[\tassertop L]
    {\vars \semi \cons \semi \D, (x : \tassert{\phi'} A)
    \entailpot{q} \eassume{x}{\phi} \semi Q :: (z : C)}
    {\vars \semi \cons \land \phi' \proves \phi \qquad
    \vars \semi \cons \land \phi \semi \D, (x : A) \entailpot{q} Q :: (z : C)}
    \\
    \infer[\tassumeop R]
    {\vars \semi \cons \semi \D \entailpot{q}
    \eassume{x}{\phi} \semi P :: (x : \tassume{\phi'} A)}
    {\vars \semi \cons \land \phi' \proves \phi \qquad
    \vars \semi \G \land \phi \semi \D \entailpot{q} P :: (x : A)}
    \and
    \infer[\tassumeop L]
    {\vars \semi \cons \semi \D, (x : \tassume{\phi'} A) \entailpot{q}
    \eassert{x}{\phi} \semi Q :: (z : C)}
    {\vars \semi \cons \proves \phi \qquad
    \vars \semi \cons \proves \phi' \qquad
    \vars \semi \cons \semi \D, (x : A) \entailpot{q} Q :: (z : C)}
    \\
    \infer[\m{unsat}]
    {\vars \semi \cons \semi \D
    \entailpot{q} \imposs :: (x : A)}
    {\vars \semi \cons \proves \bot}
  \end{mathpar}
  \caption{Explicit Typing Rules}
  \label{fig:app-explicit}
\end{figure*}
\begin{figure*}
\centering
\[
\begin{array}{lll}
(\m{id}^+C) & \m{msg}(d, w, M), \m{proc}(c, w', s, c \leftarrow d) \;\mapsto\; \m{msg}(c, w+w', M[c/d]) & \\
(\m{id}^-C) & \m{proc}(c, w, s, c \leftarrow d), \m{msg}(e, w', M(c)) \;\mapsto\; \m{msg}(e, w+w', M(c)[d/c]) & \\[1ex]
(\m{def}C) & \m{proc}(c, w, \ecut{x}{f \indv{e}}{\overline{y}}{Q}) \; \mapsto \; \m{proc}(a, 0, P_f[a/x, \overline{y}/\D_f, \overline{e}/\overline{n}]), \; \m{proc}(c, w, Q[a/x]) & \mbox{($a$ fresh)} \\[1ex]
({\oplus}S) & \m{proc}(c, w, c.k \semi P) \;\mapsto\; \m{proc}(c', w, P[c'/c]), \m{msg}(c, 0, c.k \semi c \leftarrow c') & \mbox{($c'$ fresh)} \\
({\oplus}C) & \m{msg}(c, w', c.k \semi c \leftarrow c'), \m{proc}(d, w, \m{case}\;c\;(\ell \Rightarrow Q_\ell)_{\ell \in L})
\;\mapsto\; \m{proc}(d, w+w', Q_k[c'/c]) \\[1ex]
({\with}S) & \m{proc}(d, w, c.k \semi Q) \;\mapsto\; \m{msg}(c', 0, c.k \semi c' \leftarrow c), \m{proc}(d, w, Q[c'/c]) & \mbox{($c'$ fresh)} \\
({\with}C) & \m{proc}(c, w, \m{case}\;c\;(\ell \Rightarrow Q_\ell)_{\ell \in L}), \m{msg}(c', w', c.k \semi c' \leftarrow c)
\;\mapsto\; \m{proc}(c', w+w', Q_k[c'/c]) \\[1ex]
({\one}S) & \m{proc}(c, w, \m{close}\; c) \;\mapsto\; \m{msg}(c, w, \m{close}\; c) \\
({\one}C) & \m{msg}(c, w', \m{close}\; c), \m{proc}(d, w, \m{wait}\; c \semi Q) \;\mapsto\; \m{proc}(d, w+w', Q)\\[1ex]
({\tensor}S) & \m{proc}(c, w, \m{send}\; c\; d \semi P) \;\mapsto\; \m{proc}(c', w, P[c'/c]), \m{msg}(c, 0, \m{send}\; c\; d \semi c \leftarrow c')
& \mbox{($c'$ fresh)} \\
({\tensor}C) & \m{msg}(c, w', \m{send}\; c\; d \semi c \leftarrow c'), \m{proc}(e, w, x \leftarrow \m{recv}\; c \semi Q)
\;\mapsto\; \m{proc}(e, w+w', Q[c', d/c, x]) \\[1ex]
({\lolli}S) & \m{proc}(e, w, \m{send}\; c\; d \semi Q) \;\mapsto\; \m{msg}(c', 0, \m{send}\; c\; d \semi c' \leftarrow c), \m{proc}(e, w, Q[c'/c]) & \mbox{($c'$ fresh)} \\
({\lolli}C) & \m{proc}(c, w, x \leftarrow \m{recv}\; x \semi P), \m{msg}(c', w', \m{send}\; c\; d \semi c' \leftarrow c)
\;\mapsto\; \m{proc}(c', w+w', P[c', d/c, x]) \\[1ex]
({\exists}S) & \proc{c}{w, \esendn{c}{e} \semi P} \;\mapsto\; \m{proc}(c', w, P[c'/c]), \; \msg{c}{0, \esendn{c}{e} \semi \fwd{c}{c'}} & \mbox{($c'$ fresh)} \\
({\exists}C) & \msg{c}{w', \esendn{c}{e} \semi \fwd{c}{c'}}, \; \m{proc}(d, w, \erecvn{c}{n} \semi Q) \;\mapsto\; \proc{d}{w+w', Q[e/n][c'/c]} \\[1ex]
({\forall}S) & \proc{d}{w, \esendn{c}{e} \semi P} \;\mapsto\; \m{msg}(c', 0, \esendn{c}{e} \semi \fwd{c'}{c}), \; \proc{d}{w, P[c'/c]} & \mbox{($c'$ fresh)} \\
({\forall}C) & \proc{d}{w, \erecvn{c}{n} \semi Q}, \m{msg}(c', w', \esendn{c}{e} \semi \fwd{c'}{c}) \; \mapsto\; \proc{d}{w+w', Q[e/n][c'/c]} \\[1ex]
({\tassertop}S) & \proc{c}{w, \eassert{c}{\phi} \semi P} \;\mapsto\; \m{proc}(c', w, P[c'/c]), \; \msg{c}{0, \eassert{c}{\phi} \semi \fwd{c}{c'}} & \mbox{($c'$ fresh)} \\
({\tassertop}C) & \msg{c}{w', \eassert{c}{\phi_1} \semi \fwd{c}{c'}}, \; \m{proc}(d, w, \eassume{c}{\phi_2} \semi Q) \;\mapsto\; \proc{d}{w+w', Q[c'/c]} \\[1ex]
({\tassumeop}S) & \proc{d}{w, \eassert{c}{\phi} \semi P} \;\mapsto\; \m{msg}(c', 0, \eassert{c}{\phi} \semi \fwd{c'}{c}), \; \proc{d}{w, P[c'/c]} & \mbox{($c'$ fresh)} \\
({\tassumeop}C) & \proc{d}{w, \eassume{c}{\phi_1} \semi Q}, \; \m{msg}(c', w', \eassert{c}{\phi_2} \semi \fwd{c'}{c}) \;\mapsto\; \proc{d}{w+w', Q[c'/c]} \\[1ex]
({\paypot}S) & \proc{c}{w, \epay{c}{r} \semi P} \;\mapsto\; \proc{c'}{w, P[c'/c]}, \; \msg{c}{0, \epay{c}{r} \semi \fwd{c}{c'}} & \mbox{($c'$ fresh)} \\
({\paypot}C) & \msg{c}{w', \epay{c}{r} \semi \fwd{c}{c'}}, \; \m{proc}(d, w, \eget{c}{r} \semi Q) \;\mapsto\; \proc{d}{w+w', Q[c'/c]} \\[1ex]
({\getpot}S) & \proc{d}{w, \epay{c}{r} \semi P} \;\mapsto\; \m{msg}(c', 0, \epay{c}{r} \semi \fwd{c'}{c}), \; \proc{d}{w, P[c'/c]} & \mbox{($c'$ fresh)} \\
({\getpot}C) & \proc{c}{w, \eget{c}{r} \semi Q}, \; \m{msg}(c', w', \epay{c}{r} \semi \fwd{c'}{c}) \;\mapsto\; \proc{c'}{w+w', Q[c'/c]} \\[1ex]
({\m{work}}) & \proc{c}{w, \ework{r} \semi P} \; \mapsto \; \proc{c}{w+r, P}
\end{array}
\]
\caption{Basic and Refined Operational Semantics}
\label{fig:opsem}
\end{figure*}

\clearpage

\section{Forcing Calculus}
\label{app:forcing}
We define the \emph{forcing calculus}, with the following judgment
$\vars \semi \cons \semi \D^- \semi \W \entailpot{q} P :: (x : A)$. 
Here, $\cons$ denotes the arithmetic constraints that are currently
valid, $\D^-$ stores the negative propositions, while $\W$ is
an ordered context storing the remaining propositions. This calculus
applies the $\tassumeop R, \tassertop L, \getpot R, \paypot L$ rules
eagerly, and the $\tassertop R, \tassumeop L, \paypot R, \getpot L$
rules lazily. Later, we will establish
that this forcing calculus is equivalent to the implicit type system.
First, we define the grammar of types.
\[
  \begin{array}{lcl}
    A^+ & ::= & S \mid \tassert{\phi} A^+ \mid \tpaypot{A^+}{r} \\
    A^- & ::= & S \mid \tassume{\phi} A^- \mid \tgetpot{A^-}{r} \\
    A & ::= & A^+ \mid A^- \\
    S & ::= & \ichoice{\ell : A}_{\ell \in L}
    \mid \echoice{\ell : A}_{\ell \in L}
    \mid A \tensor A \mid A \lolli A \\
    & & \mid \texists{n} A \mid \tforall{n} A
  \end{array}
\]
Here, $S$ denotes a structural type. Next, we define the rules for
eagerly applying the $\tassumeop R, \tassertop L, \getpot R,
\paypot L$ rules.
\begin{mathpar}
  \infer[\tassumeop R]
  {\vars \semi \cons \semi \D^- \semi \W \entailpot{q}
  P :: (x : \tassume{\phi} A^-)}
  {\vars \semi \cons \land \phi \semi \D^- \semi \W
  \entailpot{q} P :: (x : A^-)}
  \and
  \infer[\getpot R]
  {\vars \semi \cons \semi \D^- \semi \W \entailpot{q}
  P :: (x : \tgetpot{A^-}{r})}
  {\vars \semi \cons \semi \D^- \semi \W
  \entailpot{q+r} P :: (x : A^-)}
  \and
  \infer[\tassertop L]
  {\vars \semi \cons \semi \D^- \semi \W \cdot
  (x : \tassert{\phi} A^+) \entailpot{q} P :: (z : C^+)}
  {\vars \semi \cons \land \phi \semi \D^- \semi
  \W \cdot (x : A^+) \entailpot{q} P :: (z : C^+)}
  \and
  \infer[\paypot L]
  {\vars \semi \cons \semi \D^- \semi \W \cdot
  (x : \tpaypot{A^+}{r}) \entailpot{q} P :: (z : C^+)}
  {\vars \semi \cons \semi \D^- \semi
  \W \cdot (x : A^+) \entailpot{q+r} P :: (z : C^+)}
\end{mathpar}
If a negative type is encountered in the ordered context,
it is considered as stable and moved to $\D^-$.
\begin{mathpar}
  \infer[\m{move}]
  {\vars \semi \cons \semi \D^- \semi \W \cdot
  (x : A^-) \entailpot{q} P :: (z : C^+)}
  {\vars \semi \cons \semi \D^-, (x : A^-) \semi
  \W \entailpot{q} P :: (z : C^+)}
\end{mathpar}
The invertible rules are applied until we reach a stable sequent
defined using the normal form $\vars \semi \cons \semi \D^- \semi
\cdot \entailpot{q} P :: (x : A^+)$.
At this point, we case analyze on the program $P$ to decide which channel to force.
If $P$ is a forward, we first force right,
\begin{mathpar}
  \infer[\m{id}{-}F_R]
  {\vars \semi \cons \semi (y : B^-) \semi \cdot
  \entailpot{q} \fwd{x}{y} :: (x : A^+)}
  {\vars \semi \cons \semi (y : B^-) \semi \cdot
  \entailpot{q} \fwd{x}{y} :: \focus{x : A^+}}
\end{mathpar}
then force left.
\begin{mathpar}
  \infer[\m{id}{-}F_L]
  {\vars \semi \cons \semi (y : B^-) \semi \cdot
  \entailpot{q} \fwd{x}{y} :: \focus{x : S}}
  {\vars \semi \cons \semi \focus{y : B^-} \semi
  \cdot \entailpot{q} \fwd{x}{y} :: (x : S)}
\end{mathpar}
While forced, we apply the lazy rules $\tassertop R, \tassumeop L,
\paypot R, \getpot L$.
\begin{mathpar}
  \infer[\tassertop R]
  {\vars \semi \cons \semi \D^- \semi \cdot \entailpot{q} P :: \focus{x : \tassert{\phi} A^+}}
  {\vars \semi \cons \proves \phi \and
  \vars \semi \cons \semi \D^- \semi \cdot \entailpot{q} P :: \focus{x : A^+}}
  \and
  \infer[\paypot R]
  {\vars \semi \cons \semi \D^- \semi \cdot \entailpot{q} P :: \focus{x : \tpaypot{A^+}{r}}}
  {\vars \semi \cons \proves q \geq r &
  \vars \semi \cons \semi \D^- \semi \cdot \entailpot{q-r} P :: \focus{x : A^+}}
  \and
  \infer[\tassumeop L]
  {\vars \semi \cons \semi \D^-, \focus{x : \tassume{\phi} A^-} \semi \cdot \entailpot{q} P :: (z : C^+)}
  {\vars \semi \cons \proves \phi \and
  \vars \semi \cons \semi \D^-, \focus{x : A^-} \semi \cdot \entailpot{q} P :: (z : C^+)}
  \and
  \infer[\getpot L]
  {\vars \semi \cons \semi \D^-, \focus{x : \tgetpot{A^-}{r}} \semi \cdot \entailpot{q} P :: (z : C^+)}
  {\vars {\semi} \cons \proves q {\geq} r &
  \vars \semi \cons \semi \D^-, \focus{x : A^-} \semi \cdot \entailpot{q-r} P :: (z : C^+)}
\end{mathpar}
Finally, the forward is typechecked after the forcing.
\begin{mathpar}
  \infer[\m{id}]
  {\vars \semi \cons \semi \focus{y : S} \semi \cdot \entailpot{q} \fwd{x}{y} :: (x : S)}
  {\vars \semi \cons \proves q = 0}
\end{mathpar}
On encountering a spawn, we force each channel used in the context
in order. If a forced channel becomes structural, we lose force on that
channel, and force the next. These are rules $\m{spawn}{-}F_n$ and
$\m{spawn}{-}F_i$ in Figure~\ref{fig:spawn-forcing}.
\begin{figure*}
  \begin{mathpar}
    \inferrule*[Right=$\m{spawn}{-}F_n$]
    {\vars \semi \cons \semi \D^-, \focus{y_n : A^-} \semi \cdot \entailpot{q}
    (\ecut{x}{\mathcal{X} \indv{e_c}}{y_1 \ldots y_n}{Q_x}) :: (z : C^+)}
    {\vars \semi \cons \semi \D^-, (y_n : A^-) \semi \cdot \entailpot{q}
    (\ecut{x}{\mathcal{X} \indv{e_c}}{y_1 \ldots y_n}{Q_x}) :: (z : C^+)}
    \and
    \inferrule*[Right=$\m{spawn}{-}F_i$]
    {\vars \semi \cons \semi \D^-, \focus{y_i : A^-}, (y_{i+1} : S) \semi \cdot \entailpot{q}
    (\ecut{x}{\mathcal{X} \indv{e_c}}{y_1 \ldots y_n}{Q_x}) :: (z : C^+)}
    {\vars \semi \cons \semi \D^-, (y_i : A^-), \focus{y_{i+1} : S} \semi \cdot \entailpot{q}
    (\ecut{x}{\mathcal{X} \indv{e_c}}{y_1 \ldots y_n}{Q_x}) :: (z : C^+)}
    \and
    \inferrule*[Right=$\m{spawn}$]
    {
      q = p[\overline{e_c}/\overline{v}] \and
      \vars \semi \cons \proves r \geq q \and
      \mathcal{X} \indv{v} : \overline{y_i' : A_i} \entailpot{p} x' : A \in \Sg \\\\
      \vars \semi \cons \proves \overline{y_i' : B_i} \equiv_{\alpha} \D_1^-, y_1 : S \and
      \Sg \semi \vars \semi \cons \semi \D_2^- \semi (x :
      A[\overline{e_c} / \overline{v}]) \entailpot{r - q} Q_x :: (z : C)
    }
    {\Sg \semi \vars \semi \cons \semi \D_1^-, \focus{y_1 : S}, \D_2^- \semi \cdot \entailpot{r}
    (\ecut{x}{\mathcal{X} \indv{e_c}}{y_1 \ldots y_n}{Q_x}) :: (z : C^+)}
  \end{mathpar}
  \caption{Spawn rules in forcing calculus}
  \label{fig:spawn-forcing}
\end{figure*}
Once the first channel is forced, we typecheck the continuation,
which could possibly be an unstable sequent, as described in $\m{spawn}$
rule in Figure~\ref{fig:spawn-forcing}.
For the structural types, we force the channel that is being communicated
on.
\begin{mathpar}
  \infer[\oplus F_R]
  {\vars \semi \cons \semi \D^- \semi \cdot \entailpot{q} (\esendl{x}{k} \semi P) ::
  (x : A^+)}
  {\vars \semi \cons \semi \D^- \semi \cdot \entailpot{q} (\esendl{x}{k} \semi P) ::
  \focus{x : A^+}}
  \and
  \infer[\oplus F_L]
  {\vars \semi \cons \semi \D^-, (x : A^-) \semi \cdot \entailpot{q}
  \ecase{x}{\ell}{Q_\ell} :: (z : C^+)}
  {\vars \semi \cons \semi \D^-, \focus{x : A^-} \semi \cdot \entailpot{q}
  \ecase{x}{\ell}{Q_\ell} :: (z : C^+)}
  \and
  \infer[\with F_R]
  {\vars \semi \cons \semi \D^- \semi \cdot \entailpot{q} \ecase{x}{\ell}{P_\ell}_{\ell \in L} ::
  (x : A^+)}
  {\vars \semi \cons \semi \D^- \semi \cdot \entailpot{q} \ecase{x}{\ell}{P_\ell}_{\ell \in L} ::
  \focus{x : A^+}}
  \and
  \infer[\with F_L]
  {\vars \semi \cons \semi \D^-, (x : A^-) \semi \cdot \entailpot{q}
  \esendl{x}{k} \semi Q :: (z : C^+)}
  {\vars \semi \cons \semi \D^-, \focus{x : A^-} \semi \cdot \entailpot{q}
  \esendl{x}{k} \semi Q :: (z : C)}
\end{mathpar}
For the tensor and lolli rules, we force both the channel
involved in the send while sending, and while receiving,
we place the new channel in the ordered context.
\begin{mathpar}
  \infer[\tensor F^1_R]
  {\vars \semi \cons \semi \D^-, (w : A^-) \semi \cdot \entailpot{q}
  \esendch{x}{w} \semi P :: (x : B^+)}
  {\vars \semi \cons \semi \D^-, \focus{w : A^-} \semi \cdot \entailpot{q}
  \esendch{x}{w} \semi P :: (x : B^+)}
  \and
  \infer[\tensor F^2_R]
  {\vars \semi \cons \semi \D^-, \focus{w : S} \semi \cdot \entailpot{q}
  \esendch{x}{w} \semi P :: (x : A^+)}
  {\vars \semi \cons \semi \D^-, (w : S) \semi \cdot \entailpot{q}
  \esendch{x}{w} \semi P :: \focus{x : A^+}}
  \and
  \infer[\tensor F_L]
  {\vars \semi \cons \semi \D^-, (x : A^-) \semi \cdot \entailpot{q} \erecvch{x}{y} \semi Q_y :: (z : C^+)}
  {\vars \semi \cons \semi \D^-, \focus{x : A^-} \semi \cdot \entailpot{q} \erecvch{x}{y} \semi Q_y :: (z : C^+)}
  \and
  \infer[\lolli F_R]
  {\vars \semi \cons \semi \D^- \semi \cdot \entailpot{q}
  \erecvch{x}{y} \semi P_y :: (x : A^+)}
  {\vars \semi \cons \semi \D^- \semi \cdot \entailpot{q}
  \erecvch{x}{y} \semi P_y :: \focus{x : A^+}}
  \and
  \infer[\lolli F^1_L]
  {\vars \semi \cons \semi \D^-, (w : A^-) \semi \cdot \entailpot{q} \esendch{x}{w} \semi Q :: (z : C^+)}
  {\vars \semi \cons \semi \D^-, \focus{w : A^-} \semi \cdot \entailpot{q} \esendch{x}{w} \semi Q :: (z : C^+)}
  \and
  \inferrule*[right=$\lolli F^2_L$]
  {\vars \semi \cons \semi \D^-, (w : S), \focus{x : A^-} \semi \cdot \entailpot{q} \hspace{6em} \\ \hspace{6em} \esendch{x}{w} \semi Q :: (z : C^+)}
  {\vars \semi \cons \semi \D^-, \focus{w : S}, (x : A^-) \semi \cdot \entailpot{q} \hspace{6em} \\ \hspace{6em} \esendch{x}{w} \semi Q :: (z : C^+)}
\end{mathpar}
The rules for $\one$ are analogous.
\begin{mathpar}
  \infer[\one F_R]
  {\vars \semi \cons \semi \cdot \semi \cdot \entailpot{q}
  \eclose{x} :: (x : A^+)}
  {\vars \semi \cons \semi \cdot \semi \cdot \entailpot{q}
  \eclose{x} :: \focus{x : A^+}}
  \and
  \infer[\one F_L]
  {\vars \semi \cons \semi \D^-, (x : A^-) \semi \cdot \entailpot{q} \ewait{x} \semi Q :: (z : C^+)}
  {\vars \semi \cons \semi \D^-, \focus{x : A^-} \semi \cdot \entailpot{q} \ewait{x} \semi Q :: (z : C^+)}
\end{mathpar}
Since the $\forall$ and $\exists$ operators are structural,
they follow the same structure.
\begin{mathpar}
  \infer[\exists F_R]
  {\vars \semi \cons \semi \D^- \semi \cdot \entailpot{q} \esendn{x}{t} \semi P :: (x : A^+)}
  {\vars \semi \cons \semi \D^- \semi \cdot \entailpot{q} \esendn{x}{t} \semi P :: \focus{x : A^+}}
  \and
  \infer[\exists F_L]
  {\vars \semi \cons \semi \D^-, (x : A^-) \semi \cdot \entailpot{q} \erecvn{x}{k} \semi Q :: (z : C^+)}
  {\vars \semi \cons \semi \D^-, \focus{x : A^-} \semi \cdot \entailpot{q} \erecvn{x}{k} \semi Q :: (z : C^+)}
\end{mathpar}
\begin{mathpar}
  \infer[\forall F_R]
  {\vars \semi \cons \semi \D^- \semi \cdot \entailpot{q} \erecvn{x}{k} \semi P :: (x : A^+)}
  {\vars \semi \cons \semi \D^- \semi \cdot \entailpot{q} \erecvn{x}{k} \semi P :: \focus{x : A^+}}
  \and
  \infer[\forall F_L]
  {\vars \semi \cons \semi \D^-, (x : A^-) \semi \cdot \entailpot{q} \esendn{x}{t} \semi Q :: (z : C^+)}
  {\vars \semi \cons \semi \D^-, \focus{x : A^-} \semi \cdot \entailpot{q} \esendn{x}{t} \semi Q :: (z : C^+)}
\end{mathpar}
Finally, we provide the strutural rules of forcing calculus
in Figure~\ref{fig:forcing-structural}.
\begin{figure*}
  \begin{mathpar}
    \infer[\oplus R_k]
    {\vars \semi \cons \semi \D^- \semi \cdot \entailpot{q} (\esendl{x}{k} \semi P) ::
    \focus{x : \ichoice{\ell : A_\ell}_{\ell \in L}}}
    {\vars \semi \cons \semi \D^- \semi \cdot \entailpot{q} P :: (x : A_k)}
    \and
    \infer[\oplus L]
    {\vars \semi \cons \semi \D^-, \focus{x : \ichoice{\ell : A_\ell}_{\ell \in L}} \semi \cdot \entailpot{q}
    \ecase{x}{\ell}{Q_\ell}_{\ell \in L} :: (z : C^+)}
    {\vars \semi \cons \semi \D^- \semi (x : A_\ell) \entailpot{q} Q_\ell :: (z : C^+) \quad (\forall \ell \in L)}
  \end{mathpar}

  \begin{mathpar}
    \infer[\with R]
    {\vars \semi \cons \semi \D^- \semi \cdot \entailpot{q} \ecase{x}{\ell}{P_\ell}_{\ell \in L} ::
    \focus{x : \echoice{\ell : A_\ell}_{\ell \in L}}}
    {\vars \semi \cons \semi \D^- \semi \cdot \entailpot{q} P_\ell ::
    (x : A_\ell)}
    \and
    \infer[\with L_k]
    {\vars \semi \cons \semi \D^-, \focus{x : \echoice{\ell : A_\ell}_{\ell \in L}} \entailpot{q}
    \esendl{x}{k} \semi Q :: (z : C^+)}
    {\vars \semi \cons \semi \D^- \semi (x : A_k) \entailpot{q}
    Q :: (z : C^+)}
  \end{mathpar}

  \begin{mathpar}
    \infer[\tensor R]
    {\vars \semi \cons \semi \D^-, (w : S) \semi \cdot \entailpot{q}
    \esendch{x}{w} \semi P :: \focus{x : S \tensor B}}
    {\vars \semi \cons \semi \D^- \semi \cdot \entailpot{q}
    P :: (x : B)}
    \and
    \infer[\tensor L]
    {\vars \semi \cons \semi \D^-, \focus{x : S \tensor B} \semi \cdot \entailpot{q} \erecvch{x}{y} \semi Q_y :: (z : C^+)}
    {\vars \semi \cons \semi \D^- \semi (x : B) \cdot (y : S) \entailpot{q} Q_y :: (z : C^+)}
  \end{mathpar}

  \begin{mathpar}
    \infer[\lolli R]
    {\vars \semi \cons \semi \D^- \semi \cdot \entailpot{q}
    \erecvch{x}{y} \semi P_y :: \focus{x : S \lolli B}}
    {\vars \semi \cons \semi \D^- \semi (y : S) \entailpot{q}
    P_y :: (x : B)}
    \and
    \infer[\lolli L]
    {\vars \semi \cons \semi \D^-, (w : S), \focus{x : S \lolli B} \semi \cdot \entailpot{q} \esendch{x}{w} \semi Q :: (z : C^+)}
    {\vars \semi \cons \semi \D^- \semi (x : B) \entailpot{q} Q :: (z : C^+)}
  \end{mathpar}

  \begin{mathpar}
    \infer[\exists R]
    {\vars \semi \cons \semi \D^- \semi \cdot \entailpot{q} \esendn{x}{e} \semi P :: \focus{x : \texists{n} A}}
    {\vars \semi \cons \vdash e : \m{nat} \and
    \vars \semi \cons \semi \D^- \semi \cdot \entailpot{q} P :: (x : A[e/n])}
    \and
    \infer[\exists L]
    {\vars \semi \cons \semi \D^-, \focus{x : \texists{n} A} \semi \cdot \entailpot{q} \erecvn{x}{n} \semi Q_n :: (z : C^+)}
    {\vars,n \semi \cons \semi \D^- \semi (x : A) \entailpot{q} Q_n :: (z : C^+)}
    \and
    \infer[\forall R]
    {\vars \semi \cons \semi \D^- \semi \cdot \entailpot{q} \erecvn{x}{n} \semi P_k :: \focus{x : \tforall{n} A}}
    {\vars,n \semi \cons \semi \D^- \semi \cdot \entailpot{q} P_n :: (x : A)}
    \and
    \infer[\forall L]
    {\vars \semi \cons \semi \D^-, \focus{x : \tforall{n} A} \semi \cdot \entailpot{q} \esendn{x}{e} \semi Q :: (z : C^+)}
    {\vars \semi \cons \vdash e : \m{nat} \and
    \vars \semi \cons \semi \D \semi (x : A[e/n]) \entailpot{q} Q :: (z : C^+)}
  \end{mathpar}

  \begin{mathpar}
    \infer[\one R]
    {\vars \semi \cons \semi \cdot \semi \cdot \entailpot{q} \eclose{x} :: \focus{x : \one}}
    {\vars \semi \cons \proves q = 0}
    \and
    \infer[\one L]
    {\vars \semi \cons \semi \D^-, \focus{x : \one} \semi \cdot \entailpot{q} \ewait{x} \semi Q :: (z : C^+)}
    {\vars \semi \cons \semi \D^- \semi \cdot \entailpot{q} Q :: (z : C^+)}
  \end{mathpar}
  \caption{Structural rules in forcing calculus}
  \label{fig:forcing-structural}
\end{figure*}

\begin{lemma}[Invertibility]\label{lem:invertible}
  The rules $\tassumeop R$, $\tassertop L, \getpot R, \paypot L$
  and $\m{move}$ are invertible.
\end{lemma}

\begin{definition}
  We define the 4 notations below.
  \[
  \begin{array}{rclcrcl}
    \phiR{\tassume{\phi} A} & = & \phi \land \phiR{A} & &
    \phiR{A^+} & = & \true \\
    [0.3em]
    \tpR{\tassume{\phi} A} & = & \tpR{A} & &
    \tpR{A^+} & = & A^+ \\
    [0.3em]
    \phiL{\tassert{\phi} A} & = & \phi \land \phiL{A} & &
    \phiL{A^-} & = & A^- \\
    [0.3em]
    \tpL{\tassert{\phi} A} & = & \tpL{A} & &
    \tpL{A^-} & = & A^-
  \end{array}
  \]
\end{definition}

\begin{lemma}[Move Left]\label{lem:move-left-ctx}
  If $\vars \semi \cons \semi \D^- \semi \W \cdot
  (x : A) \entailpot{q} P :: (z : C^+)$, then
  $\vars \semi \cons \land \phiL{A}
  \semi \D^-, (x : \tpL{A}) \semi \W \entailpot{q} P :: (z : C^+)$.
\end{lemma}

\begin{proof}
  If $\tpL{A} = S$ and $\phiL{A} = \phi$, then
  $\vars \semi \cons \land \phi \semi \D^-, (x : S) \semi \W
  \entailpot{q} P :: (z : C^+)$ by inverting $\tassumeop L$ or $\m{move}$.
\end{proof}

\begin{lemma}[Stable Sequent]\label{lem:stable-full}
  If $\vars \semi \cons \semi \D_1^- \semi \D_2 \entailpot{q} P
  :: (x : A)$, then $\vars \semi \cons \land \phiL{\D_2} \land \phiR{A}
  \semi \D_1^-, \tpL{\D_2} \semi \cdot \entailpot{q} P :: (x : \tpR{A})$.
\end{lemma}

\begin{proof}
  Successive application of Lemma~\ref{lem:move-left-ctx}
  after inverting $\tassertop R$.
\end{proof}

\begin{lemma}[Lazy $\tassertop R$]\label{lem:lazyR}
  If $\vars \semi \cons \semi \D^- \semi \W \entailpot{q} P :: (x : A^+)$ and
  $\vars \semi \cons \proves \phi$, then $\vars \semi \cons \semi \D^- \semi \W
  \entailpot{q} P :: (x : \tassert{\phi} A^+)$.
\end{lemma}

\begin{proof}
  The proof proceeds by induction on the typing judgment. Since we
  do not allow quantifier alternation, the $\tassumeop R$ rule
  cannot be applied in either case. The $\tassertop L$ rule can
  apply in either case. The $\m{move}$ rule applies in either case.
  The $\m{id}R$ rule can be applied in either case. The $\m{id}L$
  cannot be applied in either case since it requires the offered
  channel to be forced. The same holds for $\tassertop R$. The
  $\tassumeop L$ rule can be applied in either case. If the last
  rule applied is $\m{id}$, we first apply $\m{id}R$, then $\m{id}$.
  All the spawn rules still apply in either case. The same holds
  for all the forcing rules. Finally, the right structural rules
  do not apply in either case, and the left rules are allowed
  in either case.
\end{proof}

\begin{lemma}[Lazy $\tassumeop L$]\label{lem:lazyL}
  If $\vars \semi \cons \semi \D^- \semi \W \cdot (x : A^-) \entailpot{q}
  P :: (z : C^+)$ and $\vars \semi \cons \proves \phi$, then
  $\vars \semi \cons \semi \D^- \semi \W \cdot (x : \tassume{\phi} A^-)
  \entailpot{q} P :: (z : C^+)$.
\end{lemma}

\begin{theorem}[Soundness of QRecon]
  If $\vars \semi \cons \semi \cdot \semi \D \entailpot{q} P :: (x : A)$,
  then $\vars \semi \cons \semi \D \ivdash{q} P :: (x : A)$.
\end{theorem}

\begin{proof}
  Soundness is trivial, since every rule in forcing calculus is valid in
  the implicit type system.
\end{proof}

\begin{theorem}[Completeness of QRecon]
  If $\vars \semi \cons \semi \D \ivdash{q} P :: (x : A)$,
  then $\vars \semi \cons \semi \cdot \semi \D \entailpot{q} P :: (x : A)$.
\end{theorem}

\begin{proof}
  The proof proceeds by induction on the typing judgment of the implicit
  type system. We case analyze on the last rule applied. First, we consider
  the structural cases.
  \begin{itemize}[leftmargin=*]
    \item Case ($\oplus R_k$) : $P = \esendl{x}{k} \semi P$ and
    $A = \ichoice{\ell : A_\ell}_{\ell \in L}$.
    \begin{mathpar}
      \infer[\oplus R_k]
      {\vars \semi \cons \semi \D \ivdash{q} \esendl{x}{k} \semi P :: (x : \ichoice{\ell : A_\ell}_{\ell \in L})}
      {\vars \semi \cons \semi \D \ivdash{q} P :: (x : A_k)}
    \end{mathpar}
    By the induction hypothesis, $\vars \semi \cons \semi \cdot \semi \D \entailpot{q} P :: (x : A_k)$.
    Using Lemma~\ref{lem:stable-full}, we get $\vars \semi \cons \land \phiL{\D} \land \phiR{A_k}
    \semi \tpL{\D} \semi \cdot \entailpot{q} P :: (x : \tpR{A_k})$. Since this is a stable
    sequent, we reapply the $\oplus R_k$ rule, we get (note that $\tpR{A_k} = A_k$
    and $\phiR{A_k} = \true$)
    \begin{mathpar}
      \infer[\oplus R_k]
      {\vars \semi \cons \land \phiL{\D} \semi \tpL{\D} \semi \cdot \entailpot{q}
      \esendl{x}{k} \semi P :: \focus{x : \ichoice{\ell : A_\ell}}}
      {\vars \semi \cons \land \phiL{\D} \semi \tpL{\D} \semi \cdot \entailpot{q}
      P :: (x : A_k)}
    \end{mathpar}
    Now, applying the $\oplus F_R$ rule,
    \begin{mathpar}
      \infer[\oplus F_R]
      {\vars \semi \cons \land \phiL{\D} \semi \tpL{\D} \semi \cdot \entailpot{q}
      \esendl{x}{k} \semi P :: (x : \ichoice{\ell : A_\ell})}
      {\vars \semi \cons \land \phiL{\D} \semi \tpL{\D} \semi \cdot \entailpot{q}
      \esendl{x}{k} \semi P :: \focus{x : \ichoice{\ell : A_\ell}}}
    \end{mathpar}
    Reapplying the $\m{move}$ and $\tassumeop L$ rules, we get
    $\vars \semi \cons \semi \cdot \semi \D \entailpot{q} \esendl{x}{k} \semi P ::
    (x : \ichoice{\ell : A_\ell}_{\ell \in L})$.
  
    \item Case ($\oplus L$) : $P = \ecase{x}{\ell}{Q_\ell}_{\ell \in L}$ and
    $\D = \D, (x : \ichoice{\ell : A_\ell}_{\ell \in L})$.
    \begin{mathpar}
      \inferrule*[right=$\oplus L$]
      {\vars \semi \cons \semi \D, (x : A_\ell) \ivdash{q} Q_\ell :: (z : C)}
      {\vars \semi \cons \semi \D, (x : \ichoice{\ell : A_\ell}_{\ell \in L}) \ivdash{q} \hspace{8em}\\
      \hspace{8em}\ecase{x}{\ell}{Q_\ell}_{\ell \in L} :: (z : C)}
    \end{mathpar}
    By the induction hypothesis, $\vars \semi \cons \semi \cdot \semi (x : A_\ell) \cdot \D
    \entailpot{q} Q_\ell :: (z : C)$. First, we use Lemma~\ref{lem:invertible} to get
    $\vars \semi \cons \land \phiR{C} \semi \cdot \semi (x : A_\ell) \cdot \D \entailpot{q}
    Q_\ell :: (z : \tpR{C})$. Then, we use Lemma~\ref{lem:move-left-ctx}
    successively on $\D$ to get $\vars \semi \cons \land \phiR{C} \land \phiL{\D} \semi
    \tpL{\D} \semi (x : A_\ell) \entailpot{q} Q_\ell :: (z : \tpR{C})$. Now, we apply
    the $\oplus L$ rule to get
    \begin{mathpar}
      \inferrule*[right=$\oplus L$]
      {\vars {\semi} \cons \land \phiR{C} \land \phiL{\D} {\semi}
      \tpL{\D} \semi (x : A_\ell) \entailpot{q} Q_\ell {::} (z {:} \tpR{C})}
      {\vars \semi \cons \land \phiR{C} \land \phiL{\D} \semi
      \tpL{\D}, \focus{x : \ichoice{\ell : A_\ell}_{\ell \in L}} \semi \cdot
      \entailpot{q} \\ \hspace{5em} \ecase{x}{\ell}{Q_\ell}_{\ell \in L} :: (z : \tpR{C})}
    \end{mathpar}
    Applying the $\oplus F_L$ rule, we get
    $\vars \semi \cons \land \phiR{C} \land \phiL{\D} \semi
    \tpL{\D},$ $(x : \ichoice{\ell : A_\ell}_{\ell \in L}) \semi \cdot
    \entailpot{q} \ecase{x}{\ell}{Q_\ell}_{\ell \in L} :: (z : \tpR{C})$.
    Reapplying the invertible rules $\tassumeop R$ on $C$ and
    $\m{move}$ and $\tassertop L$ on $\D, (x : \ichoice{\ell : A_\ell}_{\ell \in L})$,
    we get
    $\vars \semi \cons \semi \cdot \semi \D, (x : \ichoice{\ell : A_\ell}_{\ell \in L})
    \entailpot{q} \ecase{x}{\ell}{Q_\ell}_{\ell \in L} :: (z : C)$.

    \item Case ($\with R, \with L_k$) : Analogous to $\oplus$ proofs.
    
    \item Case ($\lolli R$) : $P = \erecvch{x}{y} \semi P_y$ and
    $A = S \lolli B$.
    \begin{mathpar}
      \infer[\lolli R]
      {\vars \semi \cons \semi \D \ivdash{q} \erecvch{x}{y} \semi P_y
      :: (x : S \lolli B)}
      {\vars \semi \cons \semi \D, (y : S) \ivdash{q} P_y :: (x : B)}
    \end{mathpar}
    By the induction hypothesis, $\vars \semi \cons \semi \cdot \semi (y : S)
    \cdot \D \entailpot{q} P_y :: (x : B)$. Using Lemma~\ref{lem:stable-full},
    we get $\vars \semi \cons \land \phiR{B} \land \phiL{\D} \semi \tpL{\D} \semi (y : S)
    \entailpot{q} P_y :: (x : \tpR{B})$. Note that since $B$ is a positive type,
    $\phiR{B} = \true$ and $\tpR{B} = B$. Thus, we have
    $\vars \semi \cons \land \phiL{\D} \semi \tpL{\D} \semi (y : S)
    \entailpot{q} P_y :: (x : B)$. Now, apply the $\lolli R$ rule to get
    \begin{mathpar}
      \inferrule*[right=$\lolli R$]
      {\vars \semi \cons \land \phiL{\D} \semi \tpL{\D} \semi (y : S)
      \entailpot{q} P_y :: (x : B)}
      {\vars \semi \cons \land \phiL{\D} \semi \tpL{\D} \semi \cdot
      \entailpot{q} \hspace{8em} \\ \hspace{8em} \erecvch{x}{y} \semi P_y :: \focus{x : S \lolli B}}
    \end{mathpar}
    Now, we drop force using the $\lolli F_R$ rule to get
    $\vars \semi \cons \land \phiL{\D} \semi \tpL{\D} \semi \cdot
    \entailpot{q} \erecvch{x}{y} \semi P_y :: (x : S \lolli B)$.
    Finally, we apply the $\tassertop L, \m{move}$ rules to get
    $\vars \semi \cons \semi \cdot \semi \D \entailpot{q} \erecvch{x}{y} \semi P_y
    :: (x : S \lolli B)$.

    \item Case ($\lolli L$) : $P = \esendch{x}{w} \semi Q$ and
    $\D = \D, (w : S), (x : S \lolli B)$.
    \begin{mathpar}
      \inferrule*[right=$\lolli L$]
      {\vars \semi \cons \semi \D, (x : B) \ivdash{q} Q :: (z : C)}
      {\vars \semi \cons \semi \D, (w : S), (x : S \lolli B) \ivdash{q} \hspace{4em} \\ \hspace{8em}
      \esendch{x}{w} \semi Q :: (z : C)}
    \end{mathpar}
    By the induction hypothesis, we get $\vars \semi \cons \semi \cdot
    \semi (x : B) \cdot \D \entailpot{q} Q :: (z : C)$. First, we use
    Lemma~\ref{lem:invertible} to get $\vars \semi \cons \land \phiR{C}
    \semi \cdot \semi (x : B) \cdot \D \entailpot{q} Q :: (z : \tpR{C})$.
    Then, we use Lemma~\ref{lem:move-left-ctx} on $\D$ to get
    $\vars \semi \cons \land \phiR{C} \land \phiL{\D} \semi \tpL{\D} \semi
    (x : B) \entailpot{q} Q :: (z : \tpR{C})$. Now, we apply the $\lolli L$
    rule
    \begin{mathpar}
      \inferrule*[right=${\lolli}L$]
      {\vars {\semi} \cons \land \phiR{C} \land \phiL{\D} {\semi} \tpL{\D} \semi
      (x : B) \entailpot{q} Q {::} (z {:} \tpR{C})}
      {\vars \semi \cons \land \phiR{C} \land \phiL{\D} \semi \tpL{\D},
      (w : S), \focus{x : S \lolli B} \entailpot{q} \\ \hspace{8em} \esendch{x}{w} \semi Q
      :: (z : \tpR{C})}
    \end{mathpar}
    Applying the $\lolli F^2_L$ and $\lolli F^1_L$ rules in order,
    we get
    $\vars \semi \cons \land \phiR{C} \land \phiL{\D} \semi \tpL{\D},
    (w : S), (x : S \lolli B) \entailpot{q} \esendch{x}{w} \semi Q
    :: (z : \tpR{C})$. Now, reapplying the invertible rules, we get
    $\vars \semi \cons \semi \cdot \semi \D, (w : S), (x : S \lolli B) \entailpot{q}
    \esendch{x}{w} \semi Q :: (z : C)$.

    \item Case ($\tensor R, \tensor L, \one R, \one L$) : Analogous
    to $\lolli$ rules.

    \item Case ($\m{id}$) : $P = \fwd{x}{y}$.
    \begin{mathpar}
      \infer[\m{id}]
      {\cons \semi (y : A) \ivdash{q} \fwd{x}{y} :: (x : A)}
      { }
    \end{mathpar}
    We proceed by induction on the structure of $A$. If $A$ is structural,
    we use the following derivation in the forcing calculus.
    \begin{mathpar}
      \infer[\m{move}]
      {\vars \semi \cons \semi \cdot \semi (y : S) \entailpot{q} \fwd{x}{y} :: (x : S)}
      {
        \infer[\m{id}R]
        {\vars \semi \cons \semi (y : S) \semi \cdot \entailpot{q} \fwd{x}{y} :: (x : S)}
        {
          \infer[\m{id}L]
          {\vars \semi \cons \semi (y : S) \semi \cdot \entailpot{q} \fwd{x}{y} :: \focus{x : S}}
          {
            \infer[\m{id}]
            {\vars \semi \cons \semi \focus{y : S} \semi \cdot \entailpot{q} \fwd{x}{y} :: (x : S)}
            {}
          }
        }
      }
    \end{mathpar}
    If $A = \tassert{\phi} A'$, by the induction hypothesis, we know
    that $\vars \semi \cons \semi \cdot \semi (y : A') \entailpot{q} \fwd{x}{y} :: (x : A')$.
    In the forcing calculus, we apply the derivation
    \begin{mathpar}
      \infer[\tassertop L]
      {\vars \semi \cons \semi \cdot \semi (y : \tassert{\phi} A') \entailpot{q}
      \fwd{x}{y} :: (x : \tassert{\phi} A')}
      {
        \infer[\m{id}R]
        {\vars \semi \cons \semi \cdot \semi (y : A') \entailpot{q}
        \fwd{x}{y} :: (x : \tassert{\phi} A')}
        {
          \infer[\tassertop R]
          {\vars \semi \cons \semi \cdot \semi (y : A') \entailpot{q}
          \fwd{x}{y} :: \focus{x : \tassert{\phi} A'}}
          {
            \infer[\m{lose}R]
            {\vars \semi \cons \semi \cdot \semi (y : A') \entailpot{q}
            \fwd{x}{y} :: \focus{x : A'}}
            {\vars \semi \cons \semi \cdot \semi (y : A') \entailpot{q}
            \fwd{x}{y} :: (x : A')}
          }
        }
      }
    \end{mathpar}
    A similar argument would work for $A = \tassume{\phi} A'$.

    \item Case ($\m{spawn}$) : Since all types involved in a spawn
    are structural, this case is trivial.

    \item Case ($\tassertop R$) : $A = \tassert{\phi} A$.
    \begin{mathpar}
      \infer[\tassertop R]
      {\vars \semi \cons \semi \D \ivdash{q} P :: (x : \tassert{\phi} A)}
      {\vars \semi \cons \proves \phi \and
      \vars \semi \cons \semi \D \ivdash{q} P :: (x : A)}
    \end{mathpar}
    By the induction hypothesis, we get $\vars \semi \cons \semi \cdot \semi
    \D \entailpot{q} P :: (x : A)$. We use Lemma~\ref{lem:lazyR} to get
    $\vars \semi \cons \semi \cdot \semi \D \entailpot{q} P :: (x : \tassert{\phi} A)$
    since we know that $\vars \semi \cons \proves \phi$.

    \item Case ($\tassertop L$) : $\D = \D, (x : \tassert{\phi} A)$
    and $P = Q$.
    \begin{mathpar}
      \infer[\tassertop L]
      {\vars \semi \cons \semi \D, (x : \tassert{\phi} A) \ivdash{q} Q :: (z : C)}
      {\vars \semi \cons \land \phi \semi \D, (x : A) \ivdash{q} Q :: (z : C)}
    \end{mathpar}
    By the induction hypothesis, we get that $\vars \semi \cons \land \phi \semi \cdot
    \semi \D \cdot (x : A) \entailpot{q} Q :: (z : C)$. First, we invert the
    $\tassumeop R$ rule to get
    $\vars \semi \cons \land \phi \land \phiR{C} \semi \cdot \semi \D \cdot (x : A)
    \entailpot{q} Q :: (z : \tpR{C})$. Then, we apply the $\tassertop L$ rule
    to get $\vars \semi \cons \land \phiR{C} \semi \cdot \semi \D \cdot (x : \tassert{\phi} A)
    \entailpot{q} Q :: (z : \tpR{C})$. Then we reapply the $\tassumeop R$ to
    get back $\vars \semi \cons \semi \cdot
    \semi \D \cdot (x : \tassert{\phi} A) \entailpot{q} Q :: (z : C)$.

  \end{itemize}

\end{proof}

\end{document}